\documentclass[11.5pt]{article}
\usepackage{graphicx} 
\usepackage{amsmath}
\usepackage{braket}
\usepackage{graphicx}
\usepackage{amssymb}
\usepackage{subcaption}
\usepackage{float}
\usepackage[T1]{fontenc}
\usepackage[utf8]{inputenc}

\usepackage{mwe}
\usepackage{amsthm}
\usepackage{mathtools}
\usepackage[ruled,vlined]{algorithm2e}
\usepackage{tikz}
\usepackage{wrapfig}
\usepackage{booktabs}    
\usepackage{array}       
\usepackage{caption}     
\usepackage{sectsty}
\usepackage{ifthen}

\newboolean{showfigures}
\setboolean{showfigures}{true}

\sectionfont{\centering\normalsize}
\subsectionfont{\centering\small}
\subsubsectionfont{\centering\footnotesize}

\usepackage{tocloft}

\usepackage{hyperref}

\title{\textit{Noise-Robust Self-Testing: Detecting Non-Locality in Noisy Non-Local Inputs}} 
\author{By\\Romi Lifshitz\\\\ Supervisor\\Prof. Thomas Vidick}
\date{April 2025}
\usepackage{hyperref}
\hypersetup{
    colorlinks=true,
    linkcolor=blue,
    filecolor=magenta,      
    urlcolor=blue,
    pdftitle={Overleaf Example},
    pdfpagemode=FullScreen,
    citecolor=blue,
    }
\usepackage{setspace}
\numberwithin {claim}{section}

\newtheorem{definition}{Definition}

\newtheorem{question}{Question}

\newtheorem{theorem}{Theorem}

\newtheorem{lemma}{Lemma}

\newtheorem{observation}{Observation}

\usepackage[left=1.5cm, right=1.5cm, top=1.5cm, bottom=1.5cm]{geometry}

\usepackage{nomencl}
\makenomenclature
\usepackage{multicol} 
\usepackage{pdfpages}

\begin{document}
\onehalfspacing

\includepdf[pages=-]{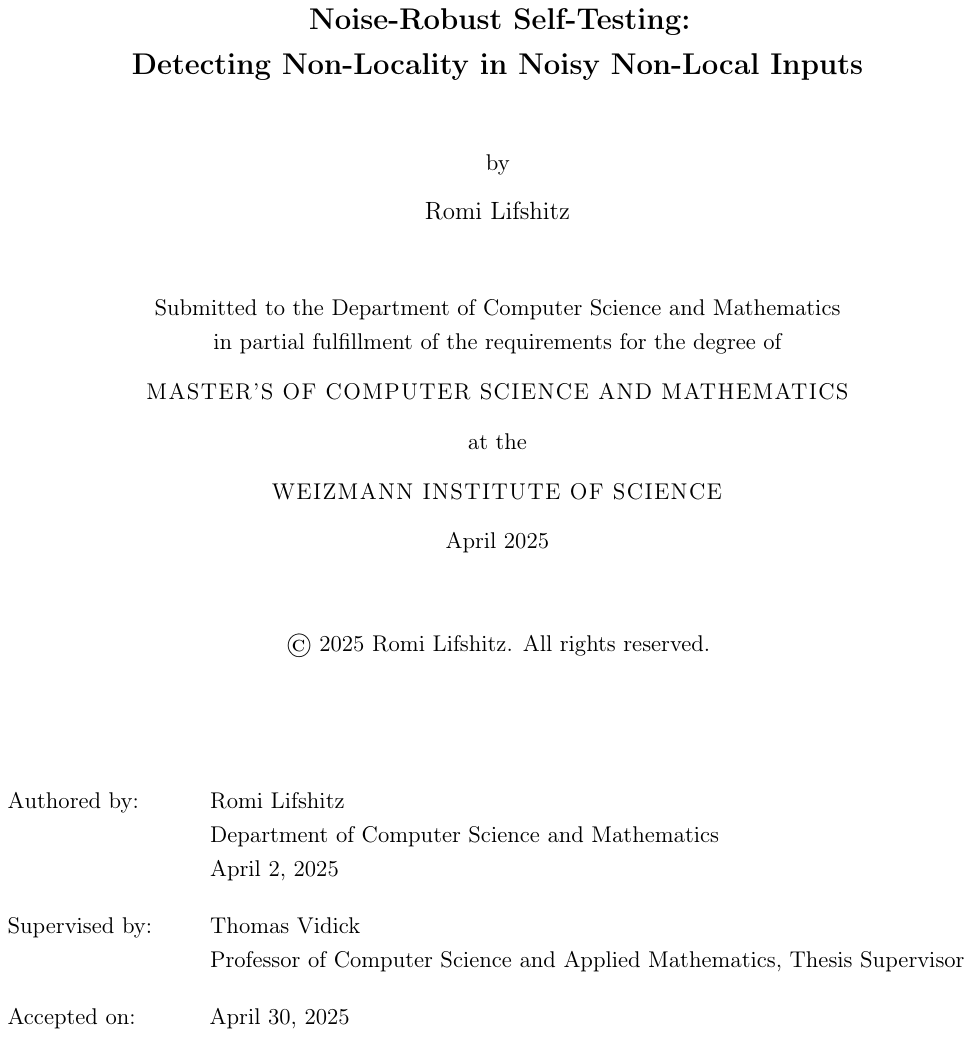}

\newpage
\tableofcontents
\listoffigures
\listoftables

\section*{\Large List of Abbreviations}
\printnomenclature[1cm]

\begin{multicols}{2}
\begin{itemize}
  \item CHSH Game -- Clauser, Horne, Shimony, \& Holt Game
  \item MSG -- Magic Square Game
  \item W -- Werner State
  \item DI -- Device-Independent
  \item GHZ -- Greenberger-Horne-Zeilinger
  \item EPR -- Einstein-Podolsky-Rosen
  \item $\mathfrak{C}_G$ -- Convincingness of Game $G$
  \item $\kappa_G$ -- Gapped Score of Game $G$
  \item $\mathcal{H}$ -- Hilbert Space
  \item $\mathbb{I}$ -- Identity operator
  \item POVM -- Positive Operator-Valued Measure
  \item LHV (or LHVM) -- Local Hidden Variable (Model)
  \item CPTP Map -- Completely Positive \& Trace Preserving Map 
  \item XOR -- Exclusively-OR
  \item 2-CHSH Game -- 2-Clauser, Horne, Shimony, \& Holt Game
  \item Def(s). -- Definition(s)
  \item Fig(s). -- Figure(s)
  \item Sec(s). -- Section(s)
  \item 2-CHSH-OPT -- Optimized 2-CHSH Game
  \item LP -- Linear Program(ming)
  \item Const. -- Constant
  \item Coeff. -- Coefficient
  \item IFF -- if \& only if
\end{itemize}
\end{multicols}

\newpage
\section*{ABSTRACT}
Non-local games test for non-locality and entanglement in quantum systems and are used in self-tests for certifying quantum states in untrusted devices. However, these protocols are tailored to ideal states, so realistic noise prevents maximal violations and leaves many partially non-local states undetected. Selecting self-tests based on their `robustness' to noise can tailor protocols to specific applications, but current literature lacks a standardized measure of noise-robustness. Creating such a measure is challenging as there is no operational measure for comparing tests of different dimensionalities and input-output settings. We propose and study three comparative measures: \textit{noise-tolerance}, \textit{convincingness}, and an analytic approximation of convincingness called the \textit{gapped score}. Our computational experiments and analytic framework demonstrate that convincingness provides the most nuanced measure for noise-robustness. We then show that the CHSH game has the highest noise-robustness compared to more complex games (2-CHSH variants and the Magic Square Game) when given equal resources, while with unequal resources, some 2-CHSH variants can outperform CHSH at a high resource cost. This work provides the first systematic and operational framework for comparing noise-robustness in self-testing protocols,  laying a foundation for theoretical advances in understanding noise-robustness of self-tests and practical improvements in quantum resource utilization.

\section{INTRODUCTION}\label{sec:intro}

\subsection{Motivation}
\noindent \textit{Bell tests}, fundamental to quantum computing and communication, demonstrate \textit{non-locality} (and hence, \textit{entanglement}) in physical systems by showing violations of inequalities that cannot be explained by classical, local hidden-variable theories \cite{bell1964einstein}. These tests have evolved into protocols called \textit{non-local games}, which are crucial for detecting entanglement in quantum resources \cite{hensen2015loophole, shalm2015strong, giustina2015significant}.
\textit{Self-tests}, a specific class of non-local games, enable device-independent (DI) certification of non-locality in quantum states and measurements \cite{vsupic2020self}. In these tests, a referee interacts with spatially isolated players who have prepared some physical
system in a device. Using a self-test, the referee can determine the presence of non-locality in the device based solely on its input-output statistics and the assumption that quantum mechanics is true. Self-tests have proven valuable in DI Quantum Key Distribution \cite{vazirani2019fully, arnon2019simple, yin2017satellite}, randomness certification \cite{pironio2010random}, and experimental demonstrations of entanglement \cite{zhang2018experimental, wu2021robust, hu2023self}. Moreover, they certify select quantum properties  more efficiently than tomography \cite{goh2019experimental}.

However, the DI property of self-tests comes at a cost: self-tests are optimized for specific pure entangled states. For instance, the  Clauser, Horne, Shimony, Holt (CHSH) self-test \cite{clauser1969proposed} achieves maximal quantum violation only with a maximally entangled state, the EPR pair up to local isometries. This presents practical challenges, as physical imperfections and noise in real experiments make achieving maximal violations impossible with current technology. Consequently, mixed states with partial non-locality may go undetected, leading to the waste of potentially useful quantum resources. Moreover, no established methodology exists to rank different self-tests based on their ability to detect non-locality in noisy states. This prevents practitioners from selecting tests with appropriate sensitivity levels for specific applications--whether they need robust detection of noisy entanglement (as in entanglement distillation) or stringent verification of more secure states (as in quantum key distribution).

While previous works have explored methods for detecting non-locality in noisy systems, they either rely on device-dependent approaches—such as weak measurements \cite{singh2018analysing}, local noise-reversal channels \cite{pawela2013enhancing}, or optimized entanglement witnesses \cite{aolita2015open}—or require extensive resources, like the DI parallel-repetition-based approach proposed in \cite{arnon2017noise}. Other studies have characterized lower bounds on the violations achieved for
specific input states  (i.e., GHZ, W, graph states) affected by noisy quantum channels \cite{chaves2013detecting, divianszky2016bounding, wiesniak2021symmetrized, laskowski2014highly, sohbi2015decoherence, gawron2008noise}.
However, these works interchangeably use the terms noise-robustness, -tolerance, or -resistance of a game to describe its \textit{ability to detect non-locality in noisy systems}, with unclear and un-unified definitions. Each approach is uniquely tied to the specific noise models, games, and experimental setups in a work. This fragmentation reveals a critical gap. The field lacks: (1) a unified definition and measure of the \textit{noise-robustness of self-tests}, and (2) a systematic framework for comparing different tests' abilities to detect non-locality in noisy systems.
This thesis addresses these limitations by developing rigorous methods for defining, computing, and comparing the noise-robustness of self-tests, particularly towards noisy states rather than noisy measurements. In doing so, we provide tools for experimentalists and theorists to select a self-test which is more suitable to the level of robustness needed in their application, noise model, and resource constraints.
We are, hence, motivated to ask the following research questions to help define the noise-robustness of self-tests:
\begin{question}\label{question:1}
    Given a set of self-tests, which test is able to detect non-locality in the largest regime of noisy states, under the assumption of a particular noise model? 
\end{question}

\begin{question}\label{question:2}
    Given a fixed set of noisy resources (under a noise model) and a set of candidate self-tests, which self-test should one use to produce the largest deviation from local deterministic statistics, knowing that the statistics will not be ideal? In other words, how should we use a noisy entanglement resource to demonstrate non-locality most effectively?
\end{question}

\begin{question}\label{question:3}
    Is it possible to expand the set of noisy states whose non-locality is detected by a self-test?
\end{question}

Questions \ref{question:1} and \ref{question:2} capture different perspectives on noise-robustness—either focusing on the regime of noisy states a game can certify or how to best leverage a limited, imperfect resource supply—while \ref{question:3} targets the ultimate aim of devising more noise-robust games. The challenge in answering these questions is that, due to the distinct setup and norms of games, there is no straightforward way to compare their: (1) sets of violating states, and (2) violations and scores. 
For example, consider comparing the CHSH game, whose ideal state is a 2-qubit EPR pair, versus the Magic Square Game (MSG), whose ideal state is two EPR pairs (4 qubits). The games operate optimally on different dimensionalities, their scores exist on different scales, and a violation of the local bound ($\omega_c$) in one game cannot be directly mapped to a meaningful violation in the other. While a ratio of the score attained from an experiment to the local bound has been used as a comparative measure \cite{perez2008unbounded, palazuelos2016survey}, it is inadequate as it lacks operational interpretation and normalization between games. 

From a geometric perspective, consider how different games relate to the quantum and local sets \cite{cirel1980quantum, popescu1994quantum, brunner2014bell}. Each non-local game defines a distinct axis through these sets, creating different directional cuts through the space of correlations. This geometric view highlights fundamental questions: How are points compared along different axes? How do we characterize and compare the sets of states violating different games? These questions remain poorly understood, highlighting the lack of theoretical tools for systematically comparing non-local games and understanding their relative capabilities.

We note that throughout the thesis, we use the terms non-local games and self-tests interchangeably (unless stated otherwise), as we are only concerned with non-local games which are self-tests. 

\subsection{Contribution and Main Ideas}
We answer the above questions by proposing novel methods for comparing the abilities of non-local games to detect non-locality in noisy states (Sec. \ref{sec:def-tools-noiserobustness}). Using these methods, we construct a rigorous definition for noise-robustness (Secs. \ref{sec:def-tools-noiserobustness}, \ref{sec:analytic-methods}), and a framework for testing and optimizing the noise-robustness of non-local games (Sec. \ref{sec:noise-rob-framework}). Further, we apply the new methods computationally and analytically (Sec. \ref{sec:comp-methods}, \ref{sec:analytic-methods}, \ref{sec:results-disc}) to compare the noise-robustness abilities of existing, and new non-local games in finite and infinite (noisy) resource regimes. Specifically, we construct the following measures of noise-robustness:
\begin{enumerate}
    \item \textbf{Noise-tolerance of} $G$: noise-tolerance characterizes a non-local game's ability to detect non-locality in noisy versions of its winning state, where a game is considered $(\eta,\varepsilon)$-tolerant if it can detect non-locality in its winning state under visibility (probability of no noise) $\eta$ in a noise model $\varepsilon$. To compare noise-tolerance between different games, we examine whether they can detect non-locality in their respective winning states under the same visibility $\eta$, with appropriate extensions of the noise channel for games of different dimensionalities (i.e., comparing CHSH's tolerance of a noisy EPR pair versus MSG's tolerance of two noisy EPR pairs).
    \item \textbf{Convincingness of} $G$ ($\mathfrak{C}_{G}$): 
    Convincingness is a statistical measure that quantifies how likely a game's observed statistics could have occurred due to locality, using a $p$-value calculation\footnote{During the course of this thesis, we observed that \cite{araujo2020bell} employs an identical $p$-value formulation, focused on optimizing Bell‑inequality violations for single‑shot rejection of local‑hidden‑variable models. In contrast, we derive the $p$-value independently to define noise‑robustness and develop a comprehensive framework for comparing and ranking the noise‑robustness of non‑local games with differing structures, physical noise channels, and resource constraints.}. A higher convincingness (lower $p$-value) indicates stronger evidence of non-locality in the input state, allowing comparison of different games' effectiveness at detecting non-locality in noisy states on a normalized scale regardless of their setup or bounds.
    
    \item \textbf{Gapped Score ($\kappa_G$)}: The gapped score is an efficient approximation of the convincingness measure for ranking games' effectiveness at detecting non-locality in noisy states, without evaluating $\mathfrak{C}_G$ on input states. While convincingness offers a complete analysis, our gapped score provides a more analytic, simpler calculation based on a game's sensitivity to a noise model and its quantum-classical separation. The score reliably ranks games when moderate to high numbers of noisy resources are available for proving non-locality. See Def. \ref{def:quad-score-dep} for the formal definition.
\end{enumerate}
We refer the reader to Sec. \ref{sec:def-tools-noiserobustness} for more detail. Equipped with these new definitions, we utilize them to guide designing games with improved noise-robustness. Lastly, after investigating the behaviour of these noise-robustness measures, we arrive at the proposed definition of noise-robustness in Def. \ref{def:noise-rob-1}: \textit{A non-local game $G_1$ is more noise-robust to a noisy channel than $G_2$ if $G_1$ is significantly convincing for higher noise levels on the input state than $G_2$.} This definition can be further simplified to Def. \ref{def:gapped-score-def-of-noise-rob} if a moderate to high resource supply is available for games to prove non-locality: \textit{A non-local game $G_1$ is more noise-robust to a noisy channel than $G_2$ if $\kappa_{G_1} > \kappa_{G_2}$.} These definitions can provide practical guidance, based in theory, to experimentalists who wish to identify entanglement in imperfect hardware, while preserving their resources. Further, as the latter definition depends on an analytically computable expression ($\kappa_G$), it is suitable for theoretical studies of noise-robustness in asymptotic resource regimes. To our knowledge, we are the first to build a framework for evaluating and studying the noise-robustness of non-local games. Moreover, we are the first to use the $p$-value for comparing between different non-local games' robustness to noisy inputs. 
 
\section{PRELIMINARIES}
\subsection{Quantum Information}
Let us define relevant quantum information concepts which we use throughout the work. For more information, see \cite{Nielsen_Chuang_2010}.

\begin{definition}[Hilbert Space]
A $d$-dimensional Hilbert space $\mathcal{H} = \mathbb{C}^d$ is a complete inner product space over $\mathbb{C}$.
\end{definition}

\begin{definition}[Bra-ket]
For a vector $|\psi\rangle \in \mathcal{H}$, its dual vector $\langle\psi| = (|\psi\rangle)^\dagger \in \mathcal{H}^*$, where $\dagger$ is the conjugate transpose.
\end{definition}

\begin{definition}[Qubit and $n$-Qubit System]
A qubit is a quantum state in $\mathcal{H}=\mathbb{C}^2$. The computational basis states are denoted $\{\ket{0}, \ket{1}\}$. An $n$ qubit system is a quantum state in $\mathcal{H} = (\mathbb{C}^2)^{\otimes n} \cong \mathbb{C}^{2^n}$.
\end{definition}

\begin{definition}[Quantum State]
A quantum state of a system with Hilbert space $\mathcal{H}$ is described by a density operator $\rho$ acting on $\mathcal{H}$. The state is:
\begin{itemize}
   \item Pure if and only if $\rho^2 = \rho$ (equivalently, $\text{Tr}(\rho^2) = 1$)
   \item Mixed if and only if $\rho^2 \neq \rho$ (equivalently, $\text{Tr}(\rho^2) < 1$)
\end{itemize}
A pure state can alternatively be represented by a unit vector $|\psi\rangle \in \mathcal{H}$, with corresponding density operator $\rho = |\psi\rangle\langle\psi|$.
\end{definition}

\begin{definition}[Separable \& Entangled States]
    The state $\rho_{AB}$ is separable (not entangled) if and only if it can be written as $\rho_{AB} = \sum_i p_i \rho_i^A \otimes \rho_i^B$, where $p_i$ are probabilities ($\sum_i p_i = 1$) and $\rho_i^A$, $\rho_i^B$ are density matrices for systems $A$ and $B$. $\rho_{AB}$ is entangled if and only if it is not separable.
\end{definition}

\begin{definition}[Bipartite System]
A bipartite quantum system is a composite system with $\mathcal{H} = \mathcal{H}_A \otimes \mathcal{H}_B$, where $\mathcal{H}_A$ and $\mathcal{H}_B$ are the Hilbert spaces of subsystems $A$ and $B$, respectively.
\end{definition}

\begin{definition}[Maximally Mixed State]
The maximally mixed state for a quantum system $\rho \in \mathbb{C}^{d \times d}$ is $\rho = \frac{\mathbb{I}_d}{d}$.
\end{definition}

\begin{definition}[Maximally Entangled State]
     A bipartite system, $\rho_{AB}$, with subsystems $A, B$ in dimension $d$, is maximally entangled if and only if $\rho_{AB}$ is pure and the reduced densitiy matrices, $\rho_A = Tr_B(\rho_{AB}),\rho_B = Tr_A(\rho_{AB}) $, are maximally mixed. The EPR pair is the maximally entangled state, $\ket{\psi} = \frac{\ket{00} + \ket{11}}{\sqrt{2}}$, up to local isometries.
\end{definition}

\begin{definition}[POVM]
A Positive Operator-Valued Measure (POVM) is a set $\{E_i\}$ where $E_i \succeq 0$ (positive semidefinite), $\sum_i E_i = I$ (completeness). For input state $\rho$, the probability of measuring outcome $i$ is $\text{Tr}(E_i\rho)$. 
\end{definition}

\begin{definition}[Projective Measurement]
A projective measurement is a POVM $\{P_i\}$ where $P_i^2 = P_i$ and $P_iP_j = \delta_{ij}P_i$.
\end{definition}

\begin{definition}[Quantum Channel] A quantum channel is a completely positive and trace preserving (CPTP) map $T$ that maps density operators on $\mathcal{H}_A$ to density operators on $\mathcal{H}_B$.
\end{definition}

\begin{lemma}[Tensor Product of Quantum Channels]
If $T_1$ and $T_2$ are CPTP maps, then $T_1 \otimes T_2$ is also a CPTP map.
\end{lemma}

\begin{definition}[2-Qubit Depolarizing Channel]\label{def:2qbit-dep-channel}
The two-qubit depolarizing channel with visibility $\eta \in [0,1]$ is the CPTP map $\varepsilon_{\eta}^{2-qubit}(\rho) = \eta^2 \rho + (1-\eta^2)\frac{\mathbb{I}_4}{4}$, where $\rho$ is a two-qubit state (i.e., $\rho \in \mathbb{C}^{4\times4}$).
\end{definition}

\begin{definition}[4-Qubit Depolarizing Channel]\label{def:4qbit-dep-channel}
    The four-qubit depolarizing channel with visibility $\eta\in [0,1]$ is the CPTP map $\varepsilon_{\eta}^{4-qubit}(\rho) = \varepsilon_{\eta}^{2-qubit}(\rho) \otimes \varepsilon_{\eta}^{2-qubit}(\rho)$ where each factor $\varepsilon_{\eta}^{2-qubit}$ acts independently on a two-qubit subsystem.
\end{definition}
\noindent Note that when the input state is known, we drop the $\rho$ specification in the channels.

\subsection{Probability \& Statistics}
We define some probability and statistics terms, which we later use to define noise-robustness. For more details, see \cite{casella2002statistical}.
\begin{definition}[Probability Distribution]
A probability distribution is a function $P$ mapping a sample space $\Omega \rightarrow [0,1]$ such that: (1) For a discrete random variable, $X$: $\sum_{\omega \in \Omega} P(\omega) = 1$ and $P(X = x)$ gives probability of specific values, and (2) For a continuous random variable, $X$: $\int_{\Omega} f(x)dx = 1$ and $P(a \leq X \leq b) = \int_a^b f(x)dx$ gives probability over intervals.
\end{definition}

\begin{definition}[Hypothesis Test]
A hypothesis test evaluates competing probability distributions for observed data. Let $\theta$ represent the true, unknown parameter value that determines the probability distribution of our data. The test compares:
(1) The null hypothesis $H_0: \theta = \theta_0$, which specifies a single probability distribution for the data, and (2) The alternative hypothesis $H_1: \theta = \theta_1$ (or $\theta > \theta_0$, or $\theta < \theta_0$), which specifies a single (or set of) distribution(s). The test uses a test statistic $T$ computed from the data to determine which distribution better explains the observations.
\end{definition}

\begin{definition}[$p$-value (upper-tailed)]\label{def:pvalue}
Given a hypothesis test comparing  $H_0: \theta = \theta_0 \quad$ with $\quad H_1: \theta > \theta_0$, the $p$-value (upper-tailed) is the probability, computed under the distribution specified by $\theta_0$, of obtaining a test statistic value at least as extreme as the observed value $t$. Formally, $p = P(T \geq t \mid H_0) = \sum_{j=t}^n P(T = j \mid H_0)$, where the summation is over all outcomes that are as extreme or more extreme than $t$ in the direction specified by $H_1$.
\end{definition}

\subsection{Non-Local Games}\label{sec:NLGs}
Non-local games are interactive protocols (games) between multiple parties, used to detect non-locality in quantum devices through violations of Bell inequalities. In order to understand them, we require the following information \cite{scarani2019bell,watrous2020nonlocal}:

\begin{definition}[Local Hidden Variable Model]
A local hidden variable model (LHVM) for a bipartite system describes correlations between measurement outcomes $a,b$ for settings $x,y$ as probability distributions of the form:
$P(a,b|x,y) = \sum_\lambda p_\lambda P(a|x,\lambda)P(b|y,\lambda)$
where $\lambda$ represents hidden variables with probability distribution $p_\lambda$, and the probabilities $P(a|x,\lambda)$, $P(b|y,\lambda)$ depend only on local measurements and the hidden variable.
\end{definition}

\begin{definition}[Bell Inequality]
A Bell inequality is a linear constraint on probability distributions $P(a,b|x,y)$ of measurement outcomes $a,b$ for settings $x,y$ that must be satisfied by any local hidden variable model.
\end{definition}

\begin{theorem}[Non-locality implies Entanglement]
If a quantum state $\rho_{AB}$ exhibits non-locality by violating a Bell inequality, with local measurements, then $\rho_{AB}$ must be entangled (violation is sufficient, but not necessary for entanglement).
\end{theorem}

With this in mind, let us now define the setup, strategies, and scoring of 2-player games, which may be extended to $k$-player games (for a more detailed explanation, see \cite{watrous2020nonlocal, aaronson2018quantum}). Intuitively, in a 2-player non-local game, a referee, Charlie, sends questions $x, y \in \mathcal{X}, \mathcal{Y}$ to each player, Alice and Bob. Each player then returns an output $a, b \in \mathcal{A}, \mathcal{B}$, respectively, to Charlie, who then checks a relationship (i.e., an inequality) defined by a predicate, $V$. The  players \textit{win} if the inequality is satisfied, and \textit{lose} otherwise. Games with more players can be defined analogously, with an additional set of questions and answers per player, and an updated predicate. Formally:

\begin{definition}[2-Player Non-Local Game]
    A 2-player non-local \textit{game} can be defined generally, as a 6-tuple $G =(\mathcal{X},\mathcal{Y},\mathcal{A},\mathcal{B},\pi,V)$ where:
\begin{itemize}
    \item The finite non-empty sets of $\mathcal{X} = (x_1, x_2, \dots, x
_k),\mathcal{Y} = (y_1, y_2, \dots, y
_k)$ are questions and the finite non-empty sets of $\mathcal{A} = (a_1, a_2, \dots, a
_k),$ $\mathcal{B} = (b_1, b_2, \dots, b
_k)$, are answers, for arbitrary $k$ 
    \item The sample space $\Omega = \mathcal{X} \times \mathcal{Y} \times \mathcal{A} \times \mathcal{B}$ consists of all possible combinations of questions and answers
    \item $\pi \in \mathcal{P}(\mathcal{X} \times \mathcal{Y})$ is a probability vector denoting the probability of the referee selecting $(x,y) \in \mathcal{X} \times \mathcal{Y}$
    \item $V : \Omega \rightarrow \{0,1\}$ is a measurable function (predicate) determining winning outcomes,where $V(a,b|x,y)$ denotes if $(a,b)$ win (1) or lose (0) given $(x,y)$.
\end{itemize} 
\end{definition}
 
 Importantly, players can coordinate a \textit{strategy} for playing the game. For example, a strategy might be that they each return 1 upon receiving a 0. If the resources used by the players are classical, the strategy is expressed by two functions $f_1,f_2$ such that $f_1 : \mathcal{X} \rightarrow \mathcal{A}$ and $f_2 : \mathcal{Y} \rightarrow \mathcal{B}$. However, the general definition of either a classical or quantum strategy is presented in Def. \ref{def:NLG-strategy}. Moreover, by playing multiple \textit{rounds} of the game with a strategy, a statistical \textit{score} can be computed for it, by the means of \ref{def:NLG-score}. Each game also has a unique maximal score using a classical or quantum strategy (up to local isometries) known as its \textit{local} and \textit{quantum bounds} (Defs. \ref{def:NLG-classicalbound} and \ref{def:NLG-quantumbound}). Finally, the local bound is used to determine whether a strategy's state (input state) has non-locality, using the concept of a \textit{violation} (Def. \ref{def:NLG-violation}). 
\begin{definition}[Strategy]\label{def:NLG-strategy}
    A \textit{strategy} of a 2-player non-local game can be generally defined as the $3$-tuple: $\mathcal{S} = (\rho, \mathcal{M}_1, \mathcal{M}_2)$, where $\rho$ is some (input) state in the Hilbert space $\mathcal{H}$ and $\mathcal{M}_{i\in \{1,2\}}$ is a set of measurement operators (observables of projectors) used by players $P_{i\in \{1,2\}}$ on their share of the subspace, $\mathcal{H}_1$ or $\mathcal{H}_2$, respectively. During the strategy, each player selects a measurement, based on their question, $x$ or $y$, and applies it on their subspace. A classical or quantum strategy is one in which the players use only classical or quantum resources, respectively. 
\end{definition}

\begin{definition}[Score]\label{def:NLG-score}
    For each trial $t \in \Omega$, define the win/loss indicator $W: \Omega \to \{0,1\}$ as $W(x,y,a,b) = V(a,b|x,y)$. The score, or value, of game $G$ is the expectation $\mathbb{E}[W]$, which, for a given strategy $\mathcal{S} = (\rho, \mathcal{M}_1, \mathcal{M}_2)$ can be computed as $\langle \mathbf{S}, \rho \rangle =$ Tr$(\rho \mathbf{S}) \coloneqq \omega(G)$, where $\mathbf{S}$, is a Hermitian operator combining $\mathcal{M}_1$ and $\mathcal{M}_2$. $\mathbf{S}$ should account for correlations corresponding to all input and output combinations contributing nonzero probabilities to the score. 
\end{definition}
 
\begin{definition}[Local Bound]\label{def:NLG-classicalbound}
The maximal score of a game, $G$, achieved with a classical strategy is the best score achieved under the local hidden variables (LHV) model and is called the local bound (or, winning classical value) of $   G$:
\begin{align}
\omega_{c_G} = \max_{f,g} \sum_{(x,y) \in \mathcal{X} \times \mathcal{Y}} \pi(x, y) V(f(x), g(y)|x, y).    
\end{align}
\end{definition}

\begin{definition}[Quantum Bound]\label{def:NLG-quantumbound}
    The maximal score of a game, $G$, achieved with a quantum strategy is the best score achieved under the assumption of quantum mechanics and is called the quantum bound (or, winning quantum value) of $G$. To compute it, let us define each operator in the sets $\mathcal{M}_1, \mathcal{M}_2$ according to POVM elements corresponding to each player's possible questions and answers.  Let $W_{a|x}$ be a POVM element for an outcome $a$, given $x$, acting on $\mathcal{H}_1$, and similarly let $Q_{b|y}$ be a POVM element acting on $\mathcal{H}_2$. Then the quantum bound is:
\begin{align}
\omega_{q_G} = \underset{\mathcal{S}}{\text{sup}}  \sum_{(x,y)\in \mathcal{X} \times \mathcal{Y}} \pi(x,y) \sum_{(a,b) \in \mathcal{A} \times \mathcal{B}} V(a,b|x,y) \langle W_{a|x} \otimes Q_{b|y}, \rho \rangle.    
\end{align}
\end{definition}

\begin{definition}[Violation]\label{def:NLG-violation}
    A \textit{violation} of a non-local game results when some strategy achieves $\omega(G) > \omega_c(G)$, implying that non-locality and, therefore, entanglement has been used in the strategy; that LHV has been violated. If a game has been violated, this is a `win', and otherwise, it is a `loss'. 
\end{definition}

Let us also make a note about the comparability of scores and violations in games. If $\mathbf{S}$ is a projector, the score is a probability and is called the success probability, whereas if  $\mathbf{S}$ is an observable, the score is on a shifted scale. If the game is a \textit{XOR game}, the two scales may be easily interconverted. Otherwise, the interconversion is more challenging. Moreover, due to the different input-output settings of games, measurements, and normalizations $\pi$, the scores and violations are incomparable between different games. Also note that we simplify $\omega_{c_G} = \omega_c, \omega_{q_G} = \omega_q$ when the reference game is clear.

\subsection{Self-Testing \& Common Self-Tests}
In this thesis, we are concerned with non-local games which enforce a \textit{local strategy} and are \textit{device-independent}. Such non-local games are called self-tests, or self-testing protocols, and are used to certify the non-locality of either a state or measurements used by the device during the process of a game \cite{vsupic2020self}. In particular:
\begin{definition}[Self-Test] A self-test is a non-local game defined by the following rules:
\begin{enumerate}
    \item The players are restricted to local strategies during the protocol; however, they may share a resource such as randomness or entanglement, established before the game starts.
    \item The physical systems of the two players comprise a device that is considered to be a black box (\textit{device-independence}). That is, the referee sends questions to and receives answers from the players, making no assumptions about the device's inner workings. The only assumption is that the laws of quantum mechanics hold. Hence, only the statistics of the inputs and outputs of the device are used to discern a property  (i.e., non-locality) of the system inside the box. 
\end{enumerate}
\end{definition}  

\subsubsection{CHSH: An initial non-local game}
Let us consider various commonly played 2-player non-local games. The first is the \textit{Clauser, Horne, Shimony, and Holt (CHSH)} game (Fig. \ref{diag:CHSH}a) \cite{clauser1969proposed}, one of the most commonly used components in self-testing protocols:

\begin{definition}[CHSH Game]
    The CHSH Game is the 2-player, 2-input, 2-output non-local game where the question and answer sets are binary values,  $\mathcal{X} = \mathcal{Y} = \mathcal{A} = \mathcal{B} = \{0,1\}$, and the probability distribution across all possible questions, $\pi$, is uniform such that $\pi(0,0) = \pi(0,1) = \pi(1,0) = \pi(1,1) = \frac{1}{4}$. Further, the predicate for the CHSH game $V$ is:
\begin{equation}
V(a,b|x,y) = \begin{cases}
   1 & \text{if } a \oplus b = x \wedge y \\
   0 & \text{if } a \oplus b \neq x \wedge y.
\end{cases}
\end{equation}
\end{definition}

More intuitively, the CHSH game entails two non-signalling players, Alice ($A$) and Bob ($B$). A classical referee, Charlie (\textit{C}), selects bits $x,y \in \{0,1\}$ uniformly at random, sending $x$ to \textit{A} and $y$ to \textit{B}. The two players then carry out a pre-established strategy for satisfying the CHSH inequality: $a \oplus b = x \wedge y$. The score of the game, the quantum and local bounds, and optimal quantum strategy are as follows:
\begin{definition}[General CHSH Strategy \& Score]\label{def:CHSH-score}
Using projective measurements, where the correlations across settings are denoted by $P(a,b|x,y) = \text{Tr}(\rho(M_{a|x} \otimes M_{b|y}))$, the CHSH score (as a success probability) for any state $\rho \in \mathcal{H}^{AB}$ is:
\begin{align}
\omega_{chsh} = \frac{1}{4}\sum_{x,y,a,b} V_\rho(a,b|x,y)P(a,b|x,y) = \frac{1}{4}\sum_{\substack{x,y,a,b \\ a\oplus b = x\wedge y}} \text{Tr}(\rho(M_{a|x} \otimes M_{b|y})),
\end{align}
where $M_{a|x}$ and $M_{b|y}$ are Alice and Bob's measurements, respectively. Alternatively, the same score is commonly  computed on a shifted scale as $\frac{\text{Tr}(\rho \mathbf{S})}{8} + \frac{1}{2}$  using a Hermitian operator, $\mathbf{S}_{chsh} = A_0 B_0 + A_0 B_1 + A_1 B_0 - A_1 B_1$, with:
\[A_0 = M_{a=0|x=0}-M_{a=1|x=0}, \; A_1 = M_{a=0|x=1}-M_{a=1|x=1}\]
\[B_0 = M_{b=0|y=0}-M_{b=1|y=0},\; B_1=M_{b=0|y=1}-M_{b=1|y=1}.\]
\end{definition}
\begin{theorem}[CHSH Optimal Scores]
    The local and quantum bounds of CHSH are $75\%$ and $85.36\%$, respectively, in terms of success probabilities, and 2 and $2\sqrt{2}$, respectively, on the shifted scale attained with the observable operator \cite{Tsirelson}.
\end{theorem}

\begin{theorem}[CHSH Optimal Quantum Strategy]
    The optimal CHSH strategy can be defined as $\mathcal{S} = (\rho_{EPR}, \mathbf{S})$ where $\rho_{EPR}$ is the EPR pair, up to local isometries 
 and $\mathbf{S}_{chsh}$  with
$A_0 = \sigma_z,  A_1 = \sigma_x, 
B_0 = \frac{\sigma_z + \sigma_x}{\sqrt{2}}, 
\text{ and }B_1 = \frac{\sigma_z - \sigma_x}{\sqrt{2}}$
for Pauli matrices $\sigma_x, \sigma_z$ \cite{scarani2019bell}.
\end{theorem}

\subsubsection{2-CHSH Game}
The 2-CHSH game (Fig. \ref{diag:CHSH}b) is an extension of the CHSH, wherein Alice and Bob play two CHSH games simultaneously with one referee. The referee sends and receives two inputs ($x_1,x_2$) and  outputs ($a_1,a_2$) from Alice and two inputs ($y_1,y_2$) and outputs ($b_1,b_2$) from Bob. The players win IFF both CHSH games win. More formally, the game is defined as follows:
\begin{definition}[2-CHSH Game]The 2-CHSH is a non-local game with two players, each receiving two questions and giving two answers to the referee, corresponding to the binary values, $\mathcal{X}_1 = \mathcal{Y}_1 = \mathcal{X}_2 = \mathcal{Y}_2 = \mathcal{A}_1 = \mathcal{B}_1 = \mathcal{A}_2 = \mathcal{B}_2 = \{0,1\}$, with uniform probability across questions, $\pi$, such that $\pi(x_1,y_1,x_2,y_2) = \frac{1}{16}$.
\noindent The predicate $V_2$ is defined as:
\begin{align}
V_2(a_1,b_1,a_2,b_2|x_1,y_1,x_2,y_2) &= V(a_1,b_1|x_1,y_1) \wedge V(a_2,b_2|x_2,y_2) \\ &= \begin{cases}
   1 & \text{if } (a_1 \oplus b_1 = x_1 \wedge y_1) \wedge (a_2 \oplus b_2 = x_2 \wedge y_2) \\
   0 & \text{otherwise},
\end{cases}
\end{align}
where the game is won if and only if both constituent CHSH games are won simultaneously. 
\end{definition}

\begin{definition}[General 2-CHSH Strategy \& Score] \label{def:2-chsh-strategy-score}
    The score for the 2-CHSH can be computed using projective measurements for convenience, where  $P(a_1,b_1,a_2,b_2|x_1,y_1,x_2,y_2) = \text{Tr}(\rho(M_{a_1|x_1} \otimes M_{a_2|x_2} \otimes M_{b_1|y_1} \otimes M_{b_2|y_2}))$: {\small
\begin{align}
\omega_{2chsh} &= \frac{1}{16}\sum_{x_1,y_1,x_2,y_2,a_1,b_1,a_2,b_2} V_{2}(a_1,b_1,a_2,b_2|x_1,y_1,x_2,y_2)P(a_1,b_1,a_2,b_2|x_1,y_1,x_2,y_2) \\ 
 &= \frac{1}{16}\text{Tr}\biggl(\rho \cdot \sum_{\substack{x_1,y_1,x_2,y_2,a_1,b_1,a_2,b_2 \\ (a_1\oplus b_1 = x_1\wedge y_1)\wedge(a_2\oplus b_2 = x_2\wedge y_2)}}(M_{a_1|x_1} \otimes M_{a_2|x_2} \otimes M_{b_1|y_1} \otimes M_{b_2|y_2})\biggr).
\end{align}} 
\noindent The measurements are defined for Alice as $M_{0|0} = \frac{\mathbb{I} + A_0}{2}, M_{1|0} = \frac{\mathbb{I} - A_0}{2}, M_{0|1} = \frac{\mathbb{I} + A_1}{2}, M_{1|1} = \frac{\mathbb{I} - A_1}{2}$, and for Bob as $
M_{0|0} = \frac{\mathbb{I} + B_0}{2},
M_{1|0} = \frac{\mathbb{I} - B_0}{2},
M_{0|1} = \frac{\mathbb{I} + B_1}{2},
M_{1|1} = \frac{\mathbb{I} - B_1}{2}$.
\end{definition}
\noindent Note that in the code, we do $M_{a_1|x_1} \otimes M_{b_1|y_1} \otimes M_{a_2|x_2} \otimes M_{b_2|y_2}$ due to how the order of the qubits is stored in our code quantum objects. 

\begin{definition}[2-CHSH Optimal Strategy, \& Bounds]
    The optimal strategy for the 2-CHSH is to play the optimal CHSH strategy on each of the two CHSH instances and output a win only if both CHSH games are won (i.e., take the AND of the two outcomes). The quantum bound is $\approx 0.729$, achieved with $\rho_{EPR} \otimes \rho_{EPR}$, and the local bound is $\frac{10}{16}$ \cite{PhysRevA.66.042111}. \
\end{definition}

\subsubsection{The Magic Square Game}

We introduce another common game, the Magic Square Game (MSG). It is a 2-player game (Alice and Bob) with one referee (Charlie). As an overview, in the MSG, Charlie first samples 2 bits $x,y \in \mathcal{X}, \mathcal{Y}$ uniformly at random, and sends them to Alice and Bob, respectively. Next, Alice returns 3 bits, $a_{x_i} \in \{\pm 1\}$, to Charlie, and Bob returns one bit,  $b_{x_y} \in \{\pm 1\}$. Lastly, the referee checks that the bits satisfy the parity and consistency checks in a predicate, $V$. If both conditions are satisfied, accept (win). Otherwise, reject (loss). The game and general strategy for the MSG are formally defined below.

\begin{definition}[Magic Square Game]
The Magic Square Game is the 2-player, $(6,3)$-input, $(3,1)$-output non-local game where the question and answer sets are $\mathcal{X} = \{1,\dots,6\}$, $\mathcal{Y} = \{1,2,3\}$, $\mathcal{A} = \{\pm1\}^3$, $\mathcal{B} = \{\pm1\}$, and the probability distribution across all possible questions, $\pi$, is uniform such that $\pi(x,y) = \frac{1}{18} \quad \forall x \in \mathcal{X}, y \in \mathcal{Y}$. Further, the predicate is:

\[V(a,b|x,y) = \begin{cases}
    1 \text{ (accept)} & \text{if } (a_{x_1}a_{x_2}a_{x_3} = c_x \text{ (parity check)}) \wedge (a_{x_y} = b_{x_y} \text{ (consistency check)})\\
    0 \text{ (reject)} & \text{otherwise}
\end{cases}\]

where $a = (a_{x_1}, a_{x_2}, a_{x_3}) \in \mathcal{A}$ is Alice's output triple, $b \in \mathcal{B}$ is Bob's output bit, and $c_x$ is defined as:
\[c_x = \begin{cases}
    -1 & \text{if } x = 6 \text{ (third column)}\\
    +1 & \text{otherwise}
\end{cases}\]
\end{definition}

\begin{definition}[General MSG Strategy \& Score]
The players possess some state, $\rho$. They define a grid of mutually commuting observables to choose from, where the products of observables in each row and column are as follows: 
\[
\begin{array}{cccc}
O_1 & O_2 & O_3 & | I \\
O_4 & O_5 & O_6 & | I\\
O_7 & O_8 & O_9 & | I\\
-- & -- & -- & \\
I & I & -I
\end{array}    
\]
\noindent To obtain her three bits, Alice measures the three observables, which we denote by  $O_{x_1}, O_{x_2}, O_{x_3}$, in the $x$'th block (row or column) on her share of $\rho$. To obtain his bit, Bob measures the $y$'th observable in block $x$ of the grid, $O_{x_y}^T$, on his share of $\rho$. It is also noteworthy to mention that when considering this, the parity check can be written as $a_{x_1} a_{x_2} a_{x_3} = c_x$, where $c_x I  = O_{x_1} O_{x_2} O_{x_3}$. More broadly, a scoring operator can be derived by combining Alice and Bob's measurements and all possible output combinations:
\begin{align}\label{eqn:MSG-operator} 
&\mathbf{S}_{MSG} = \frac{1}{18}\sum_{x \in \{1, \dots,6\}, a = a_{x_1},a_{x_2},a_{x_3}, b = b_{x_1},b_{x_2},b_{x_3}} A^j_{a_{x_1}, a_{x_2}, a_{x_3}} \otimes (B^{y=1}_{b_{x_1}} + B^{y=2}_{b_{x_2}} + B^{y=3}_{b_{x_3}}) \text{ where},\\
\noalign{\smallskip}
&A^x_{a_{x_1},a_{x_2},a_{x_3}} = \frac{\bigl( \mathbb{I} + a_{x_1}O_{x_1}\bigr)}{2} \frac{\bigl( \mathbb{I} + a_{x_2}O_{x_2}\bigr)}{2} \frac{\bigl( \mathbb{I} + a_{x_3}O_{x_3}\bigr)}{2}, 
B^k_{b_{x_y}} = \frac{\bigl( \mathbb{I} + b_{x_y}O_{x_y}\bigr)}{2}.
\end{align} 
The expectation value of the operator with an input state $\rho$, $\langle S_{MSG}, \rho \rangle$, can be computed to determine the MSG score. 
\end{definition}

We note that to our knowledge, we are the first to score the MSG for noisy input states using this construction. That is, we derived the score operator as we were unable to find the use of any such scoring operator in the MSG literature. Note that in Eqn. \ref{eqn:MSG-operator}, the summation loops over all possible combinations of outcomes for $a$, and likewise over the outcomes for $b$. The 1/18 comes from the uniform distribution across input settings. Also note that, by the definitions above, $O_s = B^s_{+1} - B^s_{-1}$ for $s\in \{1,\dots, 9\}$.

\begin{definition}[MSG Local and Quantum Bounds]
The optimal \textit{quantum strategy} is played with two shared EPR pairs (up to local isometries): 
$\ket{\psi} = \frac{(\ket{00}_{12} + \ket{11}_{12})}{\sqrt{2}} \otimes \frac{(\ket{00}_{34} + \ket{11}_{34})}{\sqrt{2}}
= \frac{1}{2}(\ket{0000} + \ket{0011} + \ket{1100} + \ket{1111})$ where Alice possesses qubits 1 and 3, and Bob qubits 2 and 4. The grid for the optimal quantum strategy is defined with the observables:
\[
\begin{array}{cccc}
IZ & ZI & ZZ & | \mathbb{I}_4 \\
XI & IX & XX & | \mathbb{I}_4\\
XZ & ZX & YY & | \mathbb{I}_4\\
-- & -- & -- & \\
\mathbb{I}_4 & \mathbb{I}_4 & -\mathbb{I}_4
\end{array}
\]
where each entry (i.e., IZ) acts on one of the party's set of qubits. The first of two observables (i.e., I) acts on the measuring party's first qubit and the second (i.e., Z) acts on their second qubit. The quantum bound, attained with the above strategy, is 1 (i.e., the players win 100\% of the time). Under LHV, the maximal score is 17/18 (local bound). 
\end{definition}

\ifthenelse{\boolean{showfigures}}{\begin{figure}[ht]
    \centering
    \begin{subfigure}[b]{0.3\textwidth}
        \centering
        \includegraphics[width=\textwidth]{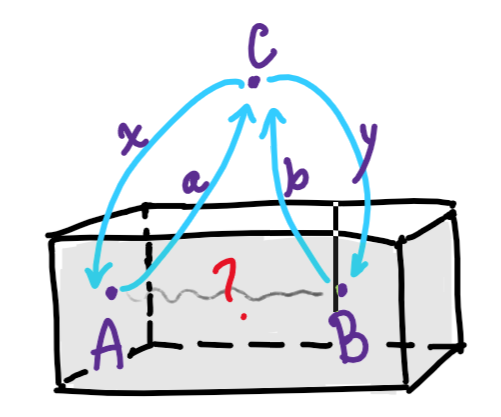}
        \caption{CHSH Game Setup}
        \label{fig:image1}
    \end{subfigure}
    \hspace{0.5cm} 
    \begin{subfigure}[b]{0.35\textwidth}
        \centering
        \includegraphics[width=\textwidth]{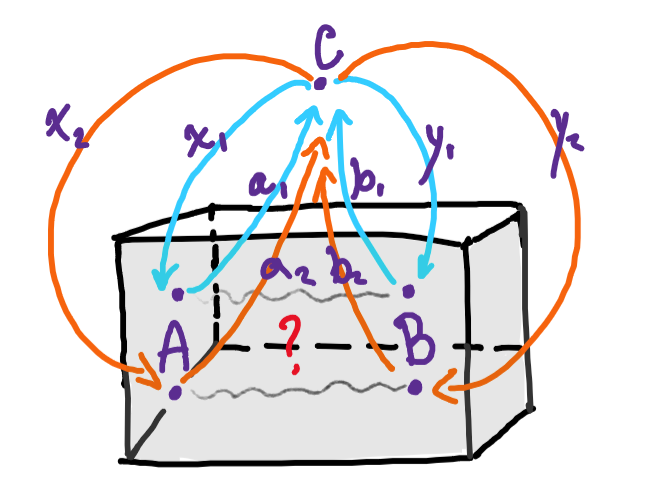}
        \caption{2-CHSH Game Setup}
        \label{fig:image2}
    \end{subfigure}
    \caption[CHSH and 2-CHSH Game Setups]{The CHSH game has 2 inputs and 2 outputs while the 2-CHSH game has 4 inputs and 4 outputs. What is inside the box (quantum or classical) is unknown to the referee, as these are self-tests.}
    \label{diag:CHSH}
\end{figure}}{}

\section{NOISE-ROBUSTNESS FRAMEWORK FOR NON-LOCAL GAMES} \label{sec:def-tools-noiserobustness}
\noindent We want to be able to formalize a reliable method for assessing the noise-robustness of a non-local game to noisy inputs, which, intuitively, is \textit{the ability of a game to detect non-locality in noisy input states}. Ideally, we want to be able to compare the noise-robustness of games with different dimensionalities and input-output settings. In this section, we propose various methods for doing so, and introduce new definitions that prove useful for this objective along the way. 

\subsection{Noise-Tolerance}
An intuitive, natural way to think about noise-robustness is to simply test whether a game is violated for a particular noisy form of the game's winning state, $\rho_{win}$ (up to local isometries). If the noisy state violates the game, then we can say that the game is \textit{tolerant} of that state. For example, in the CHSH game, the winning state, up to local isometries, is the EPR pair. Hence, we can test whether certain noisy versions of the EPR pair violate the CHSH game to conclude whether the game is tolerant of these states. Further, we can expand this intuition to study a set of noisy states defined by a \textit{noise model}, which has an explicit noise level applied to the winning state. Doing so will allow for an analysis of tolerated \textit{noise levels}. This leads to two key questions:
\begin{question}
    Let us define a state, the winning state, $\rho_{win}$ (up to local isometries), of a non-local game, $G$, and a noisy version of it, $\rho'_{win}$ (up to local isometries). Does $G$ detect the non-locality remaining in $\rho'_{win}$? That is, does $\rho'_{win}$ violate $G$'s winning condition?
\end{question} 

\begin{question}
    Define an assumed noise model by a quantum channel yielding the noisy state, $\rho'_{win}=  \varepsilon(\eta, \rho_{win})$, where the probability of no noise (visibility) applied to the winning state of a non-local game $G$ is $\eta$. What are the noise levels which still permit the noisy winning state to violate $G$? 
\end{question}
In the first question, we ask whether a game is tolerant of a specific noisy state, and in the second, we ask if a game is tolerant of a noise interval. Importantly, \textbf{we specifically consider noisy forms of the winning state} of a game as this will allow comparison of noise-tolerance across games.  That is, it is not possible to compare the tolerance of explicit states between games with different input dimensionalities (i.e., the CHSH and MSG), but all games have winning states up to local isometries. Hence, we can compare the \textbf{levels of deviation from the winning state} that still violate games. A quantum channel model for noise helps in establishing a noise parameter to quantify the deviations from the winning state. Thus, the choice of comparing channel noise levels applied to winning states helps to standardize a way for comparing the noise-tolerance between different games. We formally define the notion of noise-tolerance as follows:

\begin{definition}
 [Noise-tolerance of a non-local game] \label{def-noise-tolerance-single-game}
 Define a non-local game $G$ with a local bound $\omega_{c_G}$, and let $\mathbf{S}_G$ denote some fixed observable for the player's strategy in $G$. Further, let us define $\rho_{ref} \in \mathcal{H} = (\mathbb{C}^2)^{\otimes n}$  to be an $n$-qubit state, and $\Phi(\rho_{ref}, \eta)$ to be a noisy quantum channel, where $\eta$ is the visibility affecting $\rho_{ref}$. Then $G$ is $(\eta', \Phi)$-tolerant if $\text{Tr}(\mathbf{S}_G \; \Phi(\eta', \rho_{ref})) > \omega_{c_G}$, for some $\eta'$, noise model $\varepsilon$.
\end{definition}

That is, $G$ is $(\eta', \Phi)$-tolerant if the strategy $\mathcal{S}_G = (\rho' = \Phi(\eta', \rho_{ref}), \mathbf{S}_G)$ violates the maximal classical value ($\omega_{c_G}$), for some $\eta'$ and noise model $\Phi$. Note that we leave $\rho_{ref}$ to be a general state in the definition so that it can be flexibly chosen for the desired question or analysis. However, for our purposes we are concerned with $\rho_{ref}=\rho_{win}$ for the reasons stated above. Let us take an example of applying the definition. We can fix a game, $G_1$, to be the CHSH game. Its optimal state, $\rho_{win_1}$, is the EPR pair. Then, we could analyze the violations of $\omega_{c_{G_1}}$ yielded by strategies $\mathcal{S}^i_{G_1} = (\rho_i = \Phi(\eta_i, \rho_{win_1}), \mathbf{S}_{G_1})$, for varying visibilities $\eta_i \in \{\eta_1, \eta_2, \dots, \eta_n\}$ in a given noise model $\Phi$. Consequently, we would attain information about whether $G_1$ is $(\eta_i, \Phi)$-tolerant. Taking $\eta_{min} = (min\{\eta_1, \eta_2, \dots, \eta_l\}$ s.t. $G_1$ is $(\eta_i, \Phi)$-tolerant), we can claim that $G_1$ is at least $(\eta_{min}, \Phi)$-tolerant. Let us now define the protocol for comparing the noise-tolerance of two different non-local games: 
\begin{definition}
    [Protocol for comparing noise-tolerance of two games]
    \label{protocol-comparingnoisygames}Given two non-local games, $G_1$ and $G_2$, each with local bound $\omega_{c_{G_1}}$ and $\omega_{c_{G_2}}$, respectively, let $\mathbf{S}_{G_1}, \mathbf{S}_{G_2}$ denote some fixed Hermitian operator for the players' strategies in $G_1, G_2$, respectively. Let us also define $\rho_{win_1} \in \mathcal{H} = (\mathbb{C}^2)^{\otimes n}$ to be the $n$-qubit winning state of $G_1$, and $\rho_{win_2} \in \mathcal{H} = (\mathbf{C}^2)^{\otimes m}$ to be the $m$-qubit winning state of $G_2$. 
    \smallskip
    
    \noindent Further, let $\rho'_{1} = \Phi (\eta, \rho_{win_1})$ and $\rho'_{2} = \Phi (\eta, \rho_{win_2})$ be the noisy input states with visibility level $\eta$ to $G_1$ and $G_2$, respectively. If the games have different dimensionalities for the input Hilbert space (i.e., $m > n$), design an extension from the $n$-dimensional channel to the $m$-dimensional one, which preserves a clear relationship between the two channel's noise parameters and is CPTP. To compare the noise-tolerance of $G_1$ and $G_2$, for some $\eta^*$:
    \begin{enumerate}
        \item Decide $\text{Tr}(\mathbf{S}_{G_1} \; \Phi(\eta^*, \rho_{win_1})) > \omega_{c_{G_1}}$.
        \item Decide $\text{Tr}(\mathbf{S}_{G_2} \; \Phi(\eta^*, \rho_{win_2})) > \omega_{c_{G_2}}$. 
    \end{enumerate}
If one game is tolerant to its noisy state and the other is not, the former is more $\eta^*$-tolerant of its winning state than the other, under the noisy channels defined. Otherwise, they are both $\eta^*$-tolerant of their winning states under the noisy channels defined.
\end{definition}
Put more simply: if both conditions 1 and 2 are satisfied, then the games are tolerant to $\eta^*$ noise on their winning state under the noisy channel. Otherwise, one is more tolerant to $\eta^*$ noise. Note that the noise level applied when checking conditions 1 and 2 should be the same to both channels, regardless of whether one is an extension of the other or not. This allows us to observe the effect of noise on larger dimensions or cases where games have parallel repetition.

As can  be seen in Def. \ref{protocol-comparingnoisygames}, comparing the noise-tolerance of two games whose winning states have different dimensionalities can be complex. The situation becomes simpler if the winning state of the higher dimensional game is a tensored repetition of the winning state of the lower dimensional game. Let us illustrate this with an example. Say we wish to compare between the noise-tolerance of the CHSH and MSG games towards some quantum channel $\Phi$. The CHSH has a 2-qubit winning state (a single EPR pair), and the MSG a 4-qubit winning state (two EPR pairs), so we choose:  $\rho_{ref}^{msg} = \rho_{EPR} \otimes \rho_{EPR}$ and $\rho_{ref}^{chsh} = \rho_{EPR}$. In this case, it is natural to define the noise-tolerance comparison as:
\begin{equation}\label{CHSH-check}
    \text{Tr}(\mathbf{S}_{CHSH} \; \Phi(\eta, \rho_{EPR}))  > \omega_{c_{chsh}}
\end{equation}
\begin{equation} \label{MSG-check}
    \text{Tr}(\mathbf{S}_{MSG} \; \Phi(\eta, \rho_{EPR}) \otimes \Phi(\eta, \rho_{EPR}))  > \omega_{c_{msg}}
\end{equation}
The channel for the MSG is simply the tensor of two noisy CHSH inputs. The channels studied in this thesis, the 2- and 4-qubit depolarizing channels (Defs. \ref{def:2qbit-dep-channel} and \ref{def:4qbit-dep-channel}, respectively, with $\rho=\rho_{EPR}$), are based on this tensor product extension. By this construction, if both Eqns. \ref{CHSH-check}
 and \ref{MSG-check}
are satisfied $\forall \eta$, then both games are tolerant to $\Phi$-type noise on the EPR state. Otherwise, one game is more tolerant; that is, it can detect the non-locality of a $\Phi$-noisy EPR state better than the other. Note also that in this simple EPR extended case, in considering CHSH and MSG, we can see if increasing the number of copies of a noisy resource available in a round may improve the ability to detect non-locality in a noisy EPR pair. Later, we also consider a parallel repetition of the CHSH game, the 2-CHSH game. Further, note that when the noise-model is clearly apparent, the $\Phi$ specification in the noise-tolerance description is dropped.

\subsection{Convincingness}
\subsubsection{Motivation}
The second measure of noise-robustness is more specifically designed to capture how efficient non-local games are for detecting non-locality in noisy inputs. There were three main objectives when designing this measure:
\begin{enumerate}
    \item The measure should be a general, continuous score that is \textbf{comparative} among games which have different numbers of inputs, outputs and input state dimensions. For example, it should provide a continuous scale for comparing and ranking the CHSH, 2-CHSH, and MSG's robustness to noisy inputs. This is highly non-trivial and does not yet exist. The challenge is that the scores of different games are on completely different scales which are normalized differently, and the games are in different Hilbert Space dimensions. 
    \item The measure should have an \textbf{operational meaning tied to noise-robustness} of each non-local game.
    \item The measure should be possible to assess for \textbf{finite and infinite resources} of noisy states. This comes from both practical motivation, to understand which game is most efficient for experimentalists to utilize for non-locality detection in noisy resources, and from a theoretical motivation to pave way for understanding the asymptotic strengths or limitations of noise-robustness in different games.
\end{enumerate}
 
\subsubsection{Construction}
We call the new measure the \textit{convincingness} of a non-local game. Intuitively, the convincingness characterizes the likelihood that statistics observed from $n$ trials of a self-test have occurred due to locality (i.e., the likelihood that the statistics are due to maximal or high noise levels in the input state). This measure processes two metrics, (1) the local bound of a game and (2) the success probability of a given input state to the game, into a statistical measure that \textit{is} comparable amongst different games. One can think of this new score as normalizing the violations of each game to the same scale, allowing a more fair comparison of games' abilities to detect non-locality than simply comparing raw violations.   We aim for the convincingness to be a common-ground score which can be computed in a normalized manner amongst all games. 

Let us now describe how the convincingness score is constructed. Its definition is based on  the $p$-value, which is a statistical tool that  has been used to present very strong evidence against local realism in  loophole-free Bell Inequality experiments  \cite{giustina2015significant}. In particular, it has been used to characterize the strength (or significance) of experimental Bell violations results. The formal definition of the $p$-value can be found in Def. \ref{def:pvalue}. Intuitively, the $p$-value represents the probability that observed results occurred due to a pre-defined null hypothesis. A low (upper-tailed) $p$-value indicates that it is highly improbable for our results, or more extreme ones, to occur under the assumption that the null hypothesis is true. If the likelihood is below some significance threshold, $\alpha$, this suggests that the results may be due to an effect unaccounted for by the null hypothesis (there is significant evidence against the null hypothesis). So, practically,  the $p$-value computes the likelihood that an observed output statistical distribution can be explained by an assumed statistical distribution.  

It is natural to extend this to self-testing, where we deal with statistical distributions of physical systems. The construction is as follows. Say that we run $n$ rounds of a non-local game, $G$, noting down the violation measured for each of the $n$ rounds. Then, for each trial $t \in \Omega$, define the win/loss indicator $W: \Omega \to \{0,1\}$ as $W(x,y,a,b) = V(a,b|x,y)$. Under the null hypothesis of local realism, $W \sim \text{Bernoulli}(\omega_{c_G})$, while for a state with visibility $\eta$ and success probability $\omega_{v_G}$, $W \sim \text{Bernoulli}(\omega_{v_G})$. Hence, we may compare the likelihood that the outcome distribution was explained due to a local distribution. Upon performing a binomial test to compare the observed number of successes, $k$, from Binomial$(n, \omega_{v_G})$ to what we would expect under Binomial$(n,\omega_{c_G})$, we can take an upper-tailed test: sum over the binomial test statistic for each possible number of successes $j \in [k, n]; j \in \mathbb{Z}$, yielding the convincingness of a non-local game in Eqn. \ref{eqn:p-value}:
\begin{align}
    P(X \geq k | H_0 = \omega_{c_G}) &= \sum_{j = k}^n P(X = j | H_0 = \omega_{c_G})\\
    \iff P(X \geq \text{round}(n \cdot \omega_{v_G}) | H_0 = \omega_{c_G}) &= \sum_{\text{round}\left(n \cdot \omega_{v_G}\right) }^n \binom{n}{j}\omega_{c_G}^j (1-\omega_{c_G})^{n-j} \label{eqn:p-value}
\end{align} 
Concretely,  this sum gives the probability of observing $k$ or more game violations under a local model. This convincingness test also provides an easy check: if the likelihood (Eqn. \ref{eqn:p-value}) is below a certain significance level ($\leq \alpha$), this indicates that the observed data is sufficiently inconsistent with the null hypothesis of locality for it to be rejected in favour of non-locality. We define $\alpha=0.05$ as this is the standard significance level for the $p$-value across different disciplines. This idea is formally defined in Def. \ref{def:convincingness} and the full computation is defined in Algorithm \ref{algo:convincingness}. By our definition and setup, if $\mathfrak{C}_{G} \leq 0.05$, the results of the non-local game are significantly unlikely to be due to locality, and hence, are highly likely to be due to non-locality (entanglement). In this case, we would therefore say that $G$ is a convincing game for detecting entanglement. Hence, convincingness is a simple yet powerful extension of the $p$-value definition to non-local games.

\begin{definition}[Convincingness]
\label{def:convincingness}
Define $G$ to be a non-local game with local bound $\omega_{c_G}$, and let a measured violation in one round of $G$ with a particular strategy (trial) be denoted by $\omega_{v_G}$. The convincingness, $\mathfrak{C}_{G}$, of $G$ is the upper-tailed $p$-value with the Null Hypothesis, $H_0 = \omega_{c_G}$, the Alternative Hypothesis\footnote{Note that the Alternative Hypothesis can come from any input state, so this is a general construction and not specific to knowing a noise model with visibility $\eta$.}, $H_1 = \omega_{v_G}$, and number of trials $n$. 
\end{definition}

\begin{algorithm}[t] 
\SetAlgoLined
\KwIn{Number of trials $n$, classical bound $\omega_c$, observed success probability $\omega_v$} 
\KwOut{Convincingness $\mathfrak{C}_G$ ($p$-value for non-local game, $G$)}

\BlankLine
Let $X$ be a random variable where:\\
$X = 1$ represents a game violation (success)\\
$X = 0$ represents no violation (failure)

\BlankLine
\tcp{Define probabilistic model under null hypothesis}
Set $\omega_{c}$ as bias under $H_0$ (classical bound of $G$)\\
$X \sim \text{Bernoulli}(\omega_{c})$ where $P(X = 1) = \omega_{c}$

\BlankLine
\tcp{Define probabilistic model under alternative hypothesis}
Set $\omega_{v}$ as bias under $H_1$ (observed success probability for some input to $G$)\\
$X \sim \text{Bernoulli}(\omega_{v})$ where $P(X = 1) = \omega_{v}$

\BlankLine
\tcp{Sample from binomial distribution}
Draw $n$ independent samples $X_1, X_2, \ldots, X_n$ with probability $\omega_{v}$\\
Let $k$ be the number of successes observed in these $n$ trials\\
Then, $k$ is a realization of the random variable $K \sim \text{Binomial}(n, \omega_{v})$\\
Expected successes: $k \sim \mathbb{E}[X] = n \cdot \omega_{v}$

\BlankLine
\tcp{Compute upper-tailed p-value}
$\mathfrak{C}_G = P(X \geq \text{round}(n \cdot \omega_{v})|H_0 = \omega_{c})$\\
$= \sum\limits_{j=\text{round}(n \cdot \omega_{v})}^n \binom{n}{j}\omega_{c}^j(1-\omega_{c})^{n-j}$

\BlankLine
\If{$\mathfrak{C}_G \leq 0.05$}{
    Reject null hypothesis of locality
}

\caption{Convincingness Calculation for Non-Local Games}
\label{algo:convincingness}
\end{algorithm}

Can one game be more or less convincing than the other? When looking at the bare $p$-value, the lower the $p$-value, the greater the statistical significance of the outcome. Hence, for games $G_1, G_2$, if $\mathfrak{C}_{G_1} < \mathfrak{C}_{G_2}$, then  $G_1$ is more `convincing' of non-locality (and hence, entanglement) than $G_2$. 

By extension, we later focus on applying this to mixed states--that is, we will use the convincingness score to study which games provide more evidence of non-locality in noisy states.

\subsubsection{Simplified Convincingness}
We can further simplify the equation for convincingness by borrowing knowledge about the behaviour of the $p$-value as in Theorem \ref{lemma:conv-approx}. The proof follows.

\begin{theorem}
    [Convincingness Bound] \label{lemma:conv-approx} For a non-local game $G$ with $\omega_{c_G}, \omega_{v_G}, n$ as in Def. \ref{def:convincingness}, and $\omega_{c_G} \neq \omega_{v_G}$, we can bound $\mathfrak{C}_{G} \sim \exp\left(-n \left(\omega_{v_G} - \omega_{c_G}\right)^2 \right)$, where we define $G(n) \sim H(n)$ to denote that the function $G(n)$ scales proportionally to some function $H(n)$ for asymptotically large $n$.
\end{theorem}

\begin{proof}
It is known that the lower-tailed $p$-value is upper bounded as follows:
\begin{align}
F(k; n, p) &= \Pr(X \leq k) \leq \exp\left(-2n \left(p - \frac{k}{n}\right)^2 \right), \end{align}
which extends to the following for the upper-tailed $p$-value:
\begin{align}
    F(n - k; n, 1 - p) &= \Pr(X \geq k) \leq \exp\left(-2n \left(\frac{k}{n} - p\right)^2 \right). 
\end{align}
Plugging in the null and alternative hypotheses for our construction of the convincingness metric, we attain:
\begin{align}
    F\left(n - \text{round}(n \cdot \omega_{v_G}); n, 1 - \omega_{c_G}\right) &\leq \exp\left(-2n \left(\frac{\text{round}(n \cdot \omega_{v_G})}{n} - \omega_{c_G}\right)^2 \right) \label{eqn:rounding-error}\\
    &= \exp\left(-2n \left(\omega_{v_G} - \omega_{c_G}\right)^2 \right). \label{eqn:2-factor}
\end{align}
Let $G(n) \sim H(n)$ denote that the function $G(n)$ scales proportionally to some function $H(n)$ for asymptotically large $n$. Then assuming $\omega_{c_G} \neq \omega_{v_G}$, we can accurately approximate the convincingness for $n \rightarrow \infty$:
\begin{align}
    \mathfrak{C}_{G} \sim \exp\left(-n \left(\omega_{v_G} - \omega_{c_G}\right)^2 \right), \label{eqn:pval-simplified}
\end{align}
up to a universal constant factor in the exponent, where the rounding error in line (\ref{eqn:rounding-error}) is negligible. We ignore the factor of 2 in line (\ref{eqn:2-factor}) as it 
is constant for all games and hence does not affect their comparative scaling.
\end{proof}

\subsubsection{Noise-Robustness Framework for Comparative Scores}\label{sec:noise-rob-framework}

\ifthenelse{\boolean{showfigures}}{
\begin{wrapfigure}{r}{0.5\textwidth}
    \centering
    \begin{tikzpicture}[scale=0.5,
        axis/.style={thick, ->},
        note/.style={font=\footnotesize}
    ]
        \node[circle,fill=black,inner sep=2pt] (origin) at (0,0) {};
        
        \draw[axis] (origin) -- (3,0) node[right,label,text width=2cm,align=center] {II: State\\ ($\rho$)};
        \draw[axis] (origin) -- (-3,0) node[left,label,text width=3cm,align=center] {IV: Measurement\\ Operators\\ ($\mathcal{M}_1, \mathcal{M}_2$)};
        \draw[axis] (origin) -- (0,3) node[above,label,text width=2cm,align=center] {I: Weights\\ ($\pi(x,y)$)};
        \draw[axis] (origin) -- (0,-3) node[below,label,text width=2.5cm,align=center] {III: Coefficients\\ ($V(a,b|x,y)$)};
        \node[rotate=90, note,align=center] at (-0.4,1.5) {Game};
        \node[note] at (1.5,0.3) {Strategy};
    \end{tikzpicture}
    \caption[Optimization Axes for Non-Local Games.]{Optimization axes for non-local games, illustrating the four adjustable components of a non-local game configuration. That is, we can optimize or evaluate over the following--I: The weights, $\pi(x,y)$, of the measurements in $\mathbf{S}$, II: Quantum state $\rho$, III: The coefficients derived from $V(a,b|x,y)$ for the operators in $\mathbf{S}$, and IV: The measurement operators $\mathcal{M}_1, \mathcal{M}_2$ used to construct $\mathbf{S}$. }
    \label{fig:NLG-axes}
\end{wrapfigure}
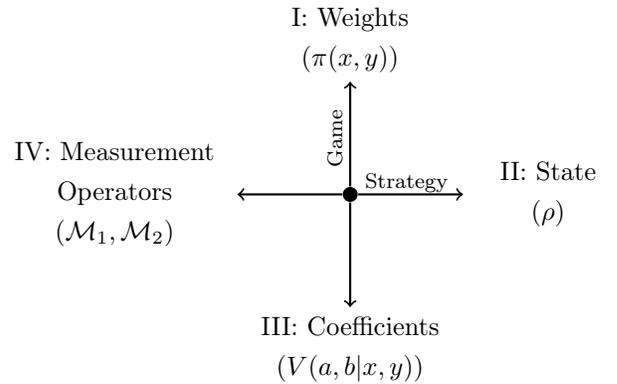}{}
Given that we now have a continuous, comparative score across different games ($\mathfrak{C}_G$), we now use it to develop a framework for studying the noise-robustness of games to noisy input states, and particularly, to noisy channels on the input. Let us first understand a nuance which we glossed over in the prior section. While we have been describing $\mathfrak{C}_G$ as the convincingness of a `game', both the `game' and `strategy' are  hiding within $\omega_{v_G} = \text{Tr}(\mathbf{S}_{G} \rho_{v_G})$ and $\omega_{c_G} = \text{Tr}(\mathbf{S}_{G} \rho_{c_G})$ of  $\mathfrak{C}_G$. That is,
the \textit{state} $\rho$ and the \textit{measurement operators} of $\mathbf{S}$ come purely from the strategy $\mathcal{S} = (\rho, \mathcal{M}_1, \mathcal{M}_2)$. In contrast, the structure of $\mathbf{S}$--the \textit{weights} $\pi(x,y)$ and the \textit{coefficients} (on the measurements) derived from $V(a,b|x,y)$--comes purely from the game definition $G =(\mathcal{X},\mathcal{Y},\mathcal{A},\mathcal{B},\pi,V)$. This means that there are four axes, assuming that the input-output settings and number of players defining the game are fixed, as in Fig. \ref{fig:NLG-axes}. We call a fixed choice of the four axes a \textit{configuration} of a non-local game. This framework provides a natural decomposition for understanding noise-robustness in non-local games: we can conceptualize it both through the analysis of game performance across varying noise levels and through the systematic optimization of game parameters for a single or set of noisy states. Specifically, we can do the following.
\begin{enumerate}
    \item \textbf{Analysis across noise levels (Fig. \ref{subfig:non-fixedII}):} We can vary the state axis (II) while holding all other axes constant. That is, first establish a configuration for which you wish to test noise-robustness. Second, evaluate the comparative score of the configuration at different input state noise levels. In this work, we choose to vary the states according to a quantum noise channel (depolarizing channel), and choose $\mathfrak{C}_G$ for the score as it is comparable among different games and considers the quantity of resources physically available. However, other noisy channels or more complex noise models can be chosen to define the state axis, and other scores can be used if they are comparable among games. By evaluating configurations across different noise levels, this analysis can elucidate differences in robustness: some configurations may only detect entanglement in highly entangled states while failing with weakly entangled states. In contrast, particularly robust configurations can identify entanglement even in weakly entangled states. 
    \item \textbf{Optimization for specific noise regimes (Fig. \ref{subfig:non-fixedIII}):} We can select a specific noisy state or regime (set of states) to define the state axis on (II), and optimize one or more of the axes for that noisy regime. For example, by optimizing the operator coefficient axis (III) for a fixed $\eta'$-noisy state (see Fig. \ref{subfig:non-fixedIII}), we can identify configurations particularly well-suited to specific noise levels, potentially revealing new configurations with enhance noise-robustness. These optimized configurations can then be evaluated for their robustness across noise levels using the first approach.
\end{enumerate}
The convincingness measure also provides an analytical framework to understand this behaviour. For weakly entangled states, $\omega_{v_G}$ approaches $\omega_{c_G}$, resulting in $\mathfrak{C}_G \rightarrow 1$, which mathematically captures their consistency with local behaviour. Conversely, highly entangled states drive $\omega_{v_G}$ toward its maximal quantum value, resulting in $\mathfrak{C}_G \rightarrow 0$ and indicating strong non-local behaviour. For an arbitrary input state $\rho_{in}$, this relationship can be expressed as:
$\mathfrak{C}_{G} \sim \exp \left(-n \left(\text{Tr}(\mathbf{S}_{G} \rho_{in}) - \omega_{c_G}\right)^2 \right)$. Note that if we wish to specify a particular noise level $\eta$ for some noise model, we can substitute an $\eta$-noisy input $\rho_{\eta}$, which achieves the value $\omega_{\eta_G} = \text{Tr}(\mathbf{S}_G\rho_{\eta})$. Moreover, in later sections we derive a single numeric, continuous score based on the simplified $\mathfrak{C}_G$, in order to rank the convincingness behaviour, and hence, noise-robustness, of different games.   

\ifthenelse{\boolean{showfigures}}{\begin{figure}[htbp]
    \centering
    \begin{subfigure}[b]{0.45\textwidth}
        \centering
        \begin{tikzpicture}[
  scale=0.5,
  baseline=(current bounding box.center),
  axis/.style={thick, ->},
  label/.style={font=\small\itshape},
  note/.style={font=\footnotesize} 
]

    \definecolor{pastelred}{RGB}{205,65,65} 
\definecolor{pastelgreen}{RGB}{67,160,71}  

            \node[circle,fill=black,inner sep=2pt] (origin) at (0,0) {};
            
            \draw[very thick, ->, pastelgreen] (origin) -- (3,0) node[right] {(II)};
            \draw[very thick, ->, pastelred]   (origin) -- (-3,0) node[left]  {(IV)};
            \draw[very thick, ->, pastelred]   (origin) -- (0,3)  node[above] {(I)};
            \draw[very thick, ->, pastelred]   (origin) -- (0,-3) node[below] {(III)};
            \node[rotate=90, note, align=center] at (-0.4,1.5) {};
            \node[note,align=center] at (1.5,0.3) {$(\rho_\eta, \;\forall \eta)$};
        \end{tikzpicture}
        \caption{Free State Axis (Analysis)}
        \label{subfig:non-fixedII}
    \end{subfigure}
    \hspace{0.5em} 
    \begin{subfigure}[b]{0.45\textwidth}
        \centering
        \begin{tikzpicture}[
  scale=0.5,
  baseline=(current bounding box.center),
  axis/.style={thick, ->},
  label/.style={font=\small\itshape},
  note/.style={font=\footnotesize}
]

\definecolor{pastelred}{RGB}{205,65,65}  
\definecolor{pastelgreen}{RGB}{67,160,71} 

            \node[circle,fill=black,inner sep=2pt] (origin) at (0,0) {};
            
            \draw[very thick, ->, pastelred]   (origin) -- (3,0) node[right] {(II)};
            \draw[very thick, ->, pastelred]   (origin) -- (-3,0) node[left]  {(IV)};
            \draw[very thick, ->, pastelred]   (origin) -- (0,3)  node[above] {(I)};
            \draw[very thick, ->, pastelgreen] (origin) -- (0,-3) node[below] {(III)};
            \node[rotate=90, note,align=center] at (-0.4,1.5) {};
            \node[note] at (1.5,0.3) {$\rho_{\eta'}$};
        \end{tikzpicture}
        \caption{Free Coefficients (Optimization)}
        \label{subfig:non-fixedIII}
    \end{subfigure}
    \caption[Axes Visualizations for Evaluating Noise-Robustness with a Comparative Score.]{\textbf{Experiments for evaluating noise-robustness with a comparative score.} Illustrating two ways to look at how a game can be robust to noise: robust across a region, or robust for a particular noise type. Green axes are varied and Red axes are fixed. 3a (left): Fixing known (i.e., optimal) weights, coefficients, and measurements while varying the visibility $\eta$ of the input state $\rho_\eta$. 3b (right): Optimizing the coefficients for a pre-selected state with visibility $\eta'$, fixing the rest of the axes. The resulting configuration from 3b can be analyzed using the method in 3a.}
    \label{fig:axis-experiments}
\end{figure}
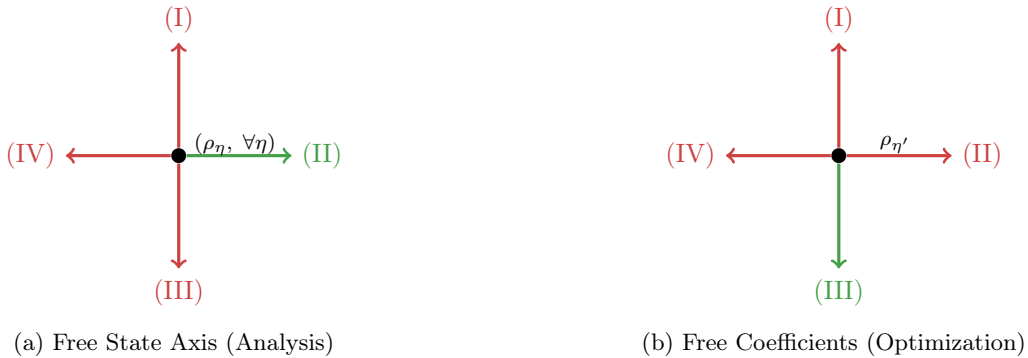}{}

\section{COMPUTATIONAL METHODS}\label{sec:comp-methods}
We demonstrate how to compare the noise-robustness of 2-player games using our noise-robustness framework for comparative scores, and the noise-tolerance. The games we focus on as an example are the CHSH, 2-CHSH, MSG, and optimized 2-CHSH variants (defined in Sec. \ref{sec:2-chsh-opt}). Comparing this set is natural as we can see if parallel repetition of the CHSH (i.e., the 2-CHSH) helps to improve its noise-robustness, and whether the 2-CHSH or its optimized versions are comparable to the MSG which uses the same amount of quantum resources (qubits) in its optimal quantum strategy. We expect these experiments to also provide a good benchmark of the quality of the convincingness for comparing the noise-robustness of games which have different number of inputs, outputs, and input state dimensions. Using the convincingness we can use this analysis to discern whether the non-locality in noisy, partially depolarized states are best detected in the CHSH, 2-CHSH, MSG, or optimized 2-CHSH games. Overall, computationally we do the following:
\begin{enumerate}
    \item Evaluate the convincingness for states with different noise levels, for different games, and compare which game permits more noisy states to be convincing. Plots of such experiments will be shown in Section \ref{sec:results-disc}. 
    \item Search for 2-CHSH configurations which are optimal for particular noise levels, included as well in the plots of (1).
    \item Analyze the convincingness of non-local games at fixed and variable, finite and near-infinite resource regimes.
    \item Evaluate and plot the noise-tolerance of non-local games, relating these results to convincingness results.
\end{enumerate}

\subsection{Analysis Across Noise Levels} \label{sec:analysis-across-noise-levels}
Let $NLG = \{\text{CHSH}, \text{ MSG}, \text{ 2-CHSH}, \text{ 2-CHSH-OPT}\}$, where the first three games are as defined in Sec. \ref{sec:NLGs}, and the 2-CHSH-OPT will be momentarily defined. Also let $\rho_\eta$ be some $\eta$-noisy input to $G \in NLG$ for a noise model with visibility $\eta$ (probability of noise $1-\eta$), yielding score $\omega_{\eta_G} = \text{Tr}(\mathbf{S}_G\rho_\eta)$. $\mathbf{S}_G$ is the operator for $G$ composed of the measurement coefficients, and weights which achieve optimal score for a maximally entangled state $\rho_{{\eta=1}}$  (up to local isometries).  

In this experiment, we fix the configuration of measurements, coefficients, and weights (axes IV, III, I) to yield the maximal score at $\rho_{{\eta=1}}$ for each game. Further, we choose the noise model on the input state to be the depolarizing channel, defining $\rho_{\eta}$ for the 2-qubit and 4-qubit Hilbert subspaces, respectively, as follows:
    \begin{align*}
\varepsilon^{2-qubit}_{\eta}(\rho_{EPR}) = \eta^2 \rho_{EPR} + (1-\eta^2)\frac{\mathbb{I}_4}{4}, \;\;\; \;\;\varepsilon^{4-qubit}_{\eta}(\rho_{EPR}) &= \eta^4(\rho_{\text{EPR}} \otimes \rho_{\text{EPR}}) + 
    \eta^2(1-\eta^2)(\rho_{\text{EPR}} \otimes \frac{\mathbb{I}_4}{4}) \\
    &\quad + \eta^2(1-\eta^2)(\frac{\mathbb{I}_4}{4} \otimes \rho_{\text{EPR}}) +
    (1-\eta^2)^2(\frac{\mathbb{I}_4}{4} \otimes \frac{\mathbb{I}_4}{4}).
\end{align*}
We also now define $\mathfrak{C}_G$ specifically for hypotheses $H_0 = \omega_{c_G}$ (local bound of $G$) and $H_1 = \omega_{\eta_G}$ (score with input $\rho_\eta$). As this is the convincingness we'll be dealing with from now on, we often refer interchangeably to $\mathfrak{C}_G$ and $\mathfrak{C}_G(\eta)$.

With this setup, we take the operator $\mathbf{S}_G$ and compute $\mathfrak{C}_G$ for all $\eta\in [0,1]$. For each game, this produces a curve $\{(\eta, \mathfrak{C}_G) \;\; \forall \eta\in [0,1]\}$, and all curves are plotted for comparison. Importantly, the order in which the curves cross the significance threshold $\alpha$ can be determined from this plot, revealing the maximal noise levels on $\rho_{EPR}$ and $\rho_{EPR} \otimes \rho_{EPR}$ at which each game is convincing of non-locality. 

Hence, the above experiment aims to answer: \textit{How does $\mathfrak{C}_G$ behave as the noise level of a state $\rho$ is varied? Which game is convincing (crosses $\mathfrak{C}_G = \alpha$), for a higher amount of noise (lower visibility, $\eta$)?} Answering the latter reveals which game is most noise-robust by the convincingness measure. Beyond these questions, the convincingness enables our study of how the amount of resources available affects each game's likelihood of detecting non-locality. That is, \textit{given a supply of identical noisy resources ($N_{res}$) to play each game, which game is most convincing for the highest noise?} We study this question both at finite and near-infinite $N_{res}$ regimes. Note that in our experiments, $N_{res}$ defines the number of noisy EPR pairs available. Hence, games with one or two EPR pairs as their winning state can play $N_{res}$ or $N_{res}/2$ rounds of the game, respectively, if we restrict them to use the same number of resources.

\subsection{Optimization for Specific Noise Regimes}\label{sec:ptimization-for-specific-noise-regimes}
We demonstrate how one can optimize a non-local game for noisy states, using the 2-CHSH game as an example. We choose to optimize the 2-CHSH game since, as will be seen in Sec. \ref{sec:results-disc}, the 2-CHSH shows worse noise-robustness than the CHSH but better noise-robustness than the MSG. Hence, the objective of the optimization is to show whether optimizing an axis of the 2-CHSH game's configuration can help improve upon and potentially surpass the CHSH's noise-robustness.

\subsubsection{Optimized 2-CHSH Construction}\label{sec:2-chsh-opt}
The optimization is based on the work from Istv{\'a}n et. al. \cite{marton2023bounding}, extending their linear programming method to find the optimal Bell coefficients (predicate, axis III) that will maximize the achievable quantum violation of the 2-CHSH when subject to a specific level of noise ($\eta'$) in the input state (axis II) (see Sec. \ref{sec:LP} for details on the LP). In theory, thinking about noise-robustness, sequential optimizations along multiple axes can be done while fixing the state axis in order to find a truly optimal configuration. In addition, this method is extendable to optimizing over a noisy regime defined by a noise interval on $\eta$. However, for simplicity and illustrative purposes, we only optimize over the coefficient axis for a specified noisy state (Fig. \ref{subfig:non-fixedIII}). This is also a good opportunity to illustrate the convenience of the noise-robustness framework. Looking to the game axes, we can concisely say that \cite{marton2023bounding} optimizes over the coefficients (axis III) for specified measurements on axis IV, whereas we now optimize the coefficients (axis III) for specified noisy states on axis II holding all other axes constant. Note that we will refer to configurations resulting from this process as Optimized 2-CHSH configurations. 

Let us now define the Optimized 2-CHSH Game (2-CHSH-OPT).  As before, the game has two players and a referee, with $\mathcal{X}_1 = \mathcal{Y}_1 = \mathcal{X}_2 = \mathcal{Y}_2 = \mathcal{A}_1 = \mathcal{B}_1 = \mathcal{A}_2 = \mathcal{B}_2 = \{0,1\}$ and uniform $\pi(x_1,y_1,x_2,y_2) = \frac{1}{16}$. The measurements for the strategy are as defined for the 2-CHSH game. The measurements, are thus, constructed as follows. 
\begin{definition}[2-CHSH-OPT Strategy]\label{def:2-chsh-opt-strategy}
Let us define the single-qubit measurement settings of Alice and Bob to be $A_{a_k}^{x_k}$ and $B_{b_k}^{y_k}$, respectively. Alice's single-qubit basis is defined by $\sigma_x$ and $\sigma_z$, and Bob's by  $\sigma_{x+z} = \frac{\sigma_x + \sigma_z}{\sqrt{2}}$ and $\sigma_{x-z} = \frac{\sigma_x - \sigma_z}{\sqrt{2}}$. From these, we can construct a set of projectors:
\begin{align*}
&A_0^+ \coloneqq \frac{\mathbb{I}_2 + \sigma_x}{2}, \quad A_0^- \coloneqq \frac{\mathbb{I}_2 - \sigma_x}{2}, \quad A_1^+ \coloneqq \frac{\mathbb{I}_2 + \sigma_z}{2}, \quad \;A_1^- \coloneqq \frac{\mathbb{I}_2 - \sigma_z}{2} \\
&B_0^+ \coloneqq \frac{\mathbb{I}_2 + \sigma_{x+z}}{2}, B_0^- \coloneqq \frac{\mathbb{I}_2 - \sigma_{x+z}}{2}, B_1^+ \coloneqq \frac{\mathbb{I}_2 + \sigma_{x-z}}{2}, B_1^- \coloneqq \frac{\mathbb{I}_2 - \sigma_{x-z}}{2}. 
\end{align*}

\noindent Combining these 1-qubit projectors, Alice and Bob's 2-qubit projective measurements are defined, respectively, as: 
\begin{align*}
\mathbb{A} = [ & (A_0^+ \otimes A_0^+), (A_0^+ \otimes A_0^-), (A_0^- \otimes A_0^+), (A_0^- \otimes A_0^-), (A_0^+ \otimes A_1^+), (A_0^+ \otimes A_1^-), (A_0^- \otimes A_1^+), (A_0^- \otimes A_1^-), \\
        &  (A_1^+ \otimes A_0^+), (A_1^+ \otimes A_0^-), (A_1^- \otimes A_0^+), (A_1^- \otimes A_0^-), (A_1^+ \otimes A_1^+), (A_1^+ \otimes A_1^-), (A_1^- \otimes A_1^+), (A_1^- \otimes A_1^-) ],\\ \smallskip
\mathbb{B} = [ & (B_0^+ \otimes B_0^+), (B_0^+ \otimes B_0^-), (B_0^- \otimes B_0^+), (B_0^- \otimes B_0^-), (B_0^+ \otimes B_1^+), (B_0^+ \otimes B_1^-), (B_0^- \otimes B_1^+), (B_0^- \otimes B_1^-)\\
        &  (B_1^+ \otimes B_0^+), (B_1^+ \otimes B_0^-), (B_1^- \otimes B_0^+), (B_1^- \otimes B_0^-), (B_1^+ \otimes B_1^+), (B_1^+ \otimes B_1^-), (B_1^- \otimes B_1^+), (B_1^- \otimes B_1^-) ].
\end{align*}
Moreover, the general input state for the game is $\rho \in \mathbb{C}^{16 \times 16}$.
\end{definition}

We have now fixed the measurement configurations. While the measurements would typically be used to construct the score operator ($\mathbf{S}_{2chsh,\eta'}$) using a set of weights and coefficients, the score is computed using the output of the linear program, which is in a different format. The linear program outputs a Bell matrix, $\mathbf{B}_{\eta'}$, whose elements are the coefficients which have been optimized for a fixed $\eta'$-noisy input; see Sec. \ref{sec:LP} for its construction. The score is computed by multiplying $\mathbf{B}_{\eta'}$ as in Def. \ref{def:2-CHSH-score} with a probability matrix, $\mathbf{P}_{\eta}$, which encodes the strategy (state and measurement) and weights axes. In this way, we consider both the predicate, the weights, and the strategy.
\begin{definition}[2-CHSH-OPT Score] \label{def:2-CHSH-score} Given a 2-CHSH game whose coefficients have been optimized for an $\eta'$-noisy input, its score is defined as: 
    \begin{align}
    \omega_{\eta_{2chsh,\eta'}} = \text{Tr}(\mathbf{P}_{\eta}\bar{\mathbf{B}}_{\eta'}^T) \text{ for } \bar{\mathbf{B}}_{\eta'} \coloneqq \frac{\mathbf{B}_{\eta'} - \mathbf{B}_{\eta'}^{min}}{(\mathbf{B}_{\eta'} - \mathbf{B}_{\eta'}^{min})^{max}}. \label{eqn:score-2chsh-opt}
\end{align}
where $\mathbf{X}^{max}$ and $\mathbf{X}^{min}$ denote the max and min elements of a matrix, $\mathbf{X}$, respectively. $\mathbf{P}_{\eta} \in \mathbb{R}^{16\times16}$ is a probability matrix for which each element is the probability of an outcome, given an $\eta$-noisy input, $\rho_{\eta} \in \mathbb{C}^{16 \times 16}$, and measurement settings $A_i = A_{a_1}^{x_1} \otimes A_{a_2}^{x_2}, B_j = B_{b_1}^{y_1} \otimes B_{b_2}^{y_2}$:
$\mathbf{P}_{\eta}^{element} = P(a_1,a_2,b_1,b_2|x_1,x_2,y_1,y_2) = \text{Tr}(\rho_{\eta} (A_{a_1}^{x_1} \otimes A_{a_2}^{x_2} \otimes B_{b_1}^{y_1} \otimes B_{b_2}^{y_2})) = \text{Tr}(\rho_{\eta} (A_i \otimes B_j))$.
The whole matrix is defined as follows with the uniform probability of input questions $\pi=1/16$:
\begin{align}
\mathbf{P}_{\eta} &=\pi \cdot \sum_{i,j=1}^{\text{16}} \text{Tr}(\rho_{\eta} (A_i \otimes B_j))|e_i\rangle\langle e_j|. \label{eqn:matrix-P}
\end{align}
The matrix $\mathbf{B}_{\eta'} \in \mathbb{R}^{16\times16}$ is the Bell matrix, which encodes the coefficients used to combine the expectations in $\mathbf{P}_{\eta}$.
\end{definition}

We make a few comments about the score construction. Note that $\eta$ represents the noise level of the strategy being tested by the game, and $\eta'$ represents the visibility that the game's coefficients have been optimized for. In our experiment, we once again define the noise model on the state to be the depolarizing channel on the input space: $\rho_{\eta} = \varepsilon^{4-qubit}_\eta(\rho_{EPR})$. Note also that the affine transform on $\mathbf{B}_{\eta'}$ is added after the linear program for our purposes, to ensure that the final score in Eqn. \ref{eqn:score-2chsh-opt} is in $[0,1]$ and, hence, can be used in the convincingness computation. Further, note that the probability matrix is constructed by summing over all possible measurement setting combinations ($16 \times 16$). This is important for the normalization of the score as well. It is also worth mentioning that the indexing of the projector lists in Def. \ref{def:2-chsh-opt-strategy} matters for the code and how $\mathbf{P}_{\eta}^{element}$ is constructed. Measurement $A_i$ corresponds to the $i'th$ element of $\mathbb{A}$ and measurement $B_j$ corresponds to the $j'th$ element of $\mathbb{B}$. For example, $A_0\otimes B_1 = (A_0^+ \otimes A_0^+) \otimes (B_0^+ \otimes B_0^+)$.  

Finally, by Def. \ref{def:2-chsh-opt-bounds}, the local bound of a 2-CHSH-OPT game can be simply computed given the game's Bell matrix and the quantum bound can be simply computed given both the game's Bell matrix and $\mathbf{P}_{\eta=1}$. 

\begin{definition}[2-CHSH-OPT Local and Quantum Bounds] \label{def:2-chsh-opt-bounds}
    For a 2-CHSH-OPT configuration, the quantum bound is $\omega_{\eta_{2chsh,\eta'}}$ for $\eta=1$. The local bound is the maximum score achieved by deterministic strategies using the game's $\bar{\mathbf{B}}_{\eta'}$ (see App. \ref{app:2chsh-opt-local-bound}). By convexity, this maximum over deterministic strategies completely characterizes the local bound.
\end{definition}

\subsubsection{Linear Program}\label{sec:LP}
Now that we have defined the construction for an optimized 2-CHSH game, we may understand how the linear program (LP) is set up to optimize it for a particular noise level in our noise model. We modify the LP and code from \cite{marton2023bounding}. In \cite{marton2023bounding}, the LP maximizes the achievable quantum violation of the 2-CHSH when
subject to noise in the \textit{measurements}. Given a probability matrix $\mathbf{P}(\gamma)$ defined for $\gamma$ noise on measurements (detection efficiency), the LP outputs the optimized Bell coefficients. As $\mathbf{P}(\gamma)$ is defined both by the input state and measurements, we set perfect detection efficiency and encode our depolarizing noise model, $\varepsilon_\eta^{4-qubit}(\rho_{EPR})$, in the probability matrix for a given visibility $\eta'$ as in Eqn. \ref{eqn:matrix-P}. So, $\mathbf{P}_{\eta'} \coloneqq \mathbf{P}(\gamma=1, \eta')$. The objective function and constraints are then the same as \cite{marton2023bounding}, with a modified probability matrix:
\begin{align*}
\text{minimize} \quad & -\sum_{i,j=1}^{16} \mathbf{P}_{\eta'}^{ij}c_{ij} \\
\text{subject to} \quad & \sum_{i,j=1}^{16} L_{k}^{ij}c_{ij} \geq 1 \quad \forall k \in \{1,\ldots,4^8\} \\
& -1 \leq c_{ij} \leq 1 \quad \forall i,j \in \{1,\ldots,16\}
\end{align*}
where the optimization variables $c_{ij}$ represent $\mathbf{B}_{\eta'}^{ij}$, the coefficients for a measurement combination, where $i,j$ range from 1 to 16 to cover all of Alice and Bob's possible measurements, respectively. The constraints ensure that no local strategy violates the Bell inequality, using the vector of local vertices $L$ to encode all deterministic LHV strategies. The code is written in MATLAB \cite{MATLAB}, which we translated and extended into Python \cite{van1995python}. 

\subsection{Noise-Tolerance Analysis}\label{sec:noise-tol-computational-methods}
We compute a simple noise-tolerance analysis for all the games tested. For each game $G \in \{ \text{CHSH}, \text{MSG}, \text{2-CHSH},$ $\text{2-CHSH-OPT Games}\}$, the setup is as follows: 
\begin{enumerate}
    \item For CHSH, set the reference state for the noise-tolerance to be its the winning state, $\rho_{EPR}$, and hence the depolarized input is $\rho_{\eta}=\varepsilon_\eta^{2-qubit}(\rho_{EPR})$. For the other 4-dimensional games, use the tensor extended depolarized input $\rho_{\eta}=\varepsilon_\eta^{4-qubit}(\rho_{EPR})$. The strategy is then: $\mathcal{S}_{G,\eta} = (\rho_{\eta}, \mathbf{S}_G)$ for the input state corresponding to $G$.
    \item Compute a set of scores, $\omega_{\eta_G} = \text{Tr}(\rho_\eta \mathbf{S_G})$, over the interval $\eta \in [0,1]$. 
    \item Plot the $\eta$ values versus the scores for the game, producing a curve.     
    \item Find the $\eta$ at which the game crosses its local bound, $\omega_{c_G}$. Let us rename this  $\eta$ to  $\eta^*$. Then, the game is $\eta^*$-tolerant.  
\end{enumerate}
We can plot all curves on the same axis, made possible by the fact that we evaluate different games for the same $\eta$ values in corresponding noise models. The tolerances of the games can then be compared, allowing us to determine which game is more tolerant to depolarizing noise. That is, we may determine which game detects non-locality for greater depolarization.

\section{ANALYTIC METHODS}\label{sec:analytic-methods}
We also develop a more theoretical analysis. Specifically, we analytically derive a single numeric score which we coin, the \textit{gapped score}, which can be mapped to each game to indicate the order in which games pass the significance threshold. The score's structure not only helps to dissect the reasons for the observed results, but also allows to have a more convenient measure for noise-robustness of a non-local game. It allows us to have a single number to rank the convincingness behaviour of a non-local game across all $\eta \in [0,1]$. Moreover, we offer a method for performing a more nuanced analysis for games which have very similar convincingness behaviour and hence similar gapped scores. 
  
\subsection{Gapped Expressions}
We start with the simplified convincingness expression for an $\eta$-noisy state, which reflects the behaviour of the convincingness for high $n$: $\mathfrak{C}_{G} \sim \exp\left(-n \left(\omega_{\eta_G} - \omega_{c_G}\right)^2 \right)$. In particular, we note that lower $\mathfrak{C}_{G}$ corresponds to a higher \textit{gap}: $\Delta_G = \omega_{\eta_G} - \omega_{c_G}$. Therefore, the gap can be an indicator of $\mathfrak{C}_{G}$ for high $n$. Expanding $\Delta_G$ and $\Delta_G^2$ for the depolarizing channel noise  model, one will find that the expression is at most an eighth degree polynomial, with even degrees:
\begin{align*}
    \Delta_G = d + c_1\eta^2 + c_2\eta^4 - \omega_{c_G}, \;\;\; \Delta_G^2 = z + r_1\eta^2 + r_2\eta^4 + r_3\eta^6 + r_4\eta^8.
\end{align*}
Hence, there are multiple candidates for a numeric score which can model convincingness across noisy regimes. We consider $\Delta_G$, $\Delta_G^2$, the most significant coefficients of both (constants, $\eta^2, \eta^4)$, or some unique transform of the coefficients, constants and local bound. Let us obtain the candidate scores for each game by deriving their coefficients and constants in $\Delta_G$. We start with a convenient way to obtain these for 4-qubit games.
\begin{theorem}
For $G \in \{\text{2-CHSH}, \text{MSG}, \text{2-CHSH-OPT}\}$) receiving input $\varepsilon^{4-qubit}_\eta(\rho_{EPR})$ for $\eta$, the gap is, 
\begin{align}
    \Delta_G = \eta^4 k_\text{ideal} + \eta^2(1-\eta^2)k_\text{part-mixed} + (1-\eta^2)^2 k_\text{mixed} -\omega_{c_G},
\end{align}
\noindent
where $k_\text{ideal}$, $k_\text{part-mixed}$, and $k_\text{mixed}$ are constants specific to each $G$.
\end{theorem}

\begin{proof}
    Define $\mathbf{O}$ to be some score operator of a game $G$ and $\pi$ to be some normalization. Then the value for $G$ played with an $\eta$-depolarized input is:
    {\small
\begin{align*}
    \omega_{\eta_{G}} = \pi\text{Tr}\biggl[\varepsilon_{\eta}^{4-qubit}(\rho_{EPR})\mathbf{O}\biggr] &=  \pi\text{Tr}\left[\left(\left(\eta^2\rho_{\text{EPR}} + (1-\eta^2)\frac{\mathbb{I}_4}{4}\right) \otimes 
    \left(\eta^2\rho_{\text{EPR}} + (1-\eta^2)\frac{\mathbb{I}_4}{4}\right)\right) \mathbf{O}\right]\\    &=\eta^4 \text{Tr}\biggl[\pi(\rho_{\text{EPR}} \otimes \rho_{\text{EPR}})\mathbf{O}\biggr] + 
    \eta^2(1-\eta^2)\text{Tr}\biggl[\pi(\rho_{\text{EPR}} \otimes \frac{\mathbb{I}_4}{4})\mathbf{O} + \pi(\frac{\mathbb{I}_4}{4} \otimes \rho_{\text{EPR}})\mathbf{O}\biggr] +
    (1-\eta^2)^2\text{Tr}\biggl[\pi(\frac{\mathbb{I}_4}{4} \otimes \frac{\mathbb{I}_4}{4})\mathbf{O}\biggr]\\
    &\coloneqq \eta^4 k_\text{ideal} + \eta^2(1-\eta^2)k_\text{part-mixed} + (1-\eta^2)^2 k_\text{mixed},
\end{align*}
}
    \noindent for constants $k_\text{ideal}, k_\text{part-mixed}, k_\text{mixed}$. Note that the score for the 2-CHSH and MSG is computed this way. That is, as defined in \ref{def:2-chsh-strategy-score}, for the 2-CHSH $\pi=1/16$ and $\mathbf{O} = \sum_{\substack{x_1,y_1,x_2,y_2,a_1,b_1,a_2,b_2 \\ (a_1\oplus b_1 = x_1\wedge y_1)\wedge(a_2\oplus b_2 = x_2\wedge y_2)}}((M_{a_1|x_1} \otimes M_{a_2|x_2}) \otimes (M_{b_1|y_1} \otimes M_{b_2|y_2}))\biggr)$ where the sum is over all measurement combinations: 
    {\small
    \begin{align*}
        (a_i,b_i,x_i,y_i) \in \{(0, 0, 0, 0), (1, 1, 0, 0), (0, 0, 0, 1), (1, 1, 0, 1), (0, 0, 1, 0), (1, 1, 1, 0), (0, 1, 1, 1), (1, 0, 1, 1)\}.
    \end{align*}}
    Further, the operator for the MSG is Eqn. \ref{eqn:MSG-operator} with $\pi=1/18$. Since 2-CHSH-OPT games deal with separate objects, it isn't clear if the above applies. Hence, we derive the value of an Optimized 2-CHSH game: 
{\small
\begin{align*}
\omega_{\eta_{2chsh, \eta'}} = \text{Tr}(\mathbf{P}_\eta \bar{\mathbf{B}}_{\eta’}^T) &= \text{Tr} \left[\frac{1}{16} \left(\sum_{a,b=1}^{16} \text{Tr}(\rho_{in} (A_a \otimes B_b))|e_a\rangle\langle e_b|\right)\left(\sum_{i,j=1}^{16} \bar{\mathbf{B}}_{\eta’}^{i,j} |e_j\rangle\langle e_i|\right)\right]\\
&= \frac{1}{16} \text{Tr}\left(\rho_{in}\sum_{a,b=1}^{16} (A_a \otimes B_b) \bar{\mathbf{B}}_{\eta’}^{a,b}\right) & \text{(by linearity \& cyclicity of trace)}\\
&= \frac{1}{16}\left(\eta^4 \text{Tr}[(\rho_{\text{EPR}}^{\otimes 2})\sum_{a,b=1}^{16} (A_a \otimes B_b) \bar{\mathbf{B}}_{\eta'}^{a,b}] + (1-\eta^2)^2\text{Tr}[(\frac{\mathbb{I}_4^{\otimes 2}}{16})\sum_{a,b=1}^{16} (A_a \otimes B_b) \bar{\mathbf{B}}_{\eta'}^{a,b}] \right. \\
&\quad \left. + \eta^2(1-\eta^2)\text{Tr}[(\rho_{\text{EPR}} \otimes \frac{\mathbb{I}_4}{4} + \frac{\mathbb{I}_4}{4} \otimes \rho_{\text{EPR}})\sum_{a,b=1}^{16} (A_a \otimes B_b) \bar{\mathbf{B}}_{\eta'}^{a,b}]\right)\\
&\coloneqq \eta^4 k_\text{ideal} + \eta^2(1-\eta^2)k_\text{part-mixed} + (1-\eta^2)^2 k_\text{mixed},
\end{align*}}
where intermediate steps are shown in App. \ref{app:2-chsh-opt-score-deriv}. Therefore, for $G \in \{\text{2-CHSH}, \text{MSG}, \text{2-CHSH-OPT}\}$, $\Delta_G = \eta^4 k_\text{ideal} + \eta^2(1-\eta^2)k_\text{part-mixed} + (1-\eta^2)^2 k_\text{mixed} - \omega_{c_G}$ for constants $k_\text{ideal}$, $k_\text{part-mixed}$, and $k_\text{mixed}$ specific to each $G$. As desired.
\end{proof}

Note that: (1) \(k_{\text{ideal}}\) is the probability that both EPR pairs are noise-free (i.e., the quantum bound of \(G\) with the maximally entangled state); (2) \(k_{\text{part-mixed}}\) is the probability that exactly one pair is noisy; and (3) \(k_{\text{mixed}}\) is the probability that both pairs are noisy (the score with the maximally mixed state). Following the method in the proof, we obtain:
    \begin{align}
        \omega_{\eta_{2chsh}} &= 0.72855\eta^4  + 
    0.85355\eta^2(1-\eta^2) +
    0.25000(1-\eta^2)^2 \\
        \omega_{\eta_{2chsh,\eta'=1}} &= 0.10937\eta^4  + 0.30936\eta^2(1-\eta^2) + 0.21875(1-\eta^2)^2\\
        \omega_{\eta_{msg}} &= \eta^4 + 1.22222\eta^2(1-\eta^2) + 0.5(1-\eta^2)^2 
    \end{align}
where we show the 2-CHSH optimized for $\eta'=1$ as an example. Recall that $\eta'$ is the visibility used to optimize and fix a 2-CHSH-OPT's Bell matrix, and $\eta$ is the visibility of the input used to play the now fixed configuration. For CHSH, we compute its value independently for input $\varepsilon_{\eta}^{2-qubit}(\rho_{EPR})=\eta^2 \rho_{\text{EPR}} + \frac{\mathbb{I} _4}{4}(1-\eta^2)$:
\begin{align}
\omega_{\eta_{chsh}} = \frac{1}{8}\text{Tr}(\mathbf{S}_{G_{chsh}} \varepsilon_{\eta}^{2-qubit}(\rho_{EPR})) + \frac{1}{2} 
    = \eta^2 \frac{\text{Tr}(\mathbf{S}_{G_{chsh}} \rho_{\text{EPR}})}{8} +  0.25 (1-\eta^2) \frac{\text{Tr}(\mathbf{S}_{G_{chsh}})}{8}  + \frac{1}{2} 
    = \frac{\sqrt{2}}{4} \eta^2  + \frac{1}{2},
\end{align}
where the last equality follows by definition of $\mathbf{S}_{G_{chsh}}$ (the expectation value that $\mathbf{S}_{G_{chsh}}$ yields given either the EPR pair, $2\sqrt{2}$, or maximally mixed state as an input state, $0$). Hence, we can observe from the final equation that the quantum value of the CHSH is 0.85355, as expected. Note that the CHSH value also has coefficients corresponding to the probabilities of maximally entangled and maximally mixed input states.
The full $\Delta_G$ expressions for all games considered are shown in Table \ref{tab:gapped-expressions-new}. Further, App. \ref{app:A} contains the raw $\omega_{\eta_G}$ unsimplified expressions showing the explicit $k$ constants.\\

\ifthenelse{\boolean{showfigures}}{
\begin{table}[htbp] 
    \centering
    \caption{Gapped Expressions for Non-Local Games}
    \label{tab:gapped-expressions-new}
    \begin{tabular}{
        >{\raggedright\arraybackslash}p{0.45\linewidth}  
        >{\raggedright\arraybackslash}p{0.45\linewidth}  
    }
    \toprule
    \textbf{$G$} & \textbf{$\Delta_G = \omega_{\eta_G} - \omega_{c_G}$} \\
    \midrule
    CHSH & $0.35355 \eta^2  +0.5 - 0.75$ \\ 
    
    \addlinespace[0.5em]
    2-CHSH & $0.35355\eta^2 + 0.12500\eta^4  + 0.25000 - 0.625$\\
    
    \addlinespace[0.5em]
    2-CHSH-OPT $\eta'=1$ & $0.30936 \eta^2 + 0.10937 \eta^4 + 0.21875 - 0.546875$\\
    
    \addlinespace[0.5em]
    2-CHSH-OPT $\eta'=0.95$ & $0.25411 \eta^2 -0.00097 \eta^4 + 0.36133 - 0.5625$ \\
    
    \addlinespace[0.5em]
    2-CHSH-OPT $\eta'=0.90$ & $0.25411 \eta^2 -0.00097 \eta^4 + 0.36133 - 0.546875$\\
    
    \addlinespace[0.5em]
    2-CHSH-OPT $\eta'=0.89$ &  $0.24307 \eta^2 + 0.34375 - 0.515625$\\
    
    \addlinespace[0.5em]
    2-CHSH-OPT $\eta'=0.88$ &  $0.24307 \eta^2 + 0.34375 - 0.515625$\\
    \addlinespace[0.5em]
    2-CHSH-OPT $\eta'=0.87$ &  $0.24307 \eta^2 + 0.34375 - 0.515625$\\
    \addlinespace[0.5em]
    2-CHSH-OPT $\eta'=0.86$ &  $0.24307 \eta^2 + 0.34375 - 0.515625$\\
    \addlinespace[0.5em]
    2-CHSH-OPT $\eta'=0.85$ &  $0.24307 \eta^2 + 0.34375 - 0.515625$\\
    
    \addlinespace[0.5em]
    2-CHSH-OPT $\eta'=0.84$ &  $0.53125 - 0.53125$\\
    \addlinespace[0.5em]
    2-CHSH-OPT $\eta'=0.83$ &  $0.53125 - 0.53125$\\
    
    \addlinespace[0.5em]
    MSG & $0.22222 \eta^2 + 0.27778 \eta^4 + 0.5 - 0.94444$\\
    \bottomrule
    \end{tabular}
    \begin{flushleft}
    \end{flushleft}
    \vspace{-1.5em}
\end{table}}{}
We note that for 2-CHSH-OPT configurations, the $\eta'$ values were selected as the results were plotted, having a more fine-grained analysis between $0.83-0.89$ as this region had some of the most convincing configurations, and eventually became unconvincing for $\eta'=0.83,0.84$. Also note that throughout the rest of the work, the $\Delta_G$ expressions are re-written with explicit $\eta$-dependence, as in Table \ref{tab:gapped-expressions-new}, and are referred to as both \textit{gapped expressions} and \textit{noise-dependent gaps.}  
\subsection{Gap-Based Scoring}\label{sec:gapped-score}
Using the gapped expressions in Table \ref{tab:gapped-expressions-new}, we explored multiple approaches for constructing a single numerical score for a non-local game, such that the score is a proxy measure for its convincingness in comparison to other games. In particular, we searched for a score which reflects the following intuition: if game $G_1$ scores higher than game $G_2$, then $G_1$ should become statistically convincing at a lower $\eta$ value than $G_2$. This aligns with our goal of identifying games that are more robust to experimental imperfections while still providing strong evidence of non-locality, as a higher score for $G_1$ would imply that $G_1$ remains convincing at higher noise levels. This can be formalized as follows: 
\begin{definition}[Significance Crossing]\label{def:sig-crossing}
    For a game $G$, let $\eta^\dagger_G = \inf{\eta \in [0,1] : \mathfrak{C}_G(\eta) \leq \alpha}$ be the visibility at which the game first becomes statistically significant at level $\alpha$, and refer to $\eta_G^\dagger$ as $G$'s `significance crossing'. 
\end{definition}
\begin{question}\label{question:gapped-score-behaviour}
Can we define a metric $\kappa_G$ to rank games such that for any two games $G_1$ and $G_2$:
$\kappa_{G_1} > \kappa_{G_2} \implies \eta_{G_1}^\dagger < \eta_{G_2}^\dagger \pm \epsilon$?
 $\epsilon$ is a small tolerance parameter accounting for numerical precision and cases where the convincingness curves $\mathfrak{C}_{G_1}(\eta)$ and $\mathfrak{C}_{G_2}(\eta)$ are nearly identical near their crossing points. 
\end{question}

Ideally, we want to use $\kappa_G$ to correctly rank the order of the significance crossings for different games, which would, importantly, rank the games' convincing abilities under our noise model. We construct a new score which we coin the \textit{gapped score} (Def. \ref{def:quad-score-dep}), inspired by convincingness, for modeling the above behaviour. As will be shown in Sec. \ref{sec:results-disc}, out of the existing ratio and gap-inspired scores explored, Def. \ref{def:quad-score-dep} best captures the desired behaviour of $\kappa_G$. The optimal candidate for $\kappa_G$, is presented here, which is based on noise-dependent gaps.
\begin{definition} \label{def:quad-score-dep}
    [Quadratic Noisy Gapped Ratio Scoring: Depolarizing Channel] For a non-local game $G$ whose input undergoes the depolarizing channel with probability of no noise $\eta$, $G$ has the \textbf{gapped expression}: \[\omega_{\eta_G} - \omega_{c_G} = c_1 \eta^2 + c_2\eta^4 + d - \omega_{c_G}.\] Then we define the \textbf{gapped score} of $G$, played with $N_{res}$ $\eta$-noisy inputs as:
\begin{equation}
\kappa_G \coloneqq \left(\frac{c_1 + c_2}{|d-\omega_{c_G}|}\right)^2 \cdot \frac{N_{res}}{-\ln(\alpha)}
\end{equation}
where the parameter $\alpha$ defines the significance threshold.
\end{definition}

\noindent Note that in contrast to the convincingness measure, which must be evaluated separately for each noise level, the gapped score provides a single number that characterizes a game's overall convincingness behaviour for all possible noise levels.

We now explain the intuition behind the design of the gapped score. The score in Def. \ref{def:quad-score-dep} captures essential features that determine when a game becomes statistically significant for non-locality. The first factor, $\left(\frac{c_1 + c_2}{|d-\omega_{c_G}|}\right)^2$, captures the relative strength of the noise-dependent terms versus the constant separation in the gap expression:
\begin{itemize}
    \item The numerator, $c_1+c_2$, represents the ``strength'' of the $\eta$-dependent terms. We use $c_1+c_2$  rather than a weighting of the terms to capture the full contribution of the noise-dependent terms when $\eta=1$ (i.e., when there is no noise), ensuring that the score reflects the maximum possible rate of gap growth under ideal conditions.
    \item The denominator $|d - \omega_{c_G}|$ represents the ``baseline separation'' that needs to be overcome. 
\end{itemize}
The ratio between these terms indicates how quickly the gap expression can cross the significance threshold as $\eta$ increases. It enforces that games with higher scores have lower $\eta^\dagger$, indicating greater robustness to noise. Moreover, the quadratic relationship amplifies the differences between gap expressions, providing a more informative score. We do find that enforcing a linear relationship instead is also characteristic of the games' relative $\eta^\dagger$'s, albeit, yielding less distinct scores.

Further, the scaling factor, $\frac{N_{res}}{-\ln(\alpha)}$, relates our score to the experimental conditions needed for statistical significance. To understand where this comes from, consider $\Delta_G$ becomes significant when its $p$-value, which we model by $p \sim \exp(-N_{res} (\omega_{\eta_G} - \omega_{c_G})^2)$, falls below $\alpha$. Taking the negative logarithm of both sides of $\exp(-N_{res} (\omega_{\eta_G} - \omega_{c_G})^2) < \alpha$, we find that significance occurs when $(\omega_{\eta_G} - \omega_{c_G})^2 > \frac{-\ln(\alpha)}{N_{res}}$. The scaling factor in our score is the inverse of this threshold, which aligns with intuition: more experimental resources ($N_{res}$) make it easier to achieve significance and thus increase the score, while stricter significance requirements (smaller $\alpha$) make it harder and thus decrease the score.

Putting it all together, higher scores indicate earlier crossing of the significance threshold. If $\mathfrak{C}_G$ crosses the threshold earlier, we expect a larger coefficient strength and smaller baseline separation yielding a higher $\kappa_G$. We also expect increased experimental resources to increase $\kappa_G$. Of course, a greater $\alpha$ will increase $\kappa_G$ as well. 

Note that the definition of $\kappa_G$ can be extended to a different noise model by deriving a gapped expression for that model. Only the degree of the polynomial, and the corresponding coefficients, $c_i$, and constant, $d$, will differ. Simply summing all of the new model's coefficients as is for the numerator may work, or a weighting may need to be derived to reliably capture the significance crossing behaviour of that noise model, by a similar process to that presented in Sec. \ref{sec:results-disc}.

\subsection{Polynomial Comparisons} \label{sec:poly-comparison}
As will be seen in Section \ref{sec:results-disc}, the gapped score achieves the desired behaviour, but has a limitation that games with extremely similar convincingness behaviour have the same or similar scores. A more fine-grained analysis for games that are similar in convincingness behaviour can be done via subtracting the gapped expressions and analyzing the behaviour of the resulting curve. Let us demonstrate this approach by comparing the standard 2-CHSH and the ($\eta'=1$)-optimized 2-CHSH games. Define the difference function $D(\eta)
  \;:=\;
  \Delta_{2chsh}(\eta)
  \;-\;
  \Delta_{2chsh, \eta'=1}(\eta)$, where,
\begin{align*}
  \Delta_{2chsh}(\eta)
  &=
  0.35355\eta^2 + 0.12500\eta^4 -0.375,\\
  \Delta_{2chsh, \eta'=1}(\eta)
  &=
  0.30936 \eta^2 + 0.10937 \eta^4 -0.32813.
\end{align*}
This yields $D(\eta)
  =
  0.04419\eta^2 + 0.01563\eta^4 -0.04688.$ Now, the behaviour of $D(\eta)$ reveals several key insights. At maximal noise ($\eta=0$), we have $D(0)=-\,0.04688<0$, indicating that the  ($\eta'=1$)-optimized 2-CHSH has a larger gap. Conversely, in the noiseless limit ($\eta=1$), we find $D(1)= 0.04419 +0.01563 -0.04688 = +0.01294>0$, showing the standard 2-CHSH achieves a larger gap. 
To characterize the complete behaviour of $D(\eta)$, note that its derivative $D'(\eta) = 0.08838\eta + 0.06252\eta^3$ is strictly positive for all $\eta > 0$. Therefore, $D(\eta)$ is strictly increasing on $[0,1]$, and by the Intermediate Value Theorem, must cross zero exactly once on this domain. Numerically, this crossing occurs at $\eta\approx 0.91$, leading to the conclusion:
\[
  D(\eta)<0 \quad\text{for}\;\eta<0.91,
  \quad
  D(\eta)>0 \quad\text{for}\;\eta>0.91,
\]
These results reveal a fundamental trade-off: the $\eta'=1$ optimized 2-CHSH performs better (has a larger gap) in high-noise regimes ($\eta<0.91$), while the standard 2-CHSH achieves superior performance in low-noise regimes ($\eta>0.91$). This suggests that optimization of the two-copy measurement bases enhanced the detection of non-locality under moderate noise, at the cost of reduced violation magnitude at the noiseless limit. 

We do note that $D(0)$ and $D(1)$ have relatively small values, which indicates the two games have very similar convincingness behaviour--as expected, since the 2-CHSH-OPT game is optimized for no noise on the input state. As we will see, the above conclusions match exactly the behaviour observed in the simulation of convincingness in Sec. \ref{sec:results-disc} (Fig. \ref{fig:pvals-N=1000,N=10000}, left). Despite the fact that both curves for the exact convincingness score are nearly overlapping, the gap accurately describes the behaviour. Therefore, this polynomial comparison method complements the gapped score approach: while the gapped score identifies which games first cross the significance threshold, direct polynomial comparison reveals the detailed comparative behaviour of games and configurations which have highly overlapping convincingness curves.

\section{RESULTS AND DISCUSSION}\label{sec:results-disc}
We now discuss the results of the noise-robustness analysis, first summarizing the overall behaviours. The results of the convincingness analyses indicate that the gapped score is a reliable measure of the noise-robustness of a non-local game for \textit{stable} resource regions, to be defined below, and more accurately reflects the order of games' significance crossings ($\eta^\dagger_G$'s) compared to the ratio ($\omega_{{\eta=1}_G}/\omega_{c_G}$), the gap between the quantum bound and local bound ($\omega_{{\eta=1}_G} - \omega_{c_G}$), and all other new measures proposed in this work. Further, the noise-tolerance behaviour observed exactly matches the trends observed for the convincingness under near-infinite resources, suggesting that the convincingness is a more nuanced measure of noise-robustness than noise-tolerance. Note that simulations used Python \cite{van1995python}, MATLAB \cite{MATLAB} and QuTiP \cite{Johansson_2012}.

\subsection{Convincingness Analysis}

Figs. \ref{fig:pvals-N=1000,N=10000} and \ref{fig:pvals-N=100K,N=1M} exhibit the results for the convincingness analysis across depolarized noise levels of the EPR pair(s). Let us recall the setup  (from Sec. \ref{sec:analysis-across-noise-levels}). Each non-local game is played at different visibilities, $\eta$, with input $\rho_{\eta} = \varepsilon^{2-qubit}_\eta(\rho_{EPR})$ for CHSH and $\rho_{\eta} = \varepsilon^{4-qubit}_\eta(\rho_{EPR}\otimes\rho_{EPR})$ for the MSG and all 2-CHSH configurations, yielding a $\omega_{\eta_G}$ for each round. Algorithm $\ref{algo:convincingness}$ is then used to evaluate each $\mathfrak{C}_G(\eta) \;\; \forall \eta\in [0,1]$, using $N_{res}$ noisy EPR pairs for each game. This translates to $N_{res}$ and $\frac{N_{res}}{2}$ rounds played for CHSH and the rest of the games, respectively. The curves $\{(\eta, \mathfrak{C}_G) \;\; \forall \eta\in [0,1]\}$ are shown in Figs. \ref{fig:pvals-N=1000,N=10000} and \ref{fig:pvals-N=100K,N=1M} for different $N_{res}$ values. The 2-CHSH-OPT configurations were optimized for particular visibilities ($\eta'\in\{0.85, 0.86, \dots, 1\}$), using the method in Sec. \ref{sec:ptimization-for-specific-noise-regimes} and the curves for these games are also plotted. Note that $\Delta_G =0$ for 0.83- and 0.84-optimized games, so they are completely unconvincing configurations and hence excluded from figures. 

Let us now discuss the results from the convincingness simulations. One of the most important  properties on the figures is the visibility levels at which $\mathfrak{C}_G = \alpha$, the significance crossings  $\eta^\dagger_G$ (\ref{def:sig-crossing}). Note that our $\alpha=0.05$. The $\eta^\dagger_G$ crossings are the visibility levels for which a game (or configuration), $G$, becomes convincing of non-locality, or in other words, will be highly likely to detect non-locality, for the given input state  $\rho_{\eta^\dagger_G}$. Based on analyzing the order of the crossings of different games, we can deduce a number of major observations: 

\begin{observation}\label{obs:CHSH-best-constant-Nres}
Let $\{\text{2-CHSH}, \text{MSG}, \text{2-CHSH-OPT configurations}\}$ be the set of candidate games and configurations being tested for noise-robustness, where lower $\eta^\dagger_G$ implies greater `noise-robustness'. Then, CHSH is the most noise-robust candidate to the depolarizing channel, assuming an equal number of noisy resources is available to all candidates.   
\end{observation}

Let us discuss Observation \ref{obs:CHSH-best-constant-Nres}. If we look at Figs. \ref{fig:pvals-N=1000,N=10000} and \ref{fig:pvals-N=100K,N=1M}, where the $N_{res}$ used by each game to prove non-locality is the same (i.e., $N_{res} = 1000, 10000,100000, \text{ or } 1000000$), then the CHSH has the lowest $\eta^\dagger$ out of all candidates. That is, \[\eta^\dagger_{chsh} < \eta^\dagger_F \; \forall F \in \{\text{2-CHSH}, \text{MSG}, \text{2-CHSH-OPT configurations}\}.\] Moreover, consider the 2-CHSH games optimized for $\eta' \in \mathcal{Q} = \{0.85, 0.86, \dots,0.89\}$, which have approximately the same $\eta^\dagger$ as each other. A pattern is seen across all experiments. As $N_{res}$ grows, the $\mathfrak{C}$ curves of these particular 2-CHSH optimized games grow closest to $\eta^\dagger_{chsh}$ out of all other games, but are unable to beat it. That is, as $N_{res} \rightarrow \infty$, we see:
{\small{
\begin{align} \label{eqn:observation-2chshopt-bounded-chsh}
\eta^\dagger_{2chsh,0.85} \approx \eta^\dagger_{2chsh,0.86} \approx \eta^\dagger_{2chsh,0.87}\approx \eta^\dagger_{2chsh,0.88}\approx \eta^\dagger_{2chsh,0.89} \rightarrow \eta^\dagger_{chsh}.
\end{align}
}}
Note that we tried optimizing the 2-CHSH game for smaller visibilities ($\eta=0.83, 0.84$), which yielded unconvincing behaviour ($\Delta_{2chsh,\eta} = 0$). Hence, for the purpose of this analysis we can assume that 2-CHSH optimized for $\mathcal{Q}$ are the best configurations of the game for the depolarizing channel. Based on this and Eqn. \ref{eqn:observation-2chshopt-bounded-chsh}, we question whether this trend is revealing of a fundamental limitation in the axes fixed for the 2-CHSH optimized non-local games. One possible limitation 

\newpage
\ifthenelse{\boolean{showfigures}}{\begin{figure}[H]
    \centering
    \includegraphics[width=\textwidth]{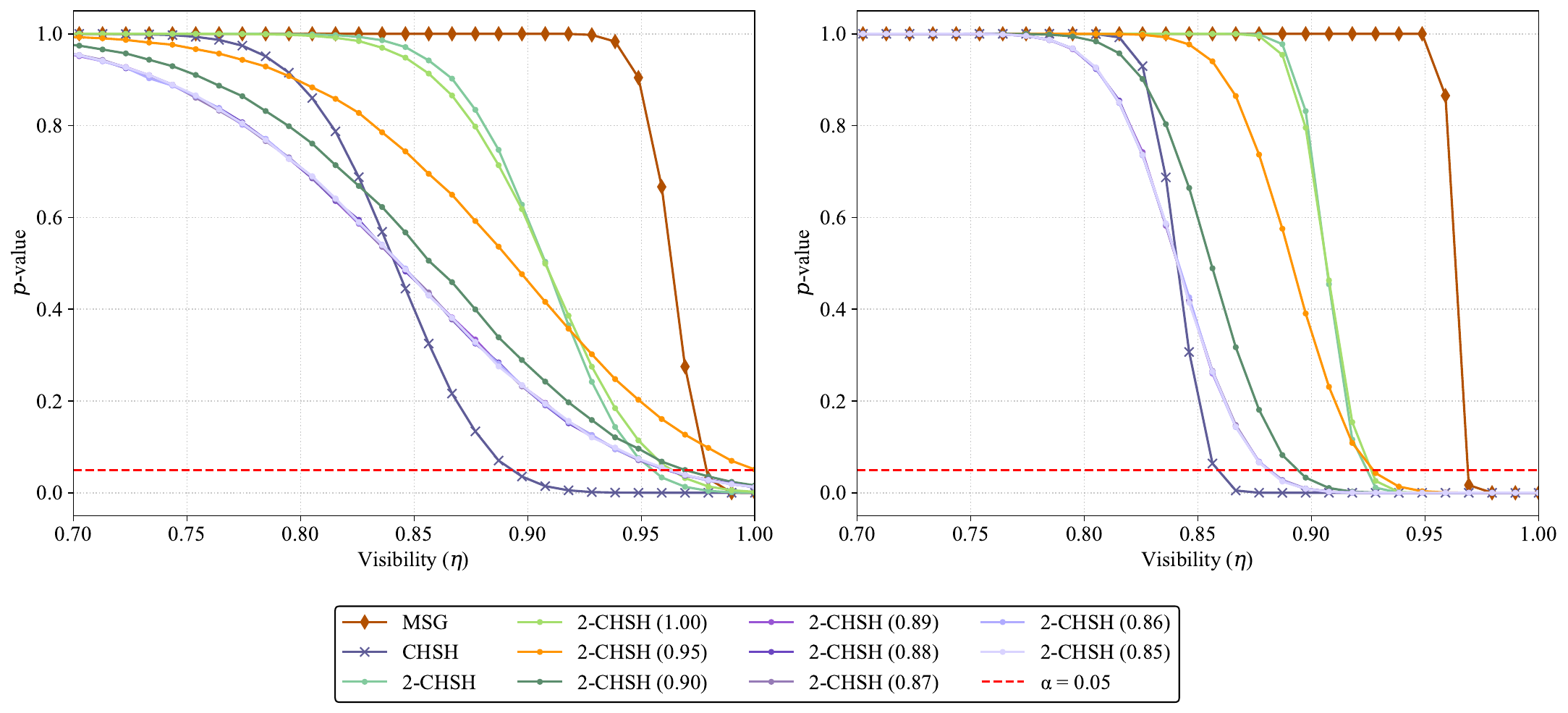}
    \caption[Convincingness Behaviour Analysis for Finite Resources (low $N_{res}$).]{\textbf{Convincingness Behaviour Analysis for Finite Resources (low $N_{res}$).} The convincingness curves were computed for the CHSH, MSG, 2-CHSH, and 2-CHSH-OPT games using the methods proposed in Sec. \ref{sec:comp-methods}. Each point on a curve was computed assuming $N_{res}$ $\eta$-noisy EPR pairs are available to that game. The $\eta$-level for which 2-CHSH-OPT games were optimized is noted in brackets in the legend. Games which have similar $\kappa_G$, computed according to the methods in Sec. \ref{sec:gapped-score}, are coloured similarly. The dotted red line denotes the significance threshold, $\alpha=0.05$. The left figure shows the $\mathfrak{C}_G(\eta)$ scores ($p$-values), assuming $N_{res}=1000$ are available for all games. The right figure shows the same, but assumes $N_{res}=10,000$. All plotted values were averaged over 10 runs, each with a random seed, where the $p$-value computed for each random seed is an average of three runs with this seed.} 
    \label{fig:pvals-N=1000,N=10000}
\end{figure}}{}
\vspace{2em}

\ifthenelse{\boolean{showfigures}}{\begin{figure}[H]
    \centering
    \includegraphics[width=\textwidth]{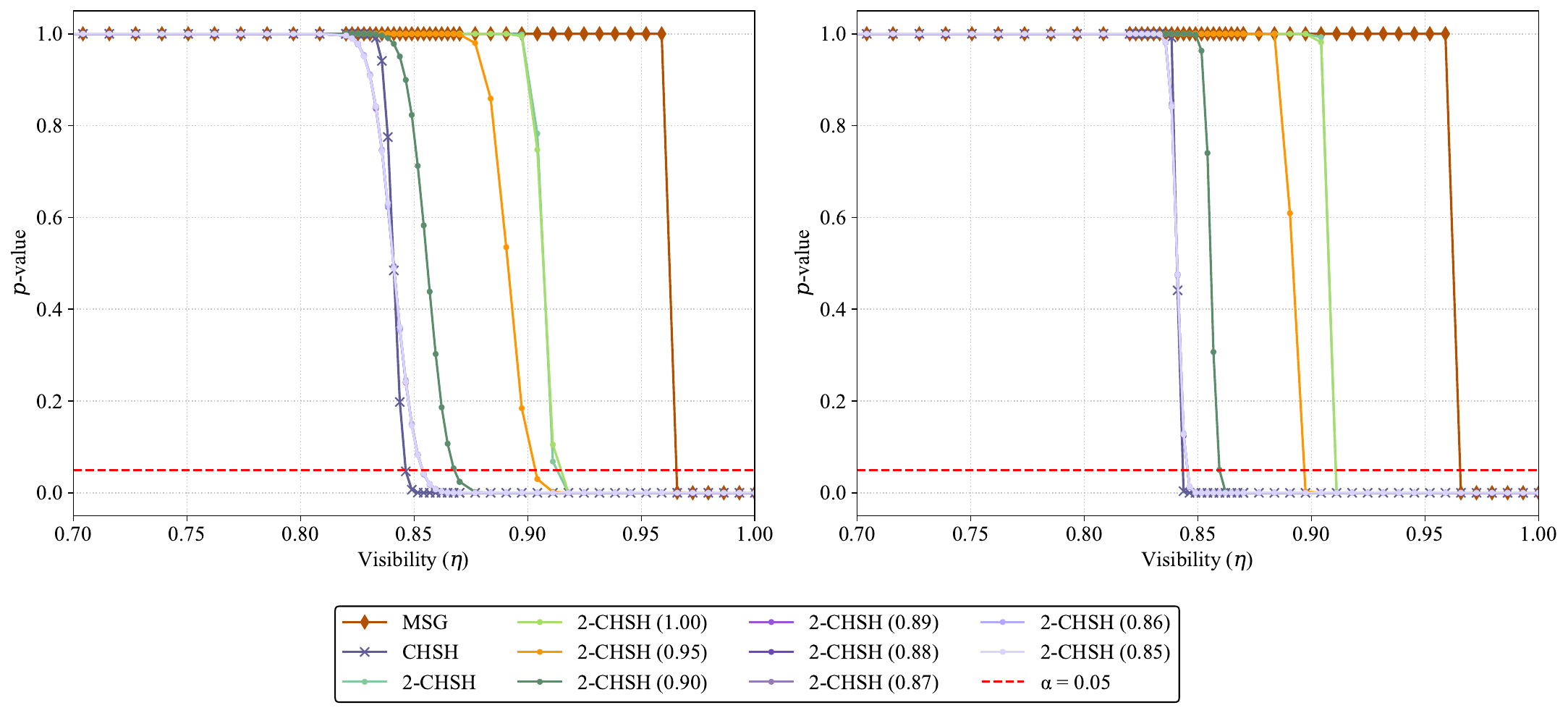}
    \caption[Convincingness Behaviour Analysis for Near-Infinite Resources (high $N_{res}$).]{\textbf{Convincingness Behaviour Analysis for Near-Infinite Resources (high $N_{res}$).} The convincingness was computed the same way as in Figure \ref{fig:pvals-N=1000,N=10000}, with the exception that left figure shows the convincingness scores for when $N_{res}=100,000$ for all games, and the right shows the same for when $N_{res}=1,000,000$.}
    \label{fig:pvals-N=100K,N=1M}
\end{figure}}{}

\newpage 
\noindent is the restriction on both games to use the same amount of resources in order to prove the presence of non-locality; perhaps the 2-CHSH-OPT games require more rounds, and hence, more resources than the CHSH game in order to do so.  We, thus, test whether allowing the 2-CHSH-OPT configurations corresponding to $\mathcal{Q}$ to use more resources than the CHSH permits their $\eta^\dagger_{2chsh,opt}$ values to move below $\eta^\dagger_{chsh}$, and by Fig. \ref{fig:pvals-unequal-resources}, observe the following positive result:

\begin{observation}\label{obs:CHSH-best-different-Nres}
Let us define noise-robustness as before, where a game $G$ is more `noise-robust' than game $G'$ if $\eta^\dagger_G < \eta^\dagger_{G'}$. Let $F'$ denote the set of 2-CHSH games optimized for $\eta' \in \mathcal{Q}$, and denote $N_{res}^{G \in F'}$ and $N^{chsh}_{res}$ the resources each game can use to prove non-locality. $\forall G\in F'$, if $N_{res}^{G \in F'} >> N^{chsh}_{res}$, then the configurations in $F'$ are more noise-robust than CHSH.
\end{observation}

\ifthenelse{\boolean{showfigures}}{\begin{figure}[H]
    \centering
    \includegraphics[width=\textwidth]{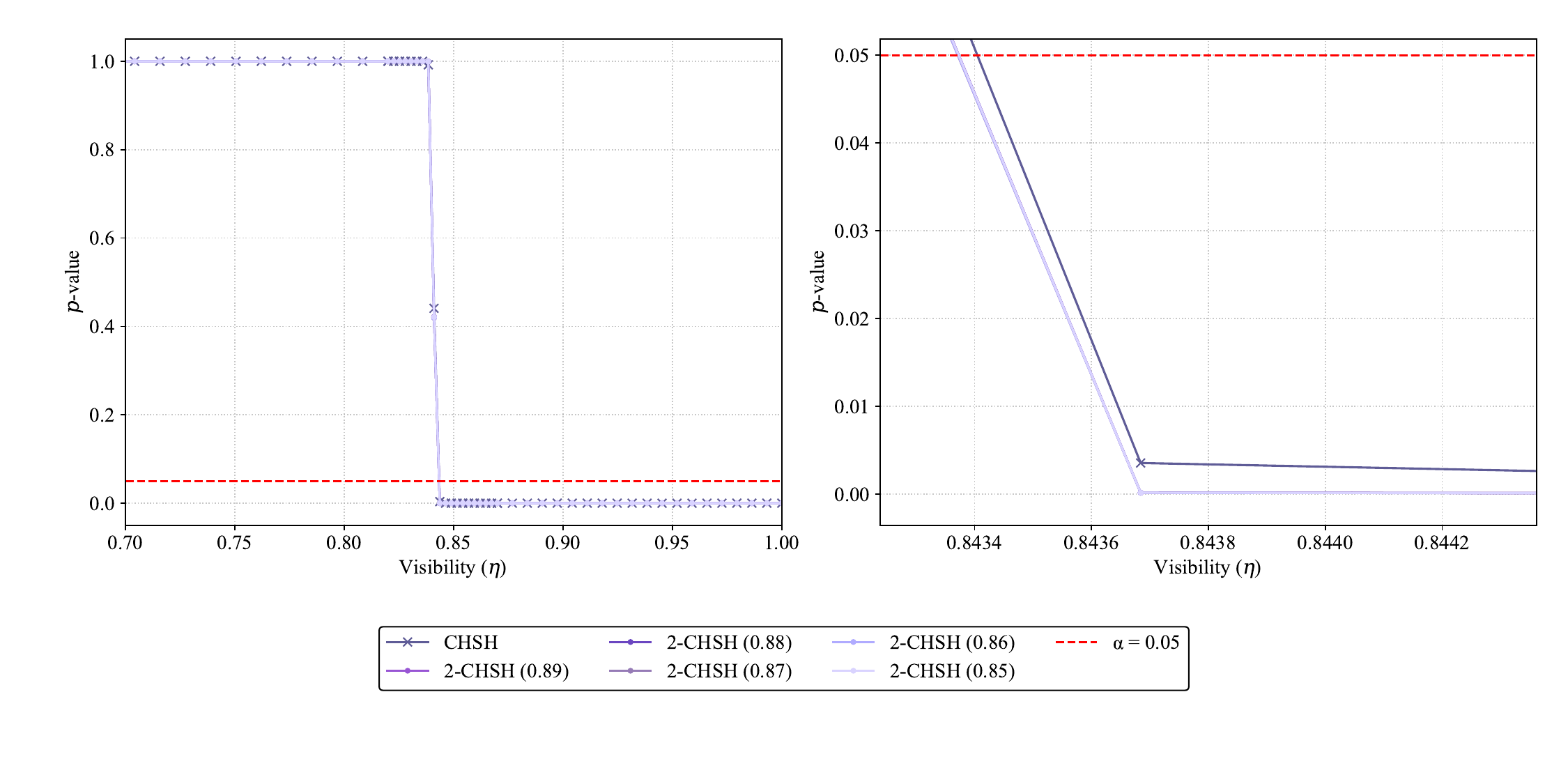}\caption[Convincingness Behaviour Analysis for Near-Infinite, Unequal Resources Between CHSH and 2-CHSH-OPT Configurations.]{\textbf{Convincingness Behaviour Analysis for Near-Infinite, Unequal Resources Between CHSH and 2-CHSH-OPT Configurations.} The right figure is a zoom-in of the left figure. We took the most convincing 2-CHSH-OPT games, and permitted them to have much more noisy resources (many more rounds) than the CHSH in order to see if the 2-CHSH-OPT's convincingness was fundamentally bounded by the CHSH's convincingness. When $N_{res}= 1,000,000$  for the CHSH and $N_{res}= 10 \times 1,000,000$ for the 2-CHSH-OPT configurations, the 2-CHSH-OPT configurations outperform the CHSH. That is, $\eta^\dagger_{2chsh,\eta'} < \eta^\dagger_{chsh}$ for $\eta' \in \mathcal{Q}$. All simulations were conducted as before and once again, all plotted values were averaged over 10 runs, each with a random seed, where the $p$-value computed for each random seed is an average of three runs with this seed. The dotted red line denotes the significance threshold, $\alpha=0.05$.} 
    \label{fig:pvals-unequal-resources}
\end{figure}}{}

Through varying the relative resources of the two in the simulations, and testing the asymptotic resource regime with resources in the millions of noisy EPR pairs, we find that when the CHSH game has $N_{res} = 1M$ resources to use and the 2-CHSH-OPT configurations optimized for $\eta'\in \mathcal{Q}$ have $10$ times that amount to use, $\eta^\dagger_{2chsh,\eta'} < \eta^\dagger_{chsh}$ (Fig. \ref{fig:pvals-unequal-resources}). This  observation is a counter-example to the hypothesis that Eqn. \ref{eqn:observation-2chshopt-bounded-chsh} holds across all noisy resource amounts, and suggests that 2-CHSH-OPT requires more resources to be robust to $\eta^\dagger_{chsh}$ visibility than the CHSH. It is valid to question whether a possible reason for these results is some error bound in the simulation, potentially stemming from the Monte Carlo seeds used to randomly select the number of successes in the $p$-value computation. However, since our figures plot the average over many random seeds, this is unlikely. Moreover, the finding aligns with our expectations for the 2-CHSH-OPT game: we expect that the 2-CHSH-OPT would perform better or the same as the CHSH either due to parallel repetition through the use of simultaneous measurements on a greater Hilbert space (to increase the probability of detecting non-locality), or due to the fact that an optimized version of the 2-CHSH game can choose to use only one pair of qubits and reach, at the very least, performance of the CHSH. Our results also align with the work of Arnon-Friedman and Yuen \cite{arnon2017noise}, whose parallel‑repetition DI protocol needs on the order of $10^6$ parallel CHSH rounds before it can certify that the high-dimensional input state---built from many noisy two‑qubit copies---contains a large amount of entanglement.
Our finding, therefore, motivates future study into exactly why substantially higher \textit{noisy} resources are needed in parallel repetition non-local game schemes in order to certify non-locality. A possible explanation is that using multiple noisy states per round fundamentally limits a game's performance, but a rigorous theoretical analysis remains open for future work. 

We also note another interesting finding. In the finite resource regime there are specific choices of significance thresholds, $\alpha$ (target $p$-values), and number of EPR pairs (per game round) where the CHSH game is not the most noise-robust.  For example, in the $N_{res}=1000$ experiment (Fig. \ref{fig:pvals-N=1000,N=10000}, left), we observe that $\eta^\dagger_{2chsh, i\in\mathcal{Q}\cup\{0.90,0.95\}} < \eta^\dagger_{chsh}$  for lower visibilities when $\alpha \approx 0.5$. While 0.5 is a very lenient threshold and leaves room for potential false positive detections of non-locality (likely ill-advised), these observations show that the convincingness of games is much more nuanced for low, more practical resource levels. It also demonstrates the importance of the experimenter's choice of certainty $\alpha$ and prompts future study into the impact of choosing higher $\alpha$'s on comparing the noise-robustness of games. It may be the case that the $\alpha=0.05$ we use is a more stringent condition than practically necessary. More generally:

\begin{observation}
    For a low (finite) $\eta$-noisy $N_{res}$ supply  available to each game, a more lenient $\alpha$ choice can enable the noise-robustness of some games to improve beyond that of CHSH; where we define noise-robustness as before by $\eta^\dagger_G$.
\end{observation}

By all of the above observations, we note that our original objective has been attained. We have shown a clear, computational method, for determining which non-local game, in a set of candidate games with different input-output settings and winning state dimensionalities, is most likely to detect non-locality in noisy (depolarized) states, under different resource or certainty conditions ($N_{res}$ available, or certainty of  non-locality, $\alpha$, decided by the experimenter). While the above methods heavily rely on simulations, we now discuss a more analytic method to arrive at similar conclusions. Specifically, since the $\eta^\dagger_G$ values can be used to compare the noise-robustness of games, in the convincingness-sense of noise-robustness, our goal is to find a closed form analytic expression for a score which best matches the order of $\eta^\dagger_G$'s across games. In particular, the candidate predictors which we consider are: the ratio  ($\omega_{{\eta=1}_G}/\omega_{c_G}$), the raw gap ($\omega_{{\eta=1}_G} - \omega_{c_G}$), the coefficients and constant in $\Delta_G$ and $\Delta_G^2$, and the gapped score ($\kappa_G$) derived in Sec. \ref{sec:analytic-methods}. All of these predictors were derived from the analytic expression of the gap, $\Delta_G$, since, as we know from Eqn. \ref{eqn:pval-simplified}, the gap heavily influences the convincingness and hence the order of significance crossings. The values of each predictor can be seen in Table \ref{tab:game_comparison}. \\

\ifthenelse{\boolean{showfigures}}{
\begin{table}[h!]
    \centering
    \caption{Comparison of Non-Local Games and Their Properties.}
    \resizebox{\textwidth}{!}{
        \begin{tabular}{l|c|c|c|c|c|c|c|c|c|c|c}
            \hline
            \textbf{Property} & \textbf{CHSH} & \textbf{2-CHSH} & \multicolumn{8}{c|}{\textbf{2-CHSH OPT}} & \textbf{MSG} \\
            \cline{4-11}
            & & & $\eta=0.89$ & $\eta=0.88$ & $\eta=0.87$ & $\eta=0.86$ & $\eta=0.85$ & $\eta=0.90$ & $\eta=1$ & $\eta=0.95$ & \\
            \hline
            Local Bound ($\omega_{c_G}$) & 0.75 & 0.625 & 0.516 & 0.516 & 0.516 & 0.516 & 0.516 & 0.547 & 0.547 & 0.562 & 0.944 \\
            Quantum Bound ($\omega_{{\eta=1}_G}$) & 0.85 & 0.729 & 0.587 & 0.587 & 0.587 & 0.587 & 0.587 & 0.614 & 0.637 & 0.614 & 1 \\
            Score ($\kappa_G, N_{res}=1K$) & \textbf{667.60} & \textbf{543.61} & \textbf{667.59} & \textbf{667.59} & \textbf{667.59} & \textbf{667.59} & \textbf{667.59} & \textbf{621.33} & \textbf{543.59} & \textbf{528.56} & \textbf{43.26} \\
            Score ($\kappa_G, N_{res}=10K$) & \textbf{6676.04} & \textbf{5436.12} & \textbf{6675.89} & \textbf{6675.89} & \textbf{6675.89} & \textbf{6675.89} & \textbf{6675.89} & \textbf{6308.87} & \textbf{5435.92} & \textbf{5285.57} & \textbf{432.62} \\
            Ratio ($\omega_{{\eta=1}_G}/\omega_{c_G}$) & 1.133 & 1.166 & 1.138 & 1.138 & 1.138 & 1.138 & 1.138 & 1.123 & 1.165 & 1.092 & 1.059 \\
            Raw Gap ($\omega_{{\eta=1}_G}-\omega_{c_G}$) & 0.1 & 0.104 & 0.071 & 0.071 & 0.071 & 0.071 & 0.071 & 0.067 & 0.09 & 0.051 & 0.056 \\
            Constant in $\Delta_G^2$ & 0.062 & 0.141 & 0.03 & 0.03 & 0.03 & 0.03 & 0.03 & 0.034 & 0.108 & 0.04 & 0.198 \\
            $\eta^2$ Coefficient in $\Delta_G^2$ & -0.177 & -0.265 & -0.084 & -0.084 & -0.084 & -0.084 & -0.084 & -0.094 & -0.203 & -0.102 & -0.198 \\
            $\eta^4$ Coefficient in $\Delta_G^2$ & 0.125 & 0.031 & 0.059 & 0.059 & 0.059 & 0.059 & 0.059 & 0.065 & 0.024 & 0.065 & -0.198 \\
            Constant in $\Delta_G$ & -0.25 & -0.375 & -0.172 & -0.172 & -0.172 & -0.172 & -0.172 & -0.186 & -0.328 & -0.201 & -1.389 \\
            $\eta^2$ Coefficient in $\Delta_G$ & 0.354 & 0.354 & 0.243 & 0.243 & 0.243 & 0.243 & 0.243 & 0.254 & 0.309 & 0.254 & 0.222 \\
            $\eta^4$ Coefficient in $\Delta_G$ & 0 & 0.125 & 0 & 0 & 0 & 0 & 0 & -0.001 & 0.109 & -0.001 & 0.278 \\
            \hline
            Sig. Crossing ($\eta^\dagger_G$, $N_{res}= 1K$)  & 0.107 & 0.046 & 0.04 & 0.038 & 0.038 & 0.038 & 0.04 & 0.03 & 0.037 & N/A & 0.022 \\
            Sig. Crossing ($\eta^\dagger_G$, $N_{res} = 10K$) & 0.142 & 0.078 & 0.12 & 0.12 & 0.12 & 0.119 & 0.12 & 0.107 & 0.075 & 0.073 & 0.032 \\
            \hline
        \end{tabular}
    }
    \label{tab:game_comparison}
    \begin{flushleft}
    \small
    \textbf{Note:} This table compares various candidate measures for approximating the order of $\eta^\dagger_G$ across games (at $\alpha=0.05$). 
    \end{flushleft}
\end{table}}{}

\noindent Prior to getting into the analysis of these values, let us summarize our main observations:

\begin{observation}\label{obs:unstable-crossing}
    For an unstable $N_{res}$ region, the order of the $\eta^\dagger_G$ crossings is not predicted reliably by the candidate scores and should be computed using the exact convincingness simulation. 
\end{observation}

\begin{observation}\label{obs:stable-crossing}
    For a stable $N_{res}$ region, the order of the $\eta^\dagger_G$ crossings is reliably captured by the gapped score, supplemented by the polynomial expression analysis for games or configurations which have the same scores. Moreover, the gapped score is the best crossing predictor in this resource region out of all candidates. 
\end{observation}

To arrive at these observations, it is important to understand what is meant by a $stable$ and $unstable$ $N_{res}$ region. In Fig. \ref{fig:crossings-phase-changes}, we plot the order of the crossings that is attained when all games and configurations are given access to $N_{res}$ $\eta$-noisy resources. Each curve shows an experiment where all games had access to $N_{res}$ resources, and we particularly plot $1-\eta^\dagger$ in order to see the trend more easily. For low $N_{res}$, it is observed that the order of the crossings is varied, or \textit{unstable}. At $N_{res} = 500$, only four out of the 11 games cross the significance threshold, and at $N_{res}=1000$ ten out of the 11 games cross the threshold. However, for mid to high $N_{res}$, starting at $N_{res} = 2000$,  the ordering of the crossings starts to take on a consistent, $stable$, pattern. Lastly, we also make the observation that as $N_{res}$ increases, the values of the crossings themselves (the $y$-axis) increases as well, plateauing eventually for very high $N_{res}$ ($O(1M)$). Hence, we conclude that the order of the crossings, which is the pattern we wish to predict with a candidate measure, has unstable and stable regions, and that each crossing ($1-\eta^\dagger$) for a game appears to have a maximum value. Note that these observations make sense, since by Eqn. \ref{eqn:pval-simplified}, for large $N_{res}$, \textit{any} $\Delta_G > 0$ will eventually produce significance (as captured by the $c_i$ coefficients), whereas for small $N_{res}$ the absolute magnitude of $\Delta_G$ critically determines the convincingness, making its prediction more challenging.\\

\ifthenelse{\boolean{showfigures}}{\begin{figure}[H]
    \centering
    \includegraphics[width=0.8\textwidth]{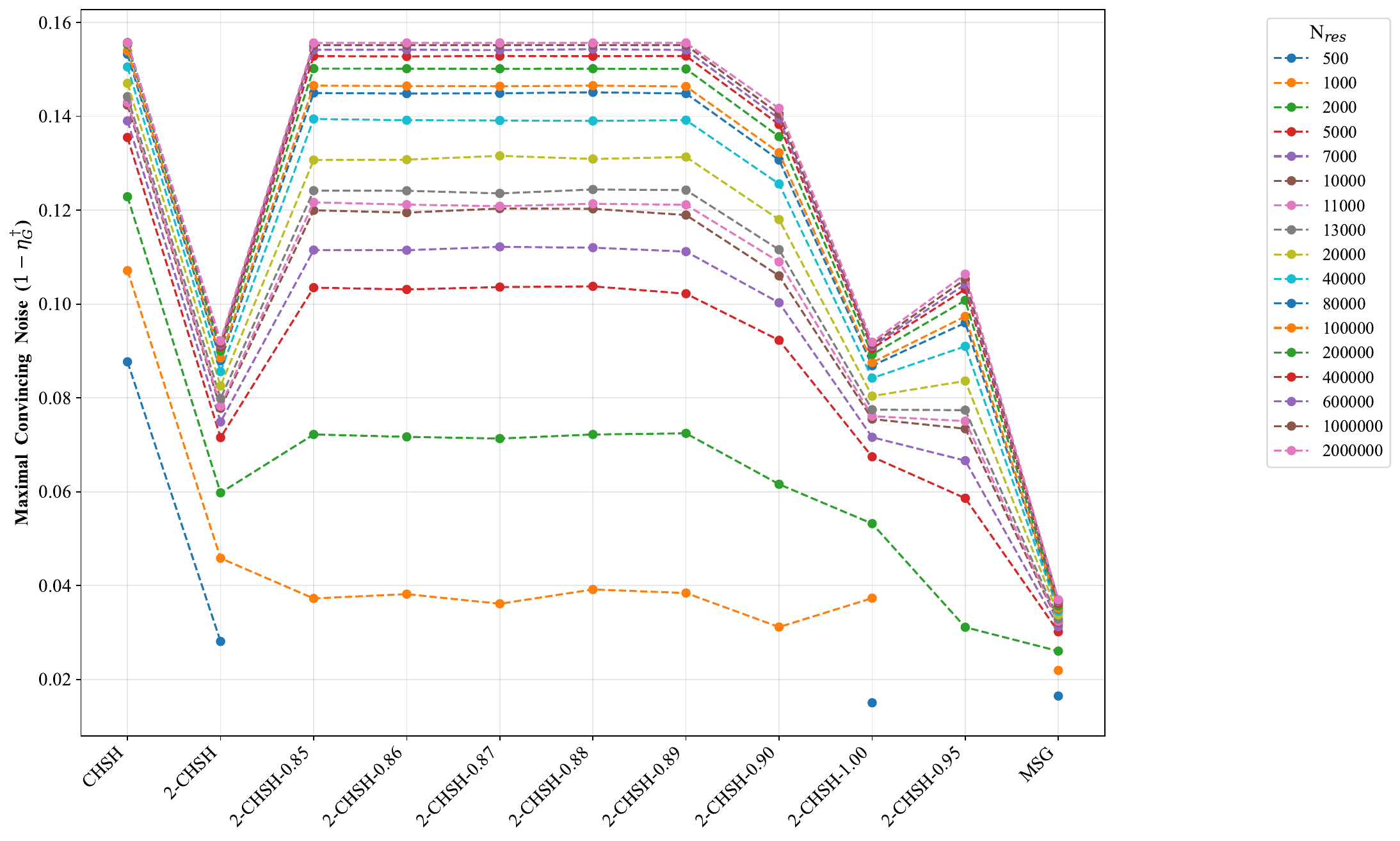}\caption[Order of Crossings ($1-\eta^\dagger$) for Games at Different N$_{res}$.]{\textbf{Order of Crossings ($1-\eta^\dagger_G$) for Games at Different N$_{res}$.}} 
    \label{fig:crossings-phase-changes}
\end{figure}}{}

With these observations in mind, we compare the crossing orders of experiments conducted for unstable ($N_{res} = 500,1000$) and stable ($N_{res} = 2000$) resource regions to the values of all candidate predictors in Fig. \ref{fig:candidates-comparison}. Doing so will enable us to ascertain which candidate is the most reliable estimator of the crossing orders in both regions. 

\ifthenelse{\boolean{showfigures}}{\begin{figure}[htbp]
    \centering
    \begin{subfigure}[b]{0.45\textwidth}
        \centering
        \includegraphics[width=\textwidth]{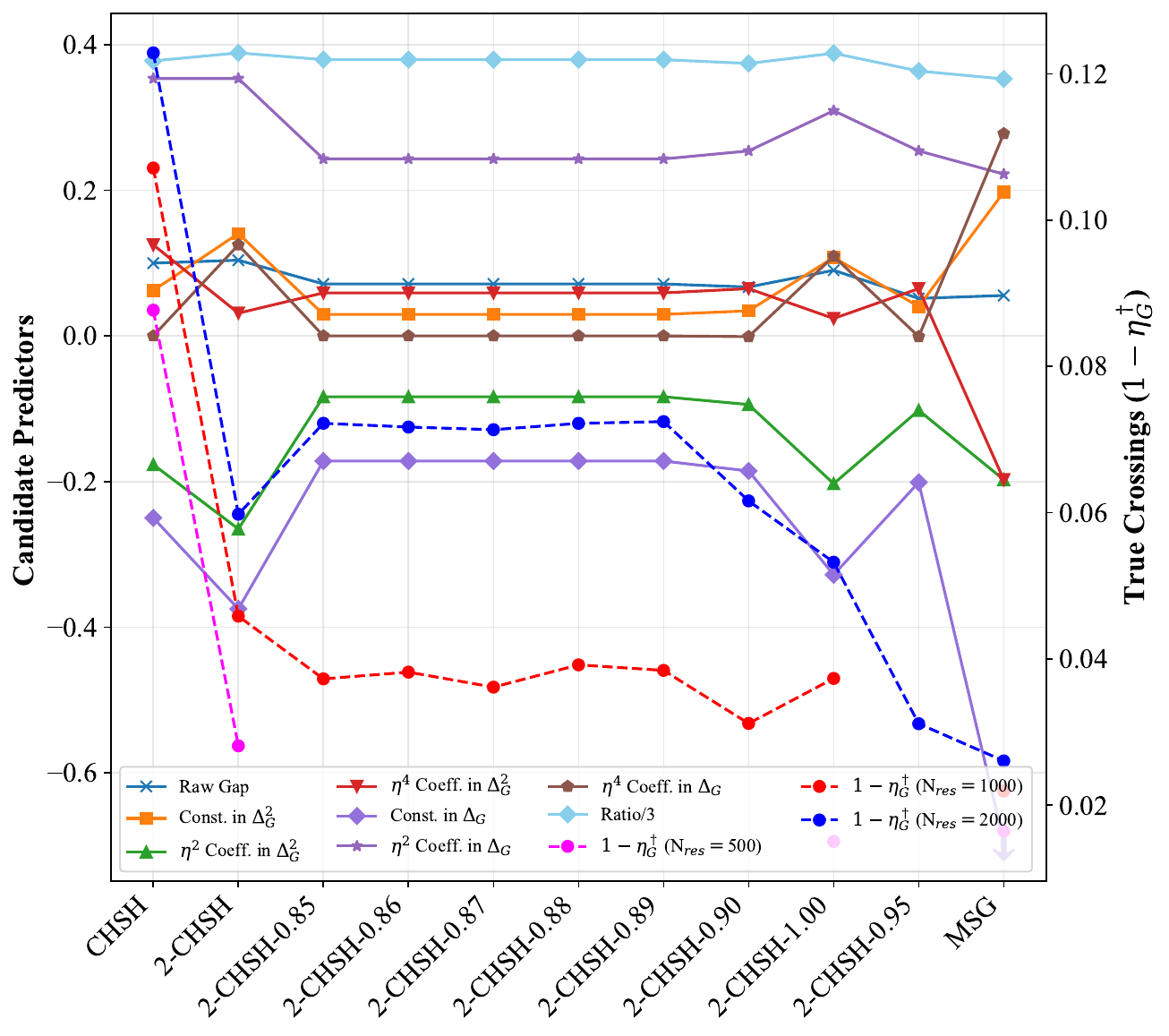}
        \label{fig:fig1}
    \end{subfigure}
    \hspace{0.02\textwidth}
    \begin{subfigure}[b]{0.45\textwidth}
        \centering
        \includegraphics[width=\textwidth]{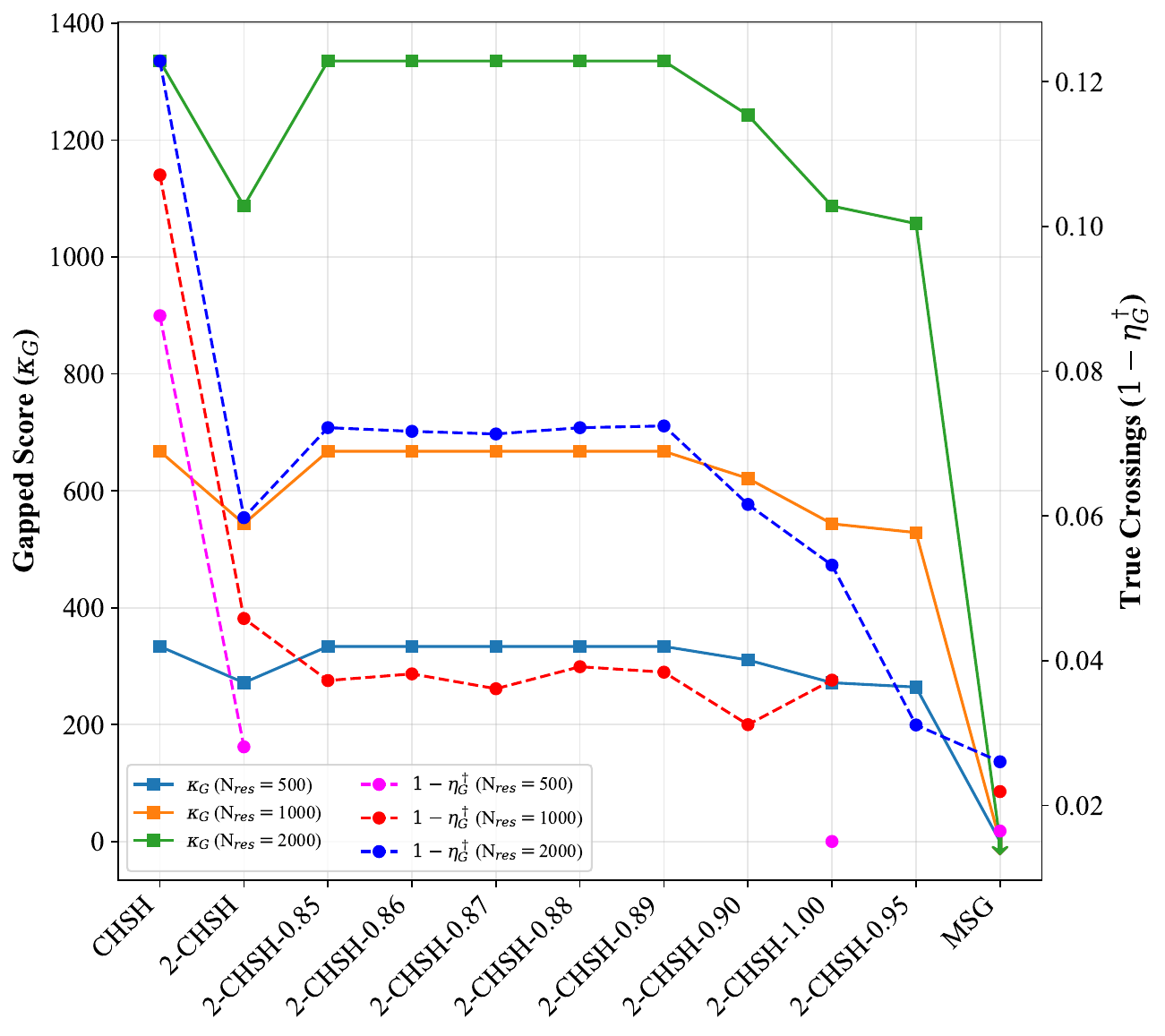}
        \label{fig:fig2}
    \end{subfigure}
    \caption[Comparing Candidate Predictors of Crossing Order.]{\textbf{Comparing Candidate Predictors of Crossing Order.} The figures show the true crossings attained at $N_{res}=500,1000, \text{ or } 2000$ across all experiments versus the gapped score (right) and all other candidate metrics (left). Some scaling was applied to metrics in order to make them  fit on a plot; however, no trends were affected. In particular, the MSG points of some curves were scaled down by 0.5, indicated by downward arrows, and all ratios were divided by 3.}
    \label{fig:candidates-comparison}
\end{figure}}{}
Overall, we observe that no candidate predictor is a reliable predictor for the unstable region, whereas the gapped score is clearly the most reliable predictor of the crossing orders in the stable region (Fig. \ref{fig:candidates-comparison}). In particular, Fig. \ref{fig:candidates-comparison} (right) we observe that for a stable crossing order, the gapped score excels at capturing: (1) the increases and decreases between games which have distinct convincingness behaviours (i.e., CHSH vs 2-CHSH vs MSG), and, (2) the nuances between similarly scoring games and configurations. For an unstable crossing order, the gapped score is less reliable. In contrast, in Fig. \ref{fig:candidates-comparison} (left) we observe that the candidates do not reliably predict the unstable and stable regions. Firstly, many candidates do not capture the main increases or decreases. For example, from CHSH to 2-CHSH, the crossings all decrease. However, many candidates (the Const. in $\Delta_G^2$, Raw Gap, Ratio, $\eta^4$ Coeff. in $\Delta_G$, $\eta^2$ Coeff. in $\Delta_G$) increase or stay the same. The most similar candidate curves to the stable crossing order are $\eta^2$ Coeff. in $\Delta_G^2$, $\eta^2$ Coeff. in $\Delta_G$, Const in  $\Delta_G$, and $\eta^4$ Coeff. in $\Delta_G^2$. These capture major increases or decreases between completely different games (i.e., CHSH, 2-CHSH, MSG), but fail to capture the differences between configurations (i.e., 2-CHSH-1.00 to 2-CHSH-0.95).  

The gapped score is not perfect, however. We observe that the CHSH and 2-CHSH-OPT games optimized for $\eta' \in \mathcal{Q}$ are predicted to have the same $\kappa_G$. On one hand, this aligns with the fact that these 2-CHSH-OPT configurations consistently have the closest $\eta^\dagger$ to $\eta^\dagger_{chsh}$. On the other, if we truly want to know the difference in the crossings for these games, this is unhelpful. Hence, we suggest complementing the gapped score with an even more nuanced theoretic polynomial analysis detailed in Sec. \ref{sec:poly-comparison}. This method uses only the gapped expressions, $\Delta_G$, to detect precise behaviours in the exact convincingness curves of games which have highly similar crossings and/or overlapping convincingness curves. By these observations and analysis, we therefore conclude Observations \ref{obs:unstable-crossing} and \ref{obs:stable-crossing}. Note that the latter observation holds for finite resources, as the crossings stabilize as early as 2000 noisy EPR pairs. We have, thus, overall devised a single, simple closed form expression for a score which can reliably rank the crossing order of games with different input-output numbers, supplemented by an additional simple polynomial analysis for games or configurations which have similar scores.

Lastly, let us discuss more theoretically why $\kappa_G$ is observed to be the best predictor of the crossings. First, it is not surprising that the raw gap and the ratio are unsuccessful predictors as they are fixed constants for a game which do not consider the noise model. The coefficients and constants, in comparison, do come from an expression which depends on our noise model: $\Delta_G = \omega_{\eta_G} - \omega_{c_G} = c_1 \eta + c_2\eta^2 + c_3\eta^3 + c_4\eta^4 + \cdots + d - \omega_{c_G}$. In our noise model, only the even degree terms have nonzero coefficients. Our initial hypothesis was that $c_2$, being the coefficient of the most significant term ($\eta^2$), would be a strong predictor for the crossings. However, $c_2$ fails to be a good predictor. By a simple plot of $c_2\eta^2 + d - \omega_{c_G}$, one can see that $c_2$ is not a good approximator of $\Delta_G$ especially at the $\eta$ regimes important to us ($\approx 0.8$ onward). Since $c_2$ is the most significant contributor, this explains why the other coefficients are unreliable predictors as well. A similar argument can be made for the $\Delta_G^2$ related predictors. In comparison, $\kappa_G$ takes all of the coefficients in $\Delta_G$ into account. Further, $\kappa_G$ considers the significance threshold and the number of resources used to play the game. Hence, it is the candidate that captures the most information about both the noise model and the convincingness experiment.

\subsection{Noise-Tolerance Analysis}

\ifthenelse{\boolean{showfigures}}{\begin{figure}[h!]
    \centering
    \includegraphics[width=\textwidth]{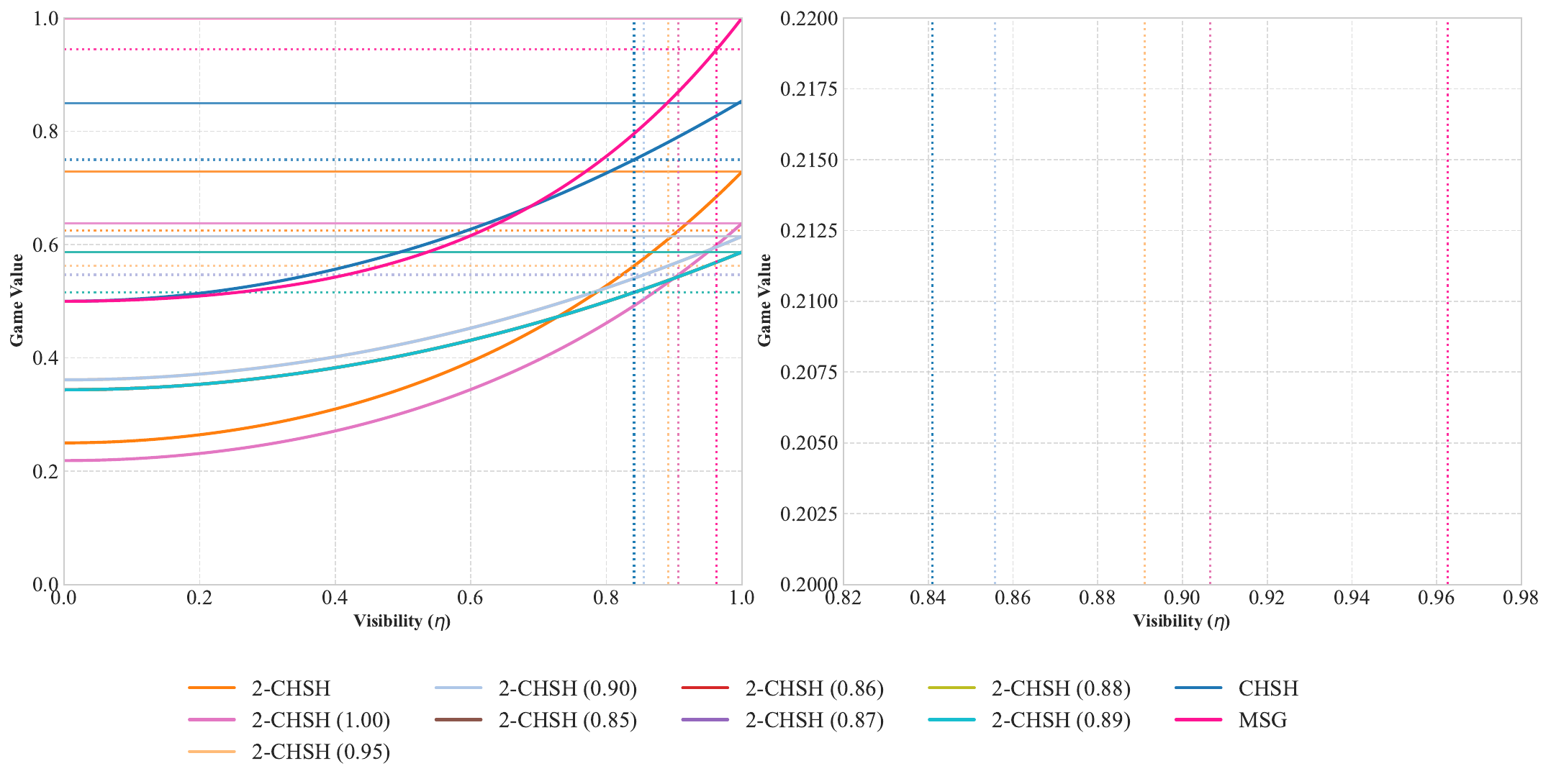}
    \caption[Noise-Tolerance Analysis.]{\textbf{Noise-Tolerance Analysis.}  Left: The horizontal dotted and filled lines represent the local and quantum bounds, respectively, for the curve of that colour. The vertical dotted lines indicate the intersection of the local bound with the curves, indicating the noise-tolerance for that game. Right: This is a zoom in on the vertical lines from the left figure. It is observed that the order of the game's noise-tolerances matches the order of the convincingness significance crossings for N=1,000,000 in Fig. \ref{fig:pvals-N=100K,N=1M}.}
    \label{fig:noise-tolerance}
\end{figure}}{}

Fig. \ref{fig:noise-tolerance} shows the result of the noise-tolerance analysis (Sec. \ref{sec:noise-tol-computational-methods}). The curves have a quadratic form due to the quadratic and quartic dependence of the depolarizing channel on the noise parameter. We note that two main properties of this graph influence the noise-tolerance: (1) the gap between local and quantum bounds, and (2) the slope of the curves. In fact, the experiment designed in Sec. \ref{sec:noise-tol-computational-methods} exactly plots the $\omega_{\eta_G}$ values derived later on in Sec. \ref{sec:analytic-methods}. Hence, the slope is attributed to the coefficients on the $\eta$ factors in the score expressions. The gap itself is defined as usual by the particular configuration of the game.

Interestingly, we see that the order of the noise-tolerances of the games (see Fig. \ref{fig:noise-tolerance} (right)) is the same as the order of the  significance crossings in the near-asymptotic convincingness experiment (Fig. \ref{fig:pvals-N=100K,N=1M}). Intuitively, this is because the noise-tolerance of a game, $G$, is defined for $\eta_G^*$ where $\omega_{\eta^*_G} = \omega_{c_G} \iff \Delta_G(\eta_G^*) = 0$. By the definition of non-local games, this is saying that the strategy's score is the local bound, and hence, $G$ does not detect any non-locality in that strategy. Relating this to convincingness, a gap of 0 implies that the game achieves completely unconvincing behaviour at $\eta_G^*$ under asymptotic resources ($\mathfrak{C}_G = 1$ at $\eta_G^*$). These are important realizations, as they show that the noise-tolerance and convincingness are two sides of the same coin, in the asymptotic resource limit. We can also see this directly on the convincingness figure: the noise-tolerance corresponds to the point just to the left of a significance crossing for a game, which is at $\mathfrak{C}_G = 1$. This shows that for our examples, the noise-tolerance results are captured by convincingness simulations. Note that we prove the equivalence between noise-tolerance and unconvincingness at asymptotic resources in Lemma \ref{pf:noise-tol-unconvincing} (App. \ref{app:noise-tol-unconvincing}). Further, we prove that in the asymptotic resource limit, the noise-tolerance level $\eta_G^* < \eta_G^\dagger$ for arbitrary $\eta_G^*,\eta_G^\dagger$ (Thm. \ref{thm:etastar-less-etadag}). Whether the order of the noise-tolerances and the order of the asymptotic convincingness crossings is the same for all possible games and, moreover, all noise models remains an open question. Our results show, however, that for the games tested, the convincingness can provide more nuanced information about their robustness to noise; likely stemming from the additional parameters in the convincingness model (i.e., $N_{res}$, $\alpha$). Moreover, the fact that the noise-tolerance results align with the asymptotic $\mathfrak{C}_G$ results suggests that the noise-tolerance, as a noise-robustness measure, is less practical for modern experimental design applications. Whether or not noise-tolerance has implications for the noise-robustness of games in the finite resource regime remains an open question. A future direction is to further develop the noise-tolerance model by adding some resource-dependent scaling factor. Although, as the convincingness naturally encodes this parameter in its statistical base, this motivates its practical use for the study of noise-robustness.

\begin{theorem}\label{thm:etastar-less-etadag}
    Let $\eta_G^\dagger$ represent the visibility level at which a convincingness curve for a game, $G$, crosses the significance threshold, $\alpha$. Let $\eta_G^*$ represent the minimum visibility level (maximal noise) that $G$ is tolerant of. Then $\eta_G^* 
    < \eta_G^\dagger$ holds for any $N_{res} \in \mathbb{N}$.  
\end{theorem}
\begin{proof}
    Let us determine $\eta_G^\dagger$ by analyzing the convincingness function. By Equation \ref{eqn:pval-simplified}, we have:
    \[
        \mathfrak{C}_G(\eta_G^\dagger) = \alpha 
        \iff \mathfrak{C}_G(\eta_G^\dagger) \sim \exp\left(-N_{res}(\Delta_G(\eta_G^\dagger))^2\right) = \alpha,
    \]
    where $\sim$ denotes asymptotic equivalence as $N_{res} \rightarrow \infty$. Taking the natural logarithm of both sides yields: $\Delta_G(\eta^\dagger_G) = \sqrt{-\,\frac{\ln(\alpha)}{N_{\mathrm{res}}}}$. Note that this value is strictly positive for any $N_{res} \in \mathbb{N}$ since $\alpha < 1$. Now, we can establish that $\eta^\dagger_G > \eta_G^*$ through the following steps. By Lemma \ref{pf:strictly-inc-gap}, $\Delta_G(\eta)$ is a strictly increasing function. Further, by definition, the noise-tolerance $\eta_G^*$ is the visibility level where the local bound and $\omega_{\eta_G}$ become equal. Hence, $\omega_{c_G} = \omega_{\eta_G} \iff \Delta_G(\eta_G^*) = \omega_{\eta_G} - \omega_{c_G} = 0$.
Therefore, since $\Delta_G(\eta)$ is strictly increasing, $\Delta_G(\eta_G^*) = 0$, and $\Delta_G(\eta_G^\dagger) = \sqrt{-\ln(\alpha)/N_{\mathrm{res}}} > 0$ for any $N_{res} \in \mathbb{N}$, we conclude that $\eta_G^* < \eta^\dagger_G$ for any $N_{res} \in \mathbb{N}$.
\end{proof}

\begin{lemma}\label{pf:strictly-inc-gap}
    Given some depolarizing channel visibility $\eta > 0$ and game $G$, $\Delta_G(\eta)$ is strictly increasing.
\end{lemma}
\begin{proof}
The value of a game played with input state $\varepsilon(\eta,\rho_{win})$ is defined as: $\omega_{\eta_G} = \text{Tr}(\mathbf{S}_G \cdot \varepsilon(\eta, \rho_{win}))$. Taking the 2-qubit depolarizing channel as an example,  \(\varepsilon(\eta, \rho_{win}) = \eta^2 \rho_{win} + (1 - \eta^2)\frac{\mathbb{I}}{4}\). Differentiating with respect to \(\eta\):
\[
\frac{\partial \omega_{\eta_G}}{\partial \eta} = \text{Tr}(\mathbf{S}_G (2\eta \rho_{win} - 2\eta \frac{\mathbb{I}}{4})) \iff \frac{\partial \omega_{\eta_G}}{\partial \eta} = 2\eta \text{Tr}(\mathbf{S}_G \cdot (\rho_{win} - \frac{\mathbb{I}}{4})).
\]
Since \(\text{Tr}(\mathbf{S}_G \cdot \rho_{win}) > \text{Tr}(\mathbf{S}_G \cdot \frac{\mathbb{I}}{4})\) (the winning state always achieves a higher score than the maximally mixed state), we have: $\frac{\partial \omega_{\eta_G}}{\partial \eta} > 0$.
Thus, \(\omega_{\eta_G}\) is a strictly increasing function of \(\eta\). This argument extends by a simple computation of the derivative to the 4-qubit depolarizing channel. Since \(\Delta_G = \omega_{\eta_G} - \omega_{c_G}\), and \(\omega_{c_G}\) is constant, the strictly increasing nature of \(\omega_{\eta_G}\) implies: $\frac{\partial \Delta_G}{\partial \eta} = \frac{\partial \omega_{\eta_G}}{\partial \eta} > 0$. Thus, \(\Delta_G\) is a strictly increasing function of \(\eta\).
\end{proof}

\section{DEFINING NOISE-ROBUSTNESS OF NON-LOCAL GAMES}
We observed in Sec. \ref{sec:results-disc} that the convincingness threshold crossings, and the more analytic gapped score, can be used to rank the probability of different non-local games to detect non-locality in noisy states. This naturally suggests the following, intuitive, definition of noise-robustness:

\begin{center}
\textit{The noise-robustness of a non-local game is the maximal noise (on its ideal state) for which it can detect non-locality with high probability.}
\end{center}

\noindent Mathematically, the convincingness is an excellent tool for quantifying the maximal noise on the input state that is supported by a non-local game. Moreover, as we have seen, the convincingness is more practically applicable to finite resource regimes than the noise-tolerance measure. Hence, we propose the following formal definition of noise-robustness:  
\begin{definition}[Noise-Robustness using Convincingness]\label{def:noise-rob-1}
    A game $G_1$ is more noise-robust to a noisy channel, $\Phi(\eta)$, than $G_2$ if $G_1$ is significantly convincing for lower visibilities $\eta$ than $G_2$. That is, $G_1$ is more $\Phi(\eta)$-noise-robust than $G_2$ if $\eta^\dagger_{G_1} < \eta^\dagger_{G_2}$, where $\eta$ is the channel's visibility $\eta$ (i.e., noise = $1-\eta$) on the ideal state of the game and $\eta^\dagger_{G_i}$ is the visibility level at which the convincingness curve of game $G_i$ crosses the chosen significance threshold, $\alpha$, using a supply of $N_{res}$ $\eta^\dagger_{G_i}$-noisy resources. Note that $\alpha$ is the target likelihood of a game detecting locality. 
\end{definition}

\noindent Further, for resource regimes which yield a stable ordering of $\eta_G^\dagger$ crossings, we know that the gapped score can reliably predict the trends in the $\eta_G^\dagger$ of different games. Hence, for these stable resource regimes we define noise-robustness more simply, presented in Def. \ref{def:gapped-score-def-of-noise-rob} below. This second definition offers a key advantage: it can be computed analytically. This makes it both more practical for experimental work and more suitable for theoretical studies of noise-robustness. Note that Def. \ref{def:gapped-score-def-of-noise-rob} is very convenient for theoretical questions about the noise-robustness in the asymptotic resource limit, since the asymptotic resource regime naturally tends towards stability.
\begin{definition}[Noise-Robustness using the Gapped Score]\label{def:gapped-score-def-of-noise-rob}
A game $G_1$ is more noise-robust to a noisy channel, $\Phi(\eta)$, than $G_2$ if $\kappa_{G_1} > \kappa_{G_2}$, where the games are played in a stable $N_{res}$ region. 
\end{definition}

\section{CONCLUSION}
This thesis presents a rigorous framework for analyzing and comparing the noise‑robustness of non‑local games. In doing so, it addresses a critical gap in self‑testing theory and practice. Through the development of novel comparison methods---particularly the convincingness measure, its analytical approximation (the gapped score), and noise-tolerance---we have established quantitative tools for evaluating how effectively different games can detect non-locality in noisy quantum states. In analyzing the capabilities of each method, we arrived at a precise and reliable definition of noise-robustness which uses the significance crossings of the convincingness. Formally, a game $G_1$ is more noise-robust to a noisy channel than $G_2$ if $G_1$ has a lower significance crossing than $G_2$ (i.e., $G_1$ is significantly convincing for higher noise levels than $G_2$) (Def. \ref{def:noise-rob-1}). For stable resource regimes, this simplifies to analytically comparing the games' gapped scores (Def. \ref{def:gapped-score-def-of-noise-rob}).

We note that the proposed noise-robustness framework successfully answers the questions posed in Sec. \ref{sec:intro}. First, it identifies the most noise‑robust game out of a set of candidates, for a given noise model. Second, it pinpoints the test that maximizes deviation from local statistics under specified resource budgets. Third, by allowing resource allocation to vary or Bell coefficients to be optimized, it systematically enlarges the set of noisy states whose non‑locality can be certified. Even more, our analysis revealed that games with lower noise-robustness under equal resource constraints can transition to being more robust when unequal resource allocations are allowed, highlighting that CHSH is not the most noise-robust game if unequal resources are permitted.

Importantly, we note that because the convincingness metric embeds resource constraints directly into its definition, it supports meaningful comparisons in both finite and asymptotic regimes--\textit{even when candidate games differ in input–output settings or state dimension}. In practice, convincingness yields finer experimental guidance than the coarse noise‑tolerance threshold while converging to the same ranking in the infinite-resource limit. For analytical work, the gapped score offers a closed‑form proxy for convincingness, enabling asymptotic studies without the computational overhead of computing significance crossings. Thus, the framework gives experimentalists a direct noise‑robustness score and theorists a closed‑form asymptotic proxy. Together, the convincingness and gapped scores let one choose non‑local games that match a given noise model and resource budget—whether the aim is maximal robustness for detecting noisy entanglement and preserving these resources, or low noise-robustness for stringent verification, security‑critical tasks.

\section{APPLICATIONS}\label{sec:apps}

To illustrate the applicability of this framework, we note that non-local games serve as fundamental subroutines in multiple quantum information processing tasks, including quantum algorithms, communication protocols, and cryptographic applications. For instance, in applications such as DI quantum key distribution (DIQKD) \cite{vazirani2019fully, renner2008security} and randomness generation \cite{pironio2010random, herrero2017quantum}---where security relies on the use of near-maximally entangled states---it may be advantageous to employ less noise-robust games than CHSH (i.e., certain 2-CHSH variants) that yield stronger guarantees under ideal conditions. Conversely, in applications such as entanglement distillation \cite{bennett1996concentrating, bennett1996purification, horodecki2009quantum}, which can operate effectively on lower-fidelity entangled states, more noise-robust games (i.e., CHSH) may be preferable, as they enable the certification and utilization of a larger set of imperfect quantum resources, thereby improving the overall efficiency of the protocol. Hence, our methods can be used to make these schemes more efficient and hopefully preserve quantum resources in the process.

Another prospective application is the use of more noise-robust tests in DIQKD scenarios where the secret-key rate is zero, or nearly-zero. We present two such scenarios. The first is theoretical rate studies that model full quantum-repeater chains, exemplified in \cite{holz2018device}. The second is single‑segment, memory‑assisted entanglement‑swapping prototypes—such as the trapped‑ion experiment in \cite{nadlinger2022experimental} and the rubidium‑atom experiment in \cite{zhang2022device}—which realize the core repeater building block without yet chaining multiple segments. In both settings every entanglement‑swapping stage reduces the visibility of the shared pair; lifting the visibility back above the violating threshold of the self‑tests used in the DIQKD protocols often demands an additional distillation round that halves the surviving‑pair rate and can exceed the memories’ decoherence window, driving the key rate to zero. A self‑test with a slightly lower threshold visibility (i.e., lower significance crossing) could potentially certify those same swapped pairs without invoking that costly distillation layer, providing a small---but strictly positive---secret-key where current tests yield none. It is important to note, however, that the certified randomness of the outcome will likely be lower, at the cost of ultimately achieving a positive key rate. Hence, the trade-off between certified randomness and the secret-key rate should be considered when increasing the noise-robustness of self-tests in DIQKD methods. A systematic investigation of this trade-off could inform the optimal choice of non-local games in different applications.

Looking to entanglement theoretic research, the noise-robustness framework also opens a possibility for researching partially entangled mixed states. Specifically, as the advantages of partially entangled mixed states in quantum protocols remain largely unexplored--possibly due to detection difficulties--our framework opens new avenues for investigating their application and fundamental properties. In contrast, as partially entangled \textit{pure} states have known advantages in applications, such as random bit generation \cite{woodhead2020maximal}, future work could extend the analysis in this work to noise models producing partially entangled pure states. In this way, more of these states can be preserved for random bit generation schemes. 

Lastly, the noise-robustness framework enables evaluating the noise-robustness of additional non-local games, under different noise models and resource constraints. The convincingness score can also be used as a general comparative score between non-local games, outside of the noise-robustness context. That is, $\mathfrak{C}$ can be used as a comparable, operational score which normalizes the scores of games with different input-output settings and state dimensionalities.

\section{STRENGTHS \& WEAKNESSES}
Our framework offers several strengths. It provides a non-experimental method for experimentalists to determine the most suitable game based on their noise model and resources. Its flexibility accommodates various noise models, enabling the inclusion of multiple imperfections, and the convincingness measure facilitates comparison across games with different dimensionalities, input-output settings, and resource constraints. Additionally, while focused on noise-robustness to input states, the framework could extend to robustness against measurement noise by comparing games optimized for specific detection inefficiencies. Finally, the gapped score provides an analytic alternative to classical simulations, allowing noise-robustness analyses in higher-dimensional Hilbert spaces.

However, there are some limitations. The framework requires knowledge of the system's noise model, which can be challenging to characterize, though advancements in quantum noise characterization make this increasingly feasible. Experimentalists can also adopt more conservative noise models by, for example, considering a combination of multiple noisy channels that are known to be affecting a device (i.e., depolarizing and dephasing noise). Another requirement is knowledge of each game's local bound, which is a one-time calculation. Practical implementation of noise-robust games on quantum hardware also presents challenges. 

\section{FUTURE RESEARCH AVENUES \& OPEN QUESTIONS}
Below we outline some of the several avenues for future investigation into the noise-robustness of non-local games.
\subsection{Non-Local Reductions}
Throughout this work we repeatedly returned to the idea of comparing games by how far into the noisy regime they still certify non-locality. Motivated by this, we propose a systematic \textit{non-local reduction} hierarchy for non-local games and point to its formal development as a promising avenue for future research.
For non-local games $A$ and $B$, we write:
\[
B \preceq_{\mathrm{NR}} A
\quad\Longleftrightarrow\quad
\text{every state that violates the classical bound of } B
\text{ also violates that of } A
\;(\text{equivalently, } \mathcal{V}(B)\subseteq\mathcal{V}(A)).
\]

The relation $\preceq_{\mathrm{NR}}$ is reflexive and transitive, so it is a preorder.  
Games higher in this hierarchy are strictly more noise-robust: they certify non-locality for all states witnessed by any lower game \emph{and} for additional, noisier states.  Thus the “harder’’ (more noise-robust) games are precisely the maximal elements, because they witness entanglement for the largest possible class of noisy states.

This definition is natural in the \textit{resource theory of non-locality}, whose free operations are \textit{local operations and shared randomness} (LOSR).  
Because LOSR cannot create Bell violations, it orders resources by convertibility:
\begin{align}
X \xrightarrow{\text{LOSR}} Y
\;\Longleftrightarrow\;
\text{there exists an LOSR protocol mapping } X \text{ to } Y .
\end{align}
Extensive work on abstract resource theories has introduced LOSR monotones and closed sets that quantify convertibility under \emph{state-level} LOSR operations~\cite{chitambar2019quantum,coecke2016mathematical}.  
In our setting this implies that, whenever $\mathcal{V}(B)\subseteq\mathcal{V}(A)$, any state that violates $B$, using LOSR alone (and fresh measurements), can be transformed into an $A$-violating one; hence, the resource detected by $B$ is therefore always convertible into one detected by $A$. 

A related—but stricter—framework orders \emph{behaviour boxes} by \emph{device-independent} LOSR wirings~\cite{wolfe2020quantifying,zjawin2023quantifying}.  
Here the free operations act only on the observed input–output \emph{probability distributions} \(P(ab|xy)\), leaving the underlying quantum state and measurements untouched.  
Such box-level convertibility can fail even when state-level convertibility holds: for example, it can be shown that the perfect Magic-Square box cannot be wired into the Tsirelson CHSH box under DI-LOSR (see App. \ref{app:non-convertability}), although the same two-singlet state violates CHSH once new measurements are chosen. These limitations motivate the state-based, game-level preorder proposed here.

Relating resource-theoretic monotones to the state sets that violate given non-local games would lay the quantitative groundwork for defining \emph{non-local reductions} between games, which by extension, can enable a ranking of games' robustness to noisy (partially non-local) states. Specifically, developing this preorder—finding complete monotones, conversion rates, and maximal elements—would add an order-theoretic method for ranking games' noise-robustness. To our knowledge, no rigorous study of such inter-game reductions yet exists, making this an open and promising research direction. 

\subsection{Further Investigations of Noise-Robustness}
The noise-robustness framework proposed in this work and our results naturally prompt the following research questions:

\begin{enumerate} \item What limits the noise-robustness of parallel-repetition schemes like 2-CHSH compared to simpler tests like CHSH? \item In finite resource regimes, is CHSH still the most robust when considering noise models other than depolarizing noise, such as more complex models used in laboratory settings? \item Could optimizing games along multiple axes enhance their noise-robustness? \item How can the stability of games in lower resource regimes be improved? \item Is there a precise relationship between the significance threshold $\alpha$ and other measures of entanglement of the tested state? \item How might this framework be applied to more complex non-local games? \item What effects do entanglement distillation schemes have on convincingness and gapped scores? 
\end{enumerate}

 Additional future directions include using the proposed framework to study the noise-robustness of more complex games with greater input-output settings or input state dimension, and applying the framework to the applications discussed in Sec. \ref{sec:apps}.

\section{ACKNOWLEDGMENTS}
I would like to express my deepest gratitude to my thesis supervisor, Prof. Thomas Vidick. Your guidance and insight throughout this thesis journey have been invaluable, and I am truly thankful for all the time and care you have devoted to mentoring, supporting, and developing my academic rigour. You have taught me so much, both through your mentorship and by example, and I will always be grateful for the opportunity to learn from you. I would also like to thank Dr. Elie Wolfe, Dr. Anand Natarajan, Prof. Debbie Leung, Prof. Dr. Jens Eisert, Prof. Nathan Wiebe, Noam Avidan, Ilya Merkulov, Itammar Steinberg, and Thomas Hahn for helpful discussions throughout. Lastly, I would like to thank my family--my mother, father, and brother--for supporting me throughout anything and everything, and pushing me to be the best version of myself, always.

\newpage
\bibliographystyle{unsrt}
\bibliography{refs}

\newpage
\section{(APPENDIX) 2-CHSH-OPT Score Derivation}\label{app:2-chsh-opt-score-deriv}
\begingroup
\small
\begin{align*}
\text{Tr}(\mathbf{P}_\eta \bar{\mathbf{B}}_{\eta’}^T) &= \text{Tr} \left[\frac{1}{16} \left(\sum_{a,b=1}^{16} \text{Tr}(\rho_{in} (A_a \otimes B_b))|e_a\rangle\langle e_b|\right)\left(\sum_{i,j=1}^{16} \bar{\mathbf{B}}_{\eta’}^{i,j} |e_j\rangle\langle e_i|\right)\right]\\
&= \frac{1}{16} \text{Tr} \left[\sum_{a,b=1}^{16}\sum_{i,j=1}^{16} \text{Tr}(\rho_{in} (A_a \otimes B_b))\bar{\mathbf{B}}_{\eta’}^{i,j} |e_a\rangle\langle e_b|e_j\rangle\langle e_i|\right]\\
&= \frac{1}{16} \text{Tr} \left[\sum_{a,b=1}^{16}\sum_{i,j=1}^{16} \text{Tr}(\rho_{in} (A_a \otimes B_b))\bar{\mathbf{B}}_{\eta’}^{i,j} \delta_{b,j} |e_a\rangle\langle e_i| \right]\\
&= \frac{1}{16} \text{Tr} \left[\sum_{a,b=1}^{16}\sum_{i=1}^{16} \text{Tr}(\rho_{in} (A_a \otimes B_b))\bar{\mathbf{B}}_{\eta’}^{i,b} |e_a\rangle\langle e_i| \right]\\
&= \frac{1}{16} \sum_{a,b=1}^{16}\sum_{i=1}^{16} \text{Tr}(\rho_{in} (A_a \otimes B_b))\bar{\mathbf{B}}_{\eta’}^{i,b} \text{Tr}(|e_a\rangle\langle e_i|)\\
&= \frac{1}{16} \sum_{a,b=1}^{16}\sum_{i=1}^{16} \text{Tr}(\rho_{in} (A_a \otimes B_b))\bar{\mathbf{B}}_{\eta’}^{i,b} \delta_{a,i}\\
&= \frac{1}{16} \sum_{a,b=1}^{16} \text{Tr}(\rho_{in} (A_a \otimes B_b))\bar{\mathbf{B}}_{\eta’}^{a,b} \\
&= \frac{1}{16} \sum_{a,b=1}^{16} \text{Tr}(\rho_{in} (A_a \otimes B_b) \bar{\mathbf{B}}_{\eta’}^{a,b}) & \text{(scalar multiplication)}\\
&= \frac{1}{16} \text{Tr}\left(\sum_{a,b=1}^{16} \rho_{in} (A_a \otimes B_b) \bar{\mathbf{B}}_{\eta’}^{a,b}\right) & \text{(linearity of trace)}\\
&= \frac{1}{16} \text{Tr}\left(\rho_{in}\sum_{a,b=1}^{16} (A_a \otimes B_b) \bar{\mathbf{B}}_{\eta’}^{a,b}\right) & \text{(cyclicity of trace)}\\
&= \frac{1}{16} \text{Tr}\left[\biggl(\eta^4(\rho_{\text{EPR}} \otimes \rho_{\text{EPR}}) + \eta^2(1-\eta^2)(\rho_{\text{EPR}} \otimes \frac{\mathbb{I}_4}{4}) \right. \\
&\quad \left. + \eta^2(1-\eta^2)(\frac{\mathbb{I}_4}{4} \otimes \rho_{\text{EPR}}) + (1-\eta^2)^2(\frac{\mathbb{I}_4}{4} \otimes \frac{\mathbb{I}_4}{4})\biggr)\sum_{a,b=1}^{16} (A_a \otimes B_b) \bar{\mathbf{B}}_{\eta’}^{a,b} \right] \\
&= \eta^4 \frac{\text{Tr}\left[(\rho_{\text{EPR}} \otimes \rho_{\text{EPR}})\sum_{a,b=1}^{16} (A_a \otimes B_b) \bar{\mathbf{B}}_{\eta’}^{a,b}\right]}{16} \\
& \quad + \eta^2(1-\eta^2) \biggl(\frac{\text{Tr}\left[(\rho_{\text{EPR}} \otimes \frac{\mathbb{I}_4}{4})\sum_{a,b=1}^{16} (A_a \otimes B_b) \bar{\mathbf{B}}_{\eta’}^{a,b}\right] + \text{Tr}\left[(\frac{\mathbb{I}_4}{4} \otimes \rho_{\text{EPR}})\sum_{a,b=1}^{16} (A_a \otimes B_b) \bar{\mathbf{B}}_{\eta’}^{a,b}\right]}{16}\biggr)  \\ &\quad + (1-\eta^2)^2\frac{\text{Tr}\left[(\frac{\mathbb{I}_4}{4} \otimes \frac{\mathbb{I}_4}{4})\sum_{a,b=1}^{16} (A_a \otimes B_b) \bar{\mathbf{B}}_{\eta’}^{a,b}\right]}{16} \\
&\coloneqq \eta^4 k_\text{ideal} + \eta^2(1-\eta^2)k_\text{part-mixed} + (1-\eta^2)^2 k_\text{mixed}.
\end{align*}
\endgroup

\newpage
\section{(APPENDIX) Convincingness Computations for Optimized 2-CHSH Games}\label{app:A}
\ifthenelse{\boolean{showfigures}}{
\begin{table}[htbp] 
    \centering
    \caption{Gapped Expressions for Non-Local Games}
    \label{tab:gapped-expressions}
    \begin{tabular}{
        >{\raggedright\arraybackslash}p{0.45\linewidth}  
        >{\raggedright\arraybackslash}p{0.45\linewidth}  
    }
    \toprule
    \textbf{$G$} & \textit{$\omega_{\eta_G}$ raw expressions} \\
    \midrule
    CHSH & $0.35355 \eta^2  +0.5$ \\ 
    \addlinespace[0.5em]
    2-CHSH & $0.12500\eta^4 + 0.35355\eta^2 +  0.25000$ \\
    \addlinespace[0.5em]
    2-CHSH-OPT $\eta'=1$ & $0.63748\eta^4  + 0.74686\eta^2(1-\eta^2) +0.21875(1-\eta^2)^2$\\
    \addlinespace[0.5em]
    2-CHSH-OPT $\eta'=0.95$ & $0.61447\eta^4  + 0.97677\eta^2(1-\eta^2) + 0.36133(1-\eta^2)^2$ \\
    \addlinespace[0.5em]
    2-CHSH-OPT $\eta'=0.90$ & $0.61447\eta^4  + 0.97677\eta^2(1-\eta^2) + 0.36133(1-\eta^2)^2$\\
    \addlinespace[0.5em]
    2-CHSH-OPT $\eta'=0.89$ &  $0.58682\eta^4  + 0.93057\eta^2(1-\eta^2) + 0.34375(1-\eta^2)^2$\\
    \addlinespace[0.5em]
    2-CHSH-OPT $\eta'=0.88$ &  $0.58682\eta^4  + 0.93057\eta^2(1-\eta^2) + 0.34375(1-\eta^2)^2$\\
    \addlinespace[0.5em]
    2-CHSH-OPT $\eta'=0.87$ &  $0.58682\eta^4  + 0.93057\eta^2(1-\eta^2) + 0.34375(1-\eta^2)^2$\\
    \addlinespace[0.5em]
    2-CHSH-OPT $\eta'=0.86$ &  $0.58682\eta^4  + 0.93057\eta^2(1-\eta^2) + 0.34375(1-\eta^2)^2$\\
    \addlinespace[0.5em]
    2-CHSH-OPT $\eta'=0.85$ &  $0.58682\eta^4  + 0.93057\eta^2(1-\eta^2) + 0.34375(1-\eta^2)^2$\\
    \addlinespace[0.5em]
    2-CHSH-OPT $\eta'=0.84$ &  $0.53125\eta^4  + 1.06250\eta^2(1-\eta^2) + 0.53125(1-\eta^2)^2$\\
    \addlinespace[0.5em]
    2-CHSH-OPT $\eta'=0.83$ &  $0.53125\eta^4  + 1.06250\eta^2(1-\eta^2) + 0.53125(1-\eta^2)^2$\\
    \addlinespace[0.5em]
    MSG & $ 0.27778 \eta^4 + 0.22222 \eta^2 + 0.5$ \\
    \bottomrule
    \end{tabular}
    \smallskip
    \begin{flushleft}
    \small\emph{Note:} 4-qubit games are in the form $k_\text{ideal}\eta^4 + k_\text{part-mixed}\eta^2(1-\eta^2) + k_\text{mixed}(1-\eta^2)^2$.
    \end{flushleft}
\end{table}}{}

\newpage
\section{(APPENDIX) 2-CHSH-OPT Local Bound Computation}\label{app:2chsh-opt-local-bound}

The local bound of a 2-CHSH-OPT game is computed by exhaustively evaluating the score of all possible deterministic local strategies. The local strategies are represented as 16-bit strings corresponding to the input-output behaviour of both parties. For each strategy, we compute its score using the game's $\bar{\mathbf{B}}_{\eta'}$ matrix, and finally identify the maximum achievable score across all bit-strings. This approach is complete since any local strategy can be written as a probabilistic mixture of deterministic strategies, and the maximum score over deterministic strategies gives the local bound by convexity. \\ 

\begin{algorithm}[h!]
\SetAlgoLined
\KwIn{Normalized Bell matrix $\bar{\mathbf{B}}_{\eta'}$ encoded in the form $B$ with elements $B[x,y][a,b]$ where:
\begin{itemize}
\item $x \in {0,1,2,3}$ represents input pairs $(x_1,x_2)$ as $(0,0),(0,1),(1,0),(1,1)$ for Alice
\item $y \in {0,1,2,3}$ represents input pairs $(y_1,y_2)$ as $(0,0),(0,1),(1,0),(1,1)$ for Bob
\item $a \in {0,1,2,3}$ represents output pairs $(a_1,a_2)$ as $(0,0),(0,1),(1,0),(1,1)$ from Alice
\item $b \in {0,1,2,3}$ represents output pairs $(b_1,b_2)$ as $(0,0),(0,1),(1,0),(1,1)$ from Bob
\end{itemize}
Each strategy $s$ is a 16-bit string where:
\begin{itemize}
\item Bits $[0-3]$: Alice's first output ($a_1$) for inputs $x \in {0,1,2,3}$
\item Bits $[4-7]$: Alice's second output ($a_2$) for inputs $x \in {0,1,2,3}$
\item Bits $[8-11]$: Bob's first output ($b_1$) for inputs $y \in {0,1,2,3}$
\item Bits $[12-15]$: Bob's second output ($b_2$) for inputs $y \in {0,1,2,3}$
\end{itemize}}
\KwOut{Maximum normalized score}
Generate all possible deterministic strategies $S = \{0,1\}^{16}$\;
$\text{max\_score} \leftarrow -\infty$\;
\ForEach{strategy $s \in S$}{
    $\text{joint\_prob} \leftarrow 0$\;
    \For{$x \leftarrow 0$ \KwTo $3$}{
        \For{$y \leftarrow 0$ \KwTo $3$}{
            $a_1 \leftarrow s[x]$\;
            $a_2 \leftarrow s[x + 4]$\;
            $b_1 \leftarrow s[y + 8]$\;
            $b_2 \leftarrow s[y + 12]$\;
            $a \leftarrow 2a_1 + a_2$\;
            $b \leftarrow 2b_1 + b_2$\;
            $\text{joint\_prob} \leftarrow \text{joint\_prob} + B[x,y,a,b]$\;
        }
    }
    $\text{max\_score} \leftarrow \max(\text{max\_score}, \text{joint\_prob})$\;
}
\Return{$\text{max\_score}$}
\caption{Compute Maximum Classical Score for a 2-CHSH-OPT Game}
\end{algorithm}

\newpage
\section{(APPENDIX) Proof of Noise-Tolerance Relation to Convincingness}
\label{app:noise-tol-unconvincing}
\begin{lemma}\label{pf:noise-tol-unconvincing}
    The noise-tolerance $\eta_G^*$ for a non-local game, $G$, corresponds to a completely unconvincing state ($\mathfrak{C}_G = 1$), under asymptotic resource conditions ($N_{res} \rightarrow \infty$). 
\end{lemma}
\begin{proof}
    By definition, the noise-tolerance, $\eta_G^*$, for a non-local game, $G$, is the visibility at which $\omega_{c_G} = \omega_{\eta_G}$. Rearranging, we see that:
\[\omega_{c_G} = \omega_{\eta_G} \iff 0 = \omega_{\eta_G} -\omega_{c_G} = \Delta_G(\eta_G^*) \implies \mathfrak{C}_G(\eta_G^*) = 1,\]
where the last implication is true for $N_{res} \rightarrow \infty$ by Eqn. \ref{eqn:pval-simplified}. In other words, we can make two statements from different perspectives: 
(1) the value of $G$ played with an input state with visibility $\eta_G^*$ equals the local bound $\omega_{c_G}$, so by the definition of non-local games, there is no non-locality in the input state and, (2) since $\mathfrak{C}_G = 1$ under $N_{res} \rightarrow \infty$, there is no likelihood of non-locality under these conditions. That is, $\eta_G^*$ is the value at which the game is completely unconvincing for asymptotic resources. As desired.
\end{proof}

\newpage
\section{(APPENDIX) Impossibility of DI–LOSR Conversion from the Magic-Square to Tsirelson Box}\label{app:non-convertability}
\textbf{This proof was contributed by Dr. Elie Wolfe, and written in this thesis with his permission.}

\begin{theorem}
Under device–independent local operations and shared randomness (DI--LOSR), the perfect bipartite Magic‑Square correlation (``Magic‑Square box'') cannot be converted into the unique quantum correlation that attains the Tsirelson bound for the CHSH game (``Tsirelson box'').
\end{theorem}

\begin{proof}
A \emph{box} (also called a \emph{correlation} or \emph{behaviour})  is a conditional probability distribution
\(P(ab\mspace{2mu}|\mspace{2mu}xy)\) with binary inputs
\(x,y\in\{0,1\}\) and outputs \(a,b\in\{0,1\}\).
A \textit{DI--LOSR wiring} is a box‑to‑box map that uses only \emph{local}
pre‑processing of the inputs and post‑processing of the outputs,
together with shared randomness.
Operationally it is a convex combination
\(
\Lambda \;=\;\sum_{r} p_r \,\Lambda^{(r)}
\)
of \emph{deterministic} local wirings
\(\{\Lambda^{(r)}\}_{r}\), each selected with probability \(p_r\) that
depends solely on the shared random seed~\(r\).
A \emph{deterministic DI local operation} is any single extremal element
\(\Lambda^{(r)}\); equivalently, it is a DI--LOSR wiring with no further
random choice once the seed is fixed (the special case of one seed). Importantly, applying such a deterministic LOSR wiring to a \emph{quantum-realizable} box always yields another quantum-realizable
box. The quantum realization of the second box will be different from the realization of the first box, but it will still be a quantum realization—possibly requiring different measurements or a different
underlying state.

We note also that the Magic‑Square box, $\mathcal{M}$, wins its game with probability 1,
whereas the Tsirelson box, $\mathcal{T}$, achieves the quantum‑maximal
CHSH value \(2\sqrt{2}\).

Assume for contradiction that a DI--LOSR wiring
\(\Lambda\) satisfies \(\Lambda(\mathcal{M})=\mathcal{T}\).
Several basic facts rule this out:

\begin{enumerate}
  \item[\textbf{(i)}] \emph{Uniqueness of the Tsirelson box.}
        Up to local relabellings of inputs and outputs, the Tsirelson
      box \(\mathcal{T}\) is the unique non-local quantum-realizable box  that
      attains the maximal CHSH value \(2\sqrt{2}\).

  \item[\textbf{(ii)}] \emph{Rationality of \(\mathcal{M}\).}
        The Magic Square box, \(\mathcal{M}\), has all probabilities given by rational numbers.

  \item[\textbf{(iii)}] \emph{Reduction to a deterministic operation.}
      Write the wiring as a convex combination
      \(\Lambda=\sum_{r} p_r\,\Lambda^{(r)}\) of deterministic
      device‑independent local operations \(\Lambda^{(r)}\) with
      weights \(p_r\ge 0\) and \(\sum_r p_r=1\).
      Because the CHSH score is a \emph{linear} function of the box, we know that:
      \[
          \text{(CHSH score of } \Lambda(\mathcal{M}) )
          \;=\;
          \sum_{r} p_r
          \bigl(\text{CHSH score of } \Lambda^{(r)}(\mathcal{M})\bigr).
      \]
      The left‑hand side is \(2\sqrt{2}\), the \emph{maximum} quantum
      CHSH value.  Since every term in the sum is bounded above by this
      maximum, each coefficient with \(p_r>0\) must itself achieve the
      value \(2\sqrt{2}\).
      By the uniqueness statement in (i), this forces every such
      deterministic wiring to output the Tsirelson box
      \(\mathcal{T}\).  In particular, there is at least one
      deterministic wiring, call it \(\Lambda_{\text{det}}\), with
      \(\Lambda_{\text{det}}(\mathcal{M})=\mathcal{T}\).

\item[\textbf{(iv)}] \emph{Deterministic wirings preserve
        rationality.}  A deterministic DI operation always maps rational probabilities to
        rational probabilities. Hence, by (ii),
        \(\Lambda_{\text{det}}(\mathcal{M})\) must have only rational
        entries.

  \item[\textbf{(v)}] \emph{Contradiction.}
        The Tsirelson box \(\mathcal{T}\) contains irrational
        probabilities (i.e.,\ terms involving \(1/\sqrt{2}\)).
        Step (iv), thus, contradicts
        \(\Lambda_{\text{det}}(\mathcal{M})=\mathcal{T}\).
\end{enumerate}

Therefore no DI--LOSR wiring can convert the Magic‑Square box into the
Tsirelson box.
\end{proof}

\end{document}